\documentclass[sigplan,10pt,screen]{acmart}
\usepackage[T1]{fontenc}

\usepackage{amsmath}
\usepackage{amssymb}
\usepackage{amsthm}
\usepackage{bussproofs}
\usepackage{caption}
\usepackage[inline]{enumitem}
\usepackage{footnote}
\usepackage{listings}
\usepackage{pdflscape}
\usepackage{pgfplots}
\usepackage{pifont}
\usepackage{stmaryrd}
\usepackage{subcaption}
\usepackage{tikz}
\usepackage{xfrac}
\usepackage{xspace}

\usepackage[capitalize]{cleveref}

\usepackage{rust}

\usepackage{todonotes}

\crefname{section}{\textbf{\S}}{\textbf{\S\S}}
\creflabelformat{section}{#2\textbf{#1}#3}

\crefname{subsection}{\textbf{\S}}{\textbf{\S\S}}
\creflabelformat{section}{#2\textbf{#1}#3}

\crefname{subsubsection}{\textbf{\S}}{\textbf{\S\S}}
\creflabelformat{section}{#2\textbf{#1}#3}

\crefname{line}{l.}{ls.}

\makesavenoteenv{tabular}
\makesavenoteenv{table*}

\def\ScoreOverhang{1pt}
\def\defaultHypSeparation{\hskip 5pt}

\definecolor{ruleColor}{rgb}{0.1, 0.3, 0.1}%
\definecolor{codepurple}{rgb}{.58,0,.82}
\definecolor{codegreen}{rgb}{0,.6,0}

\definecolor{primary}{HTML}{6C8EBF}
\definecolor{secondary}{HTML}{D6B656}
\definecolor{tertiary}{HTML}{82B366}
\definecolor{quaternary}{HTML}{666666}
\definecolor{quinary}{HTML}{B85450}
\definecolor{senary}{HTML}{9673A6}

\usetikzlibrary{arrows, arrows, arrows.meta, automata, backgrounds, calc, fit, matrix, positioning, shapes, spy}
\usepgfplotslibrary{groupplots}
\pgfplotsset{compat=1.16,
  /pgfplots/xbar legend/.style={
      /pgfplots/legend image code/.code={\draw[##1,/tikz/.cd,yshift=-.4em] (0cm,0cm) rectangle (0.8em,0.8em);},
    },
  /pgfplots/ybar legend/.style={
      /pgfplots/legend image code/.code={\draw[##1,/tikz/.cd,yshift=-.4em] (0cm,0cm) rectangle (0.8em,0.8em);},
    },
}

\setlist[description]{
  style=unboxed,
  topsep=0pt,
  leftmargin=0cm,
  listparindent=\parindent,wide
}

\lstdefinelanguage{Scribble}{
  sensitive,
  morecomment=[l]{//},
  morecomment=[s]{/*}{*/},
  morestring=[b]{"},
  morekeywords={module, type, as, global, protocol, role, rec, choice, at, or, continue, from, to}
}

\theoremstyle{plain}
\newtheorem{theorem}{Theorem}
\newtheorem{lemma}[theorem]{Lemma}

\theoremstyle{remark}
\newtheorem*{remark}{Remark}

\newcommand{\code}[1]{\texttt{#1}}

\newcommand{\async}{\code{async}\xspace}
\newcommand{\await}{\code{await}\xspace}

\newcommand{\amr}{\textsc{AMR}\xspace}
\newcommand{\mpst}{\textsc{MPST}\xspace}
\newcommand{\rumpsteak}{\textsc{Rumpsteak}\xspace}
\newcommand{\multicrusty}{\textsc{MultiCrusty}\xspace}
\newcommand{\ferrite}{\textsc{Ferrite}\xspace}
\newcommand{\sesh}{\textsc{Sesh}\xspace}
\newcommand{\rustfft}{\textsc{RustFFT}\xspace}
\newcommand{\kmc}{\(k\)-\textsc{MC}\xspace}
\newcommand{\concur}{\textsc{SoundBinary}\xspace}
\newcommand{\nuscr}{\textsc{\(\nu\)Scr}\xspace}
\newcommand{\repository}{\cite{Rumpsteak}\xspace}

\newcommand{\ppt}[1]{\code{#1}}

\newcommand{\tnat}{\code{nat}}
\newcommand{\tint}{\code{int}}

\newcommand{\gt}[0]{\mathsf{G}}
\newcommand{\gtmsg}[3]{\ppt{#1} \to \ppt{#2}: \left\{#3\right\}}
\newcommand{\gtmsgs}[2]{\ppt{#1} \to \ppt{#2}: \{}
\newcommand{\gtmsge}{\}}
\newcommand{\gtvar}[1]{\mathbf{#1}}
\newcommand{\gtrec}[1]{\mu\mathbf{#1}.}
\newcommand{\gtend}{\code{end}}
\newcommand{\gtproj}[2]{#1 \upharpoonright \ppt{#2}}

\newcommand{\lt}[0]{\mathsf{T}}
\newcommand{\tout}[1]{\ppt{#1}!}
\newcommand{\tin}[1]{\ppt{#1}?}
\newcommand{\tsel}[1]{\oplus_{#1}}
\newcommand{\tbra}[1]{\&_{#1}}
\newcommand{\tvar}[1]{\mathbf{#1}}
\newcommand{\trec}[2]{\mu\mathbf{#1}.#2}
\newcommand{\tend}{\mathtt{end}}

\newcommand{\efsm}{\mathsf{M}}
\newcommand{\oefsm}{\mathsf{M}'}
\newcommand{\api}{\mathsf{A}}
\newcommand{\proc}{\mathsf{P}}

\newcommand{\refa}[1]{\ensuremath{\mathcal{A}^{(\ppt{#1})}}}
\newcommand{\refb}[1]{\ensuremath{\mathcal{B}^{(\ppt{#1})}}}
\newcommand{\act}[1]{\code{act}(#1)}
\newcommand{\fterms}[1]{\code{terms}(#1)}
\newcommand{\tr}[1]{\mathcal{T}\left(#1\right)}

\newcommand{\nil}{\varnothing}
\newcommand{\pair}[2]{\langle#1\ \lceil\!\!\rfloor \ #2\rangle}
\newcommand{\triple}[3]{\langle#1,#2,#3\rangle}
\newcommand{\map}[1]{\left[#1\right]}

\newcommand{\si}[1]{\llbracket #1 \rrbracket_{\tiny\texttt{si}}}
\newcommand{\so}[1]{\llbracket #1 \rrbracket_{\tiny\texttt{so}}}

\newcommand{\ruleSub}[1]{\rulename{sub-#1}}

\newcommand{\ruleAlg}[1]{\rulename{#1}}

\newcommand{\veryshortrightarrow}[1][3pt]{\mathrel{%
    \vcenter{\hbox{\rule[-.5\fontdimen8\textfont3]{#1}{\fontdimen8\textfont3}}}%
    \mkern-4mu\hbox{\usefont{U}{lasy}{m}{n}\symbol{41}}}}
\newcommand{\reduce}{\shortrightarrow}

\newcommand{\ruleRed}[1]{\rulename{\(\veryshortrightarrow\)#1}}
\newcommand{\ruleRedA}{\ruleRed{$\mathcal{A}$}}
\newcommand{\ruleRedB}{\ruleRed{$\mathcal{B}$}}

\newcommand{\ruleRef}[1]{\rulename{ref-#1}}
\newcommand{\ruleRefA}{\ruleRef{$\mathcal{A}$}}
\newcommand{\ruleRefB}{\ruleRef{$\mathcal{B}$}}

\newcommand{\bnfas}{\mathrel{::=}}
\newcommand{\bnfalt}{\mathrel{\mid}}

\newcommand\feat[1]{\textbf{\scriptsize #1}}

\newcommand{\rulename}[1]{{\color{ruleColor}\textsc{\small [#1]}}}%

\newcommand{\tick}{{\tiny\ding{52}}}
\newcommand{\ftick}{{\color{tertiary}\tiny\ding{52}}}
\newcommand{\htick}{{\color{secondary}\tiny\ding{55}}}
\newcommand{\cross}{{\color{quinary}\tiny\ding{55}}}

\newcommand{\tsize}[1]{|#1|}
\newcommand{\oh}{\mathcal{O}}

\newtoggle{full}
\settoggle{full}{true}
\iftoggle{full}
{
  \newcommand{\inApp}[1]{the \cref{#1}}
}{
  \newcommand{\inApp}[1]{the full version~\cite{FullVersion}}
}

\usepackage{nopageno}
\AtBeginDocument{%
\providecommand\BibTeX{{%
\normalfont B\kern-0.5em{\scshape i\kern-0.25em b}\kern-0.8em\TeX}}}

\copyrightyear{2022}
\acmYear{2022}
\setcopyright{acmlicensed}\acmConference[PPoPP '22]{27th ACM SIGPLAN Symposium on Principles and Practice of Parallel Programming}{April 2--6, 2022}{Seoul, Republic of Korea}
\acmBooktitle{27th ACM SIGPLAN Symposium on Principles and Practice of Parallel Programming (PPoPP '22), April 2--6, 2022, Seoul, Republic of Korea}
\acmPrice{15.00}
\acmDOI{10.1145/3503221.3508404}
\acmISBN{978-1-4503-9204-4/22/02}

\ccsdesc[500]{Software and its engineering~Development frameworks and environments}
\ccsdesc[100]{Computer systems organization~Reliability}
\ccsdesc[500]{Software and its engineering~Source code generation}

\keywords{Rust, Asynchronous Message Passing, Message Reordering, Computation-Communication Overlap, Multiparty Session Types}

\makeatletter
\if@ACM@anonymous
    \renewcommand{\rumpsteak}{\textsc{Ranon}\xspace}
    \renewcommand{\repository}{\cite{Ranon}\xspace}
\fi
\makeatother

\begin{document}
\pagestyle{empty}
\title{Deadlock-Free Asynchronous Message Reordering in Rust with Multiparty Session Types}

\author{Zak Cutner}
%\email{zachary.cutner17@imperial.ac.uk}
\orcid{0000-0001-7180-4530}
\affiliation{%
    \institution{Imperial College London}
    \city{London}
    \country{UK}
}

\author{Nobuko Yoshida}
%\email{n.yoshida@imperial.ac.uk}
\orcid{0000-0002-3925-8557}
\affiliation{%
    \institution{Imperial College London}
    \city{London}
    \country{UK}
}

\author{Martin Vassor}
%\email{m.vassor@imperial.ac.uk}
\orcid{ 0000-0002-2057-0495 }
\affiliation{%
    \institution{Imperial College London}
    \city{London}
    \country{UK}
}


\begin{abstract}
    Rust is a modern systems language focused on performance and reliability.
    Complementing Rust's promise to provide "fearless concurrency", developers
    frequently exploit asynchronous message passing. Unfortunately, sending and
    receiving messages in an arbitrary order to maximise
    computation-communication overlap (a popular optimisation in message-passing
    applications) opens up a Pandora's box of subtle concurrency bugs.

    To guarantee deadlock-freedom by construction, we present \rumpsteak: a new
    Rust framework based on \emph{multiparty session types}. Previous session
    type implementations in Rust are either built upon synchronous and blocking
    communication and/or are limited to two-party interactions. Crucially, none
    support the arbitrary ordering of messages for efficiency.

    \rumpsteak instead targets asynchronous \code{async}/\code{await} code. Its
    unique ability is allowing developers to arbitrarily order send/receive
    messages while preserving deadlock-freedom. For this, \rumpsteak
    incorporates two recent advanced session type theories: (1) $k$-multiparty
    compatibility (\kmc), which \emph{globally} verifies the safety of a set of
    participants, and (2) asynchronous multiparty session subtyping, which
    \emph{locally} verifies optimisations in the context of a single
    participant. Specifically, we propose a novel algorithm for asynchronous
    subtyping that is both sound and decidable.

    We first evaluate the performance and expressiveness of \rumpsteak against
    three previous Rust implementations. We discover that \rumpsteak is around
    1.7--8.6x more efficient and can safely express many more examples by virtue
    of offering arbitrary ordering of messages. Secondly, we analyse the
    complexity of our new algorithm and benchmark it against \kmc and a
    \emph{binary} session subtyping algorithm. We find they are exponentially
    slower than \rumpsteak's.
\end{abstract}

\pagestyle{empty}
%\begin{CCSXML}
%  <ccs2012>
%  <concept_id>10010520.10010553.10010562</concept_id>
%  <concept_desc>Computer systems organization~Embedded systems</concept_desc>
%  <concept_significance>500</concept_significance>
%  <concept_id>10010520.10010575.10010755</concept_id>
%  <concept_desc>Computer systems organization~Redundancy</concept_desc>
%  <concept_significance>300</concept_significance>
%  <concept_id>10010520.10010553.10010554</concept_id>
%  <concept_desc>Computer systems organization~Robotics</concept_desc>
%  <concept_significance>100</concept_significance>
%  <concept_id>10003033.10003083.10003095</concept_id>
%  <concept_desc>Networks~Network reliability</concept_desc>
%  <concept_significance>100</concept_significance>
%  </ccs2012>
%\end{CCSXML}

%\ccsdesc[500]{Computer systems organization~Embedded systems}

%\keywords{datasets, neural networks, gaze detection, text tagging}

\maketitle
\thispagestyle{empty} 
\section{Introduction}
\label{sec:Introduction}

Rust is a statically-typed language designed for systems software development.
It is rapidly growing in popularity and has been voted ``most loved language''
over five years of surveys by Stack Overflow~\cite{StackOverflowSurvey}. Rust
aims to offer the safety of a high-level language without compromising on the
performance enjoyed by low-level languages. \emph{Message passing} over
\emph{typed channels} is common in concurrent Rust applications, where
(low-level) threads or (high-level) actors communicate efficiently and safely by
sending messages containing data.

To improve performance by maximising computation-communication
overlap~\cite{Choi2020,Sergent2018,Kim2012}, developers often wish to
arbitrarily change the order of sending and receiving messages---we will present
several examples of this technique, which we refer to as \emph{asynchronous
	message reordering} (\amr). Our challenge is to remedy communication errors such
as deadlocks, which can easily occur in message-passing applications,
particularly those that leverage \amr.

To achieve this, we introduce \rumpsteak: a framework for efficiently
coordinating message-passing processes in Rust using \emph{multiparty session
	types}. Session types~\cite{Honda1998,Takeuchi1994} (see \cite{Yoshida2020} for
a gentle introduction and~\cite{gay_behavioural_2017} for a more exhaustive one)
coordinate interactions through \emph{linearly typed channels}, that must be
used exactly once, ensuring \emph{protocol compliance} without deadlocks or
communication mismatches.

\paragraph{Current state of the art}
Since Rust's \emph{affine type system} is particularly well-suited to session
types by statically guaranteeing a linear usage of session channels, there are
several previous attempts at implementing session types in
Rust~\cite{Lagaillardie2020,Chen2021,Kokke2019,Jespersen2015}. However, their
current limitations prevent them from guaranteeing all four of
\emph{deadlock-freedom}, \emph{multiparty communication}, \emph{asynchronous
	execution} and \amr.

\paragraph{Our framework} We motivate the importance of each feature and explain
how \rumpsteak incorporates this.
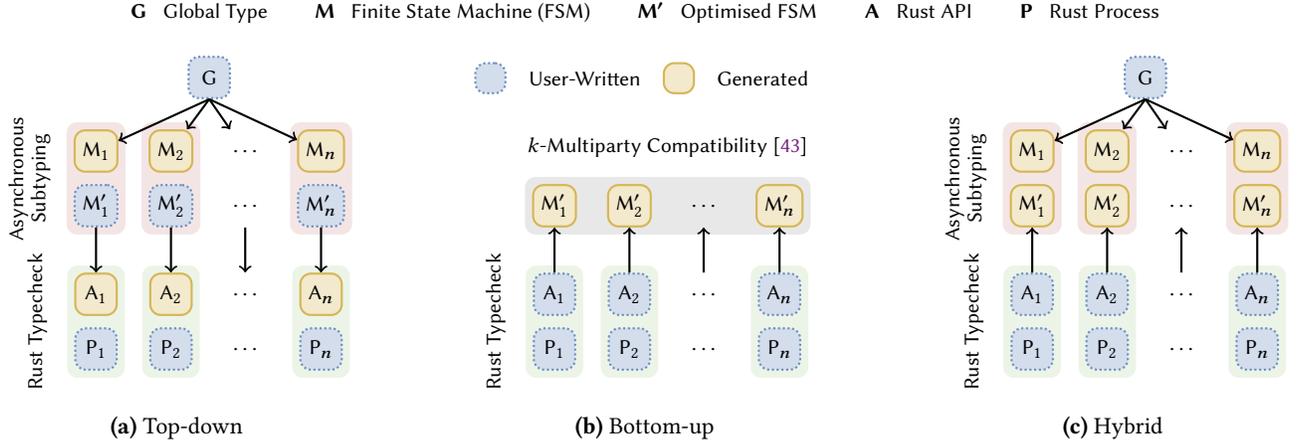
\begin{figure*}
	\centering
	\sffamily\footnotesize
	$\boldsymbol{\gt}$ \; Global Type \qquad
	$\boldsymbol{\efsm}$ \; Finite State Machine (FSM) \qquad
	$\boldsymbol{\oefsm}$ \; Optimised FSM \qquad
	$\boldsymbol{\api}$ \; Rust API \qquad
	$\boldsymbol{\proc}$ \; Rust Process \\[.4cm]
	\begin{subfigure}[b]{0.3\textwidth}
		\centering
		\begin{tikzpicture}[
				x=.3cm, y=-.95cm, align=center, thick,
				font=\footnotesize\linespread{0.8}\selectfont, draw,
				minimum size=.55cm,
				>=to, node distance=.4cm,
				hand/.style={fill=primary!30, draw=primary, densely dotted, rounded corners},
				generated/.style={fill=secondary!30, draw=secondary, rounded corners},
				compat/.style={fill=tertiary!15, draw=none, rounded corners},
				kmc/.style={fill=quaternary!15, draw=none, rounded corners},
				subtype/.style={fill=quinary!15, draw=none, rounded corners}
			]
			\node [hand, inner sep=1mm] (G) at (4,0) {$\gt$};
			\node [below=of G] (GB) {};
			\node [inner sep=1mm, right=-0.1cm of GB] (LL) {$\ldots$};
			\node [generated, inner sep=1mm, left=of LL] (L2) {$\efsm_2$};
			\node [generated, inner sep=1mm, left=of L2] (L1) {$\efsm_1$};
			\node [generated, inner sep=1mm, right=of LL] (Ln) {$\efsm_n$};
			\node [hand, inner sep=1mm, below=.15cm of L1] (M1) {$\oefsm_1$};
			\node [hand, inner sep=1mm, below=.15cm of L2] (M2) {$\oefsm_2$};
			\node [inner sep=1mm, below=.15cm of LL] (MM) {$\ldots$};
			\node [hand, inner sep=1mm, below=.15cm of Ln] (Mn) {$\oefsm_n$};
			\node [generated, inner sep=1mm, below=.6cm of M1] (A1) {$\api_1$};
			\node [generated, inner sep=1mm, below=.6cm of M2] (A2) {$\api_2$};
			\node [inner sep=1mm, below=.6cm of MM] (AA) {$\ldots$};
			\node [generated, inner sep=1mm, below=.6cm of Mn] (An) {$\api_n$};
			\node [hand, inner sep=1mm, below=.15cm of A1] (P1) {$\proc_1$};
			\node [hand, inner sep=1mm, below=.15cm of A2] (P2) {$\proc_2$};
			\node [inner sep=1mm, below=.15cm of AA] (PP) {$\ldots$};
			\node [hand, inner sep=1mm, below=.15cm of An] (Pn) {$\proc_n$};

			\begin{scope}[on background layer]
				\node [subtype, fit=(L1)(M1)] (sub1) {};
				\node [subtype, fit=(L2)(M2)] (sub2) {};
				\node [subtype, fit=(Ln)(Mn)] (subn) {};
			\end{scope}
			\node [left=1mm of sub1, anchor=south, rotate=90] {Asynchronous\\Subtyping};

			\begin{scope}[on background layer]
				\node [compat, fit=(A1)(P1)] (compat1) {};
				\node [compat, fit=(A2)(P2)] (compat2) {};
				\node [compat, fit=(An)(Pn)] (compatn) {};
			\end{scope}
			\node [left=1mm of compat1, anchor=south, rotate=90] {Rust Typecheck};

			\draw [->] (G.south) to (L1);
			\draw [->] (G.south) to (L2);
			\draw [->] (G.south) to (LL);
			\draw [->] (G.south) to (Ln);
			\draw [->] (M1) to (A1);
			\draw [->] (M2) to (A2);
			\draw [->] (MM) to (AA);
			\draw [->] (Mn) to (An);
		\end{tikzpicture}
		\caption{Top-down}
		\label{fig:topdownapproach}
	\end{subfigure}\begin{subfigure}[b]{0.4\textwidth}
		\centering
		\begin{tikzpicture}[
				x=.3cm, y=-.95cm, align=center, thick,
				font=\footnotesize\linespread{0.8}\selectfont, draw,
				minimum size=.4cm,
				>=to, node distance=.4cm,
				hand/.style={fill=primary!30, draw=primary, densely dotted, rounded corners=4},
				generated/.style={fill=secondary!30, draw=secondary, rounded corners=4},
				compat/.style={fill=tertiary!15, draw=none, rounded corners},
				kmc/.style={fill=quaternary!15, draw=none, rounded corners},
				subtype/.style={fill=quinary!15, draw=none, rounded corners}
			]
			\matrix[column sep=.2cm]{
				\node[hand]  {};  &
				\node {User-Written};  &
				\node[generated] {}; &
				\node {Generated};   \\
			};
		\end{tikzpicture} \\[.3cm]
		\begin{tikzpicture}[
				x=.3cm, y=-.95cm, align=center, thick,
				font=\footnotesize\linespread{0.8}\selectfont, draw,
				minimum size=.55cm,
				>=to, node distance=.4cm,
				hand/.style={fill=primary!30, draw=primary, densely dotted, rounded corners},
				generated/.style={fill=secondary!30, draw=secondary, rounded corners},
				compat/.style={fill=tertiary!15, draw=none, rounded corners},
				kmc/.style={fill=quaternary!15, draw=none, rounded corners},
				subtype/.style={fill=quinary!15, draw=none, rounded corners}
			]
			\node [generated, inner sep=1mm] (M1) {$\oefsm_1$};
			\node [generated, inner sep=1mm, right=of M1] (M2) {$\oefsm_2$};
			\node [inner sep=1mm, right=of M2] (MM) {$\ldots$};
			\node [generated, inner sep=1mm, right=of MM] (Mn) {$\oefsm_n$};
			\node [hand, inner sep=1mm, below=.6cm of M1] (A1) {$\api_1$};
			\node [hand, inner sep=1mm, below=.6cm of M2] (A2) {$\api_2$};
			\node [inner sep=1mm, below=.6cm of MM] (AA) {$\ldots$};
			\node [hand, inner sep=1mm, below=.6cm of Mn] (An) {$\api_n$};
			\node [hand, inner sep=1mm, below=.15cm of A1] (P1) {$\proc_1$};
			\node [hand, inner sep=1mm, below=.15cm of A2] (P2) {$\proc_2$};
			\node [inner sep=1mm, below=.15cm of AA] (PP) {$\ldots$};
			\node [hand, inner sep=1mm, below=.15cm of An] (Pn) {$\proc_n$};

			\begin{scope}[on background layer]
				\node [kmc, fit=(M1)(Mn)] (kMC) {};
			\end{scope}
			\node [above=1mm of kMC, anchor=south] {$k$-Multiparty Compatibility~\cite{Lange2019}};

			\begin{scope}[on background layer]
				\node [compat, fit=(A1)(P1)] (compat1) {};
				\node [compat, fit=(A2)(P2)] (compat2) {};
				\node [compat, fit=(An)(Pn)] (compatn) {};
			\end{scope}
			\node [left=1mm of compat1, anchor=south, rotate=90] {Rust Typecheck};

			\draw [<-] (M1) to (A1);
			\draw [<-] (M2) to (A2);
			\draw [<-] (MM) to (AA);
			\draw [<-] (Mn) to (An);
		\end{tikzpicture}
		\caption{Bottom-up}
		\label{fig:bottomup}
	\end{subfigure}\begin{subfigure}[b]{0.3\textwidth}
		\centering
		\begin{tikzpicture}[
				x=.3cm, y=-.95cm, align=center, thick,
				font=\footnotesize\linespread{0.8}\selectfont, draw,
				minimum size=.55cm,
				>=to, node distance=.4cm,
				hand/.style={fill=primary!30, draw=primary, densely dotted, rounded corners},
				generated/.style={fill=secondary!30, draw=secondary, rounded corners},
				compat/.style={fill=tertiary!15, draw=none, rounded corners},
				kmc/.style={fill=quaternary!15, draw=none, rounded corners},
				subtype/.style={fill=quinary!15, draw=none, rounded corners}
			]
			\node [hand, inner sep=1mm] (G) at (4,0) {$\gt$};
			\node [below=of G] (GB) {};
			\node [inner sep=1mm, right=-0.1cm of GB] (LL) {$\ldots$};
			\node [generated, inner sep=1mm, left=of LL] (L2) {$\efsm_2$};
			\node [generated, inner sep=1mm, left=of L2] (L1) {$\efsm_1$};
			\node [generated, inner sep=1mm, right=of LL] (Ln) {$\efsm_n$};
			\node [generated, inner sep=1mm, below=.15cm of L1] (M1) {$\oefsm_1$};
			\node [generated, inner sep=1mm, below=.15cm of L2] (M2) {$\oefsm_2$};
			\node [inner sep=1mm, below=.15cm of LL] (MM) {$\ldots$};
			\node [generated, inner sep=1mm, below=.15cm of Ln] (Mn) {$\oefsm_n$};
			\node [hand, inner sep=1mm, below=.6cm of M1] (A1) {$\api_1$};
			\node [hand, inner sep=1mm, below=.6cm of M2] (A2) {$\api_2$};
			\node [inner sep=1mm, below=.6cm of MM] (AA) {$\ldots$};
			\node [hand, inner sep=1mm, below=.6cm of Mn] (An) {$\api_n$};
			\node [hand, inner sep=1mm, below=.15cm of A1] (P1) {$\proc_1$};
			\node [hand, inner sep=1mm, below=.15cm of A2] (P2) {$\proc_2$};
			\node [inner sep=1mm, below=.15cm of AA] (PP) {$\ldots$};
			\node [hand, inner sep=1mm, below=.15cm of An] (Pn) {$\proc_n$};

			\begin{scope}[on background layer]
				\node [subtype, fit=(L1)(M1)] (sub1) {};
				\node [subtype, fit=(L2)(M2)] (sub2) {};
				\node [subtype, fit=(Ln)(Mn)] (subn) {};
			\end{scope}
			\node [left=1mm of sub1, anchor=south, rotate=90] {Asynchronous\\Subtyping};

			\begin{scope}[on background layer]
				\node [compat, fit=(A1)(P1)] (compat1) {};
				\node [compat, fit=(A2)(P2)] (compat2) {};
				\node [compat, fit=(An)(Pn)] (compatn) {};
			\end{scope}
			\node [left=1mm of compat1, anchor=south, rotate=90] {Rust Typecheck};

			\draw [->] (G.south) to (L1);
			\draw [->] (G.south) to (L2);
			\draw [->] (G.south) to (LL);
			\draw [->] (G.south) to (Ln);
			\draw [<-] (M1) to (A1);
			\draw [<-] (M2) to (A2);
			\draw [<-] (MM) to (AA);
			\draw [<-] (Mn) to (An);
		\end{tikzpicture}
		\caption{Hybrid}
		\label{fig:hybridapproach}
	\end{subfigure}
	\caption{Workflow of the \rumpsteak framework (three approaches).}
	\label{fig:approaches}
\end{figure*}
\begin{description}
	\item[Deadlock-freedom.] One of the most important properties of
		concurrent/parallel systems is that their computations are not blocked.
		Deadlock-freedom in our context states that the system can always either
		make progress by exchanging messages or properly terminate.

	\item[Multiparty communication.] Many previous Rust
		implementations~\cite{Chen2021,Kokke2019,Jespersen2015} support only
		\emph{binary} session types, which is limited to two-party
		communication. On the other hand, the majority of real-world
		communication protocols consist of more than two participants.
		\rumpsteak therefore uses \emph{multiparty} session types
		(\mpst)~\cite{Honda2016,Honda2008} to ensure deadlock-freedom in
		protocols with any number of participants (or roles).

	\item[Asynchronous execution.] Most previous Rust
		implementations~\cite{Lagaillardie2020,Kokke2019,Jespersen2015} use
		\emph{synchronous} communication channels. This approach suffers from
		performance limitations since threads are \emph{blocked} while waiting
		to receive messages. \rumpsteak instead uses \emph{asynchronous}
		communication, where lightweight asynchronous tasks share a pool of
		threads. When one task is blocked, another's work can be scheduled in
		the meantime to prevent the wasting of computational resources.

		Although \cite{Chen2021} is also based on asynchronous communication,
		only \rumpsteak closely integrates with Rust's modern \async/\await
		syntax, allowing asynchronous programs to be written sequentially. To
		achieve this, asynchronous functions are annotated with \async, causing
		them to return \emph{futures}. Developers can \await calls to these
		functions, denoting that execution should continue elsewhere until a
		result is ready.

	\item[Asynchronous message reordering.] Our main contribution is offering
		\amr, which no existing work can provide while preserving deadlock-freedom.
		To motivate this, we introduce a running example of the \emph{double
			buffering} protocol~\cite{Huang2002}. Buffering is frequently used in
		multimedia applications where a continuous stream of data must be sent from
		a source (e.g.\ a graphics card) to a sink (e.g.\ a CPU). To prevent the
		source from being blocked while the sink is busy, it writes to a buffer,
		which is later read by the sink.

		\begin{figure}[h]
			\centering
			\sffamily\footnotesize
			\begin{tikzpicture}[
					x=11.5pt, y=-10pt, thick,
					copy/.style={
							{Circle[width=3pt, length=3pt]}-to,
							shorten <=-1.5pt, shorten >=-1.5pt
						},
					time/.style={-To, every node/.style={fill=white}}
				]
				\node[anchor=south, rotate=90] at (0,1) {Source};
				\node[anchor=south, rotate=90] at (0,4) {Buffers};
				\node[anchor=south, rotate=90] at (0,7) {Sink};

				% Single buffer
				\node[anchor=south] at (5,0) {\textit{Single Buffer}};
				\path[time] (.5,8.5) edge node {\textbf{Time}} (9.5,8.5);

				\draw[draw=none, fill=secondary!30, rounded corners] (.5,.5) rectangle ++(9,1);
				\draw[draw=none, fill=senary!30, rounded corners] (.5,6.5) rectangle ++(9,1);

				\draw[primary, fill=primary!30, rounded corners] (1.5,2) rectangle ++(1,2);
				\path[copy] (2,1) edge node {} (2,3);

				\draw[primary, fill=primary!30, rounded corners] (3.5,2) rectangle ++(1,2);
				\path[copy] (4,3) edge node {} (4,7);

				\draw[primary, fill=primary!30, rounded corners] (5.5,2) rectangle ++(1,2);
				\path[copy] (6,1) edge node {} (6,3);

				\draw[primary, fill=primary!30, rounded corners] (7.5,2) rectangle ++(1,2);
				\path[copy] (8,3) edge node {} (8,7);

				% Double buffer
				\node[anchor=south] at (15,0) {\textit{Double Buffer}};
				\path[time] (10.5,8.5) edge node {\textbf{Time}} (19.5,8.5);

				\draw[draw=none, fill=secondary!30, rounded corners] (10.5,.5) rectangle ++(9,1);
				\draw[draw=none, fill=senary!30, rounded corners] (10.5,6.5) rectangle ++(9,1);

				\draw[primary, fill=primary!30, rounded corners] (11.5,2) rectangle ++(1,2);
				\draw[quinary, fill=quinary!30, rounded corners] (11.5,4) rectangle ++(1,2);
				\path[copy] (12,1) edge node {} (12,3);

				\draw[primary, fill=primary!30, rounded corners] (13.5,2) rectangle ++(1,2);
				\draw[quinary, fill=quinary!30, rounded corners] (13.5,4) rectangle ++(1,2);
				\path[copy, bend left=70] (14,1) edge node {} (14,5);
				\path[copy, bend right=70] (14,3) edge node {} (14,7);

				\draw[primary, fill=primary!30, rounded corners] (15.5,2) rectangle ++(1,2);
				\draw[quinary, fill=quinary!30, rounded corners] (15.5,4) rectangle ++(1,2);
				\path[copy] (16,1) edge node {} (16,3);
				\path[copy] (16,5) edge node {} (16,7);

				\draw[primary, fill=primary!30, rounded corners] (17.5,2) rectangle ++(1,2);
				\draw[quinary, fill=quinary!30, rounded corners] (17.5,4) rectangle ++(1,2);
				\path[copy, bend left=70] (18,1) edge node {} (18,5);
				\path[copy, bend right=70] (18,3) edge node {} (18,7);
			\end{tikzpicture}
			\caption{Illustration of the double buffering protocol.}
			\label{fig:DoubleBuffering}
		\end{figure}
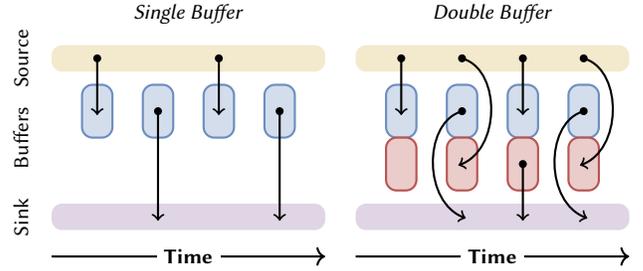

		We illustrate this effect in \cref{fig:DoubleBuffering}. With a single
		buffer, both the source and sink are constantly blocked as we cannot
		read and write simultaneously. However, adding a second buffer allows
		the source and sink to operate on different buffers at once so that (in
		the best case) they are never blocked and throughput can increase
		twofold.

		As we will see in \cref{sec:Overview}, the double buffering protocol
		takes advantage of \amr so it \emph{cannot} be expressed with standard
		\mpst theory~\cite{Honda2008}; communication with the source and sink
		are overlapped so both buffers can be accessed at once. The goal of this
		paper is to provide a method for ensuring that protocols optimised in
		this way preserve deadlock-freedom.
\end{description}

\paragraph{\rumpsteak framework} We achieve these features while allowing three
different approaches, summarised in \cref{fig:approaches}.

In the \textbf{top-down approach}, the developer writes a global type, which
describes the entire protocol (see \cref{sec:Overview} for an example). We
project this onto each participant to obtain a local finite state machine
(FSM), which describes the protocol from that participant's perspective.
Next, the developer uses \amr to propose optimised FSMs for each
participant. We \emph{locally} verify that each optimised FSM is compatible
with the projected one using an asynchronous subtyping algorithm. From each
optimised FSM, we generate an API so that the developer can write a Rust
process implementation. Our API uses Rust's type checker to ensure that
these processes conform to the protocol, and are therefore deadlock-free.

In the \textbf{bottom-up approach}, the developer manually writes an API and
process implementation for each participant. We derive the optimised FSMs
directly from these APIs and ensure that they are safe using
\emph{\(k\)-multiparty compatibility} (\kmc)~\cite{Lange2019}. \kmc globally
verifies that multiple FSMs are compatible with each other without using a
global type.

Finally, we propose a \textbf{hybrid approach} where the projected FSMs are
generated from a global type on one side (as in the top-down approach). On the
other side, the developer writes both the APIs and process implementations
directly, and we derive the optimised FSMs (as in the bottom-up approach). We
then \emph{locally} verify that the optimised FSMs are asynchronous subtypes of
the projected FSMs.

To summarise, our subtyping algorithm used in the top-down/hybrid approaches
\emph{locally} verifies the correctness of optimisations. This makes it far more
scalable than \kmc as we will see in \cref{sec:Evaluation}. These approaches are
also more intuitive to the developer since they use \emph{safety by
	construction}---it is much easier to write a global type than to determine why a
\kmc analysis has failed for a complex protocol.

\paragraph{Contribution and outline} In \cref{sec:Overview}, we give an overview
of the design and implementation of \rumpsteak. In \cref{sec:Theory}, we present
our new \emph{asynchronous subtyping algorithm} used to check that an optimised
FSM is correct. We prove our new algorithm is \emph{sound}
(\cref{thm:Soundness}) against the precise \mpst asynchronous subtyping
theory~\cite{Ghilezan2021}, \emph{terminates} (\cref{thm:Termination}) and
analyse its complexity (\cref{thm:Complexity}). In \cref{sec:Evaluation}, we
compare the performance and expressiveness of our framework with existing
work~\cite{Lagaillardie2020,Chen2021,Kokke2019,Lange2019,Bravetti2019}. We show
\begin{enumerate*}[label=\textbf{(\arabic*)}]
	\item \rumpsteak's runtime is faster and can express many more asynchronous
	protocols than other Rust
	implementations~\cite{Lagaillardie2020,Kokke2019,Chen2021}; and
	\item our subtyping algorithm is more efficient than existing
	algorithms~\cite{Lange2019,Bravetti2019}, confirming our complexity
	analysis.
\end{enumerate*}
\cref{sec:Related} discusses related work and concludes. The
\textbf{supplementary materials} available in the full
version~\cite{FullVersion} contain further examples and proofs of the theorems.
We include our source code and benchmarks in a \textbf{public GitHub
	repository}~\repository.

\begin{figure*}
  \begin{subfigure}[b]{0.3\textwidth}
    \centering
    \begin{lstlisting}[language=Scribble]
global protocol Ring(role s,
    role t) {
  rec loop {
    ready() from t to s;
    choice at s {
      value() from s to t;
      continue loop;
    } or {
      stop() from s to t;
    }
  }
}
\end{lstlisting}
    \caption{Scribble protocol description}
    \label{fig:example_scribble}
  \end{subfigure}\begin{subfigure}[b]{0.35\textwidth}
    \centering
    \begin{tikzpicture}[>=to, -to, thick, font=\footnotesize, initial left, initial distance=1.5ex, initial text={}, state/.style={draw, circle, node distance=1cm and 1.6cm, inner sep=3pt}, every path/.style={every node/.style={fill=white}}]
      \node[state,initial] (1) {};
      \node[state] (2) [right=of 1] {};
      \node[state] (3) [below=of 2] {};

      \path (1) edge[bend left] node {$\tin{t}\mathit{ready}$} (2);
      \path (2) edge[bend left] node {$\tout{t}\mathit{value}$} (1);
      \path (2) edge node {$\tout{t}\mathit{stop}$} (3);
    \end{tikzpicture}

    \vspace{.4cm}
    \begin{tikzpicture}[>=to, -to, thick, font=\footnotesize, initial left, initial distance=1.5ex, initial text={}, state/.style={draw, circle, node distance=1cm and 1.6cm, inner sep=3pt}, every path/.style={every node/.style={fill=white}}]
      \node[state,initial] (1) {};
      \node[state] (2) [right=of 1] {};
      \node[state] (3) [below=of 2] {};

      \path (1) edge[bend left] node {$\tout{s}\mathit{ready}$} (2);
      \path (2) edge[bend left] node {$\tin{s}\mathit{value}$} (1);
      \path (2) edge node {$\tin{s}\mathit{stop}$} (3);
    \end{tikzpicture}
    \caption{FSMs for the source and sink}
    \label{fig:example_automata}
  \end{subfigure}\begin{subfigure}[b]{0.3\textwidth}
    \centering
    \begin{lstlisting}[language=Rust]
type Source = Receive<T, Ready,
    Select<T, SourceChoice>>;

enum SourceChoice {
  Value(Value, Source),
  Stop(Stop, End) }

type Sink = Send<S, Ready,
    Branch<S, SinkChoice>>;

enum SinkChoice {
  Value(Value, Sink),
  Stop(Stop, End) }
\end{lstlisting}
    \caption{Rust API code (simplified excerpt)}
    \label{fig:example_rust}
  \end{subfigure}
  \caption{Different session type representations used within \rumpsteak.}
  \label{fig:rumpsteak_files}
\end{figure*}

\section{Design and Implementation}
\label{sec:Overview}

This section presents \rumpsteak and explains its three workflows: a
\emph{top-down approach} using asynchronous subtyping, a \emph{bottom-up
  approach} using $k$-multiparty compatibility (\kmc)~\cite{Lange2019} and a
\emph{hybrid approach} combining the two.

\subsection{Top-Down Approach}
\label{sec:TopDown}

All three approaches result in the same final application, so we use the example
of the top-down workflow (\cref{fig:topdownapproach}) to give a  general
overview of \rumpsteak's implementation.

\paragraph{Two-party protocol}
Before going into the details of the top-down approach, we first give an
informal insight into how the protocol is represented at each stage, shown in
\cref{fig:rumpsteak_files}. For this, we use the example of a simple streaming
protocol, later introduced in~\cref{sec:Evaluation}.

\begin{description}
  \item[A global session type] describes the protocol from a global perspective
    and includes all participants in the protocol. The global type for the
    streaming protocol (shown below) is read as follows: participant \(\ppt{t}\)
    first sends \(\mathit{ready}\) to participant \(\ppt{s}\). Next, \(\ppt{s}\)
    replies to \(\ppt{t}\) with two possible messages: either \(\mathit{value}\)
    or \(\mathit{stop}\). In the latter case, the protocol terminates (type
    \(\gtend\)). In the former case, the protocol continues with \(\gtvar{x}\),
    which is a type variable that is bound by \(\gtrec{x}\), i.e. the protocol
    starts over.
    \begin{equation*}
      \gt_\textsf{ST} = \gtrec{x}\gtmsg{t}{s}{\mathit{ready}.\gtmsg{s}{t}{
          \mathit{value}.\gtvar{x}, \quad
          \mathit{stop}.\gtend
        }}
    \end{equation*}
    In \rumpsteak, developers express global types syntactically with
    Scribble~\cite{ScribbleWebsite,Yoshida2013}: a widely used and
    target-agnostic language for describing multiparty protocols. We show the
    corresponding Scribble description for the streaming protocol
    in~\cref{fig:example_scribble}.

  \item[A local finite state machine] describes the protocol from the
    perspective of a single participant. It shows the send and receive actions
    of that participant, independent of what other participants may be doing in
    the meantime.

    In \cref{fig:example_automata}, we show the local FSMs for each participant
    in the streaming protocol, where $!$ and $?$ denote \emph{send} and
    \emph{receive} respectively in session type syntax~\cite{Yoshida2020}. For
    example, the FSM for the source first receives \(\mathit{ready}\) from the
    sink, then chooses to either send \(\mathit{value}\) and start over or
    simply send \(\mathit{stop}\) and finish.

  \item[The Rust API] is an encoding of a local FSM as Rust code. This code then
    uses Rust's type checker to confirm that a process implementation, also
    written in Rust, conforms to the FSM. For example, the encodings for the
    source and sink are shown in \cref{fig:example_rust}.
\end{description}

\paragraph{Multiparty protocol}
In the remainder of this section, we use a more complex example (the double
buffering protocol from \cref{sec:Introduction}), which allows us to illustrate
how \rumpsteak can verify multiparty protocols. We carefully go through each
step of the top-down approach, closely following \cref{fig:topdownapproach}.

To guarantee safety in our double buffering example using \mpst, the developer
defines a global type $\gt_\textsf{DB}$ for the protocol. This protocol has
three participants: a source \ppt{s}, a kernel \ppt{k} (which controls both of
the buffers) and a sink \ppt{t}. We also show the corresponding Scribble
description for the double buffering protocol in~\cref{lst:Scribble}.
\begin{multline*}
  \gt_\textsf{DB} = \gtrec{x}\gtmsgs{k}{s}\mathit{ready}.\gtmsgs{s}{k} \\
  \mathit{value}.\gtmsgs{t}{k}\mathit{ready}.\gtmsgs{k}{t}\mathit{value}.\gtvar{x}\gtmsge\gtmsge\gtmsge\gtmsge
\end{multline*}

\begin{lstlisting}[language=Scribble, label={lst:Scribble}, caption={Scribble representation of the double buffering protocol.}, float=t]
global protocol DoubleBuffering(role s,
    role k, role t) {
  rec loop {
    ready() from k to s;
    value() from s to k;
    ready() from t to k;
    value() from k to t;
    continue loop;
  }
}
\end{lstlisting}

\paragraph{Projection} Once we have created a global type, we use a projection
algorithm to derive the local FSM for each participant. In \rumpsteak, we
perform projection using \nuscr~\cite{NuScr}: a new lightweight and extensible
Scribble toolchain implemented in OCaml.

We show the projected local FSM for each of the participants in
\cref{fig:DoubleBufferingFSMs}. For instance, the kernel, whose projected type
$\efsm_\ppt{k}$ is shown in \cref{fig:DoubleBufferingEFSM},
\begin{enumerate*}[label=\textbf{(\arabic*)}, ref={\arabic*}]
  \item\label{itm:DoubleBufferingStart} informs the source that it is ready to
  receive;
  \item receives a value from the source;
  \item waits for the sink to become ready;
  \item sends a value to the sink; and
  \item repeats from step~\ref{itm:DoubleBufferingStart}.
\end{enumerate*}

\begin{figure}
  \centering
  \begin{subfigure}{.4\linewidth}
    \centering
    \begin{tikzpicture}[>=to, -to, thick, font=\footnotesize, initial left, initial distance=1.5ex, initial text={}, state/.style={draw, circle, node distance=1cm and 1.6cm, inner sep=3pt}, every path/.style={every node/.style={fill=white}}]
      \node[state,initial] (1) {};
      \node[state] (2) [right=of 1] {};
      \node[state] (3) [below=of 2] {};
      \node[state] (4) [left=of 3] {};

      \path (1) edge node {$\tout{s}\mathit{ready}$} (2);
      \path (2) edge node {$\tin{s}\mathit{value}$} (3);
      \path (3) edge node {$\tin{t}\mathit{ready}$} (4);
      \path (4) edge node {$\tout{t}\mathit{value}$} (1);
    \end{tikzpicture}
    \caption{Projected FSM for \(\ppt{k}\) ($\efsm_\ppt{k}$)}
    \label{fig:DoubleBufferingEFSM}
  \end{subfigure}\begin{subfigure}{.6\linewidth}
    \centering
    \begin{tikzpicture}[>=to, -to, thick, font=\footnotesize, initial left, initial distance=1.5ex, initial text={}, state/.style={draw, circle, node distance=1cm and 1.6cm, inner sep=3pt}, every path/.style={every node/.style={fill=white}}]
      \node[state,initial] (1) {};
      \node[state] (2) [right=of 1] {};
      \node[state] (3) [right=of 2] {};
      \node[state] (4) [below=of 3] {};
      \node[state] (5) [left=of 4] {};

      \path (1) edge node {$\tout{s}\mathit{ready}$} (2);
      \path (2) edge node {$\tout{s}\mathit{ready}$} (3);
      \path (3) edge node {$\tin{s}\mathit{value}$} (4);
      \path (4) edge node {$\tin{t}\mathit{ready}$} (5);
      \path (5) edge node {$\tout{t}\mathit{value}$} (2);
    \end{tikzpicture}
    \caption{Optimised FSM for \(\ppt{k}\) ($\oefsm_\ppt{k}$)}
    \label{fig:DoubleBufferingOFSM}
  \end{subfigure}

  \vspace{.4cm}
  \begin{subfigure}{.5\linewidth}
    \centering
    \begin{tikzpicture}[>=to, -to, thick, font=\footnotesize, initial left, initial distance=1.5ex, initial text={}, state/.style={draw, circle, node distance=1cm and 1.6cm, inner sep=3pt}, every path/.style={every node/.style={fill=white}}]
      \node[state,initial] (1) {};
      \node[state] (2) [right=of 1] {};

      \path (1) edge[bend left] node {$\tin{k}\mathit{ready}$} (2);
      \path (2) edge[bend left] node {$\tout{k}\mathit{value}$} (1);
    \end{tikzpicture}
    \caption{Projected FSM for \(\ppt{s}\) (\(\efsm_\ppt{s}\))}
  \end{subfigure}\begin{subfigure}{.5\linewidth}
    \centering
    \begin{tikzpicture}[>=to, -to, thick, font=\footnotesize, initial left, initial distance=1.5ex, initial text={}, state/.style={draw, circle, node distance=1cm and 1.6cm, inner sep=3pt}, every path/.style={every node/.style={fill=white}}]
      \node[state,initial] (1) {};
      \node[state] (2) [right=of 1] {};

      \path (1) edge[bend left] node {$\tout{k}\mathit{ready}$} (2);
      \path (2) edge[bend left] node {$\tin{k}\mathit{value}$} (1);
    \end{tikzpicture}
    \caption{Projected FSM for \(\ppt{t}\) (\(\efsm_\ppt{t}\))}
  \end{subfigure}
  \caption{FSMs for roles \(\ppt{k}\), \(\ppt{t}\) and \(\ppt{s}\) of the double buffering protocol.}
  \label{fig:DoubleBufferingFSMs}
\end{figure}
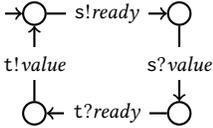
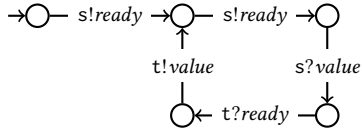
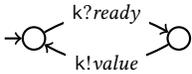
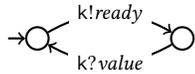

\paragraph{Asynchronous message reordering} Unfortunately, global types cannot
represent overlapping communication. Therefore, the projection $\efsm_\ppt{k}$
cannot achieve the simultaneous interactions we saw in
\cref{fig:DoubleBuffering}. In practice, developers use \amr to produce an
optimised FSM for the kernel $\oefsm_\ppt{k}$ (shown in
\cref{fig:DoubleBufferingOFSM}). Here, the kernel initially sends two
$\mathit{ready}$ messages to the source. This allows the source to write to the
second buffer while the sink is busy reading from the first. Note that the
asynchronous queue is effectively acting as the second buffer---the second
message from the source waits in the queue while the kernel is processing the
first.

Crucially, \rumpsteak verifies that $\oefsm_\ppt{k}$ is safe to use in the place
of $\efsm_\ppt{k}$ without causing deadlocks or other communication errors.
Otherwise, we would be throwing away the safety benefits that come with \mpst.
To achieve this, we must determine that $\oefsm_\ppt{k}$ is an
\emph{asynchronous subtype} of $\efsm_\ppt{k}$, which we perform using our
algorithm presented in \cref{sec:Theory}.

\paragraph{API design} Once we arrive at an optimised FSM $\oefsm_i$ for each
participant, our challenge is to create a Rust API $\api_i$, which uses Rust's
type checker to ensure that a developer-written process $\proc_i$ conforms to
$\oefsm_i$ (a relevant introduction to the Rust programming language can be found in
Jung's thesis~\cite[Chapter~8]{Jung_2020}). In the top-down approach, this API
is automatically generated from an optimised FSM.

To illustrate, we show the API for the kernel ($\api_\ppt{k}$) in
\cref{lst:API}, which checks conformance to $\oefsm_\ppt{k}$. To ensure that our
API remains readable by developers and to eliminate extensive boilerplate code,
we make use of Rust's procedural macros \cite{RustProceduralMacros}. By
decorating types with \code{\#[...]}, these macros perform additional
compile-time code generation.

\begin{description}
  \item[Roles (Participants).] Each role is represented as a struct, which
    stores its communication channels with other roles. The struct for the
    kernel (\crefrange{line:RoleStart}{line:RoleEnd}) contains channels to and
    from \ppt{s} and \ppt{t}. Developers can in fact use any custom channel that
    implements Rust's standard \code{Sink} or \code{Stream}
    interfaces~\cite{Futures} for asynchronous sends and receives respectively.
    This approach minimises the expensive creation of channels in cases where
    only bounded or unidirectional channels are required.

    Our \code{\#[derive(Role)]} macro generates methods for programmatically
    retrieving these channels from the struct. Moreover, an optional
    \code{\#[derive(Roles)]} macro can be applied on a struct containing every
    role to also generate code for automatically instantiating all roles at
    once. This will create the necessary combination of channels and assign them
    to the correct structs in order to reduce human error.

  \item[Sessions.] Following the approach of
    \multicrusty~\cite{Lagaillardie2020}, we build a set of \emph{generic
      primitives} to construct a simple API. For instance, the \code{Send}
    primitive (\cref{line:SendPrimitive}) takes a role, label and continuation
    as generic parameters. In contrast to the standard approach of creating a
    type for every state \cite{Hu2016}, these primitives reduce the number of
    types required and avoid arbitrary naming of these `state types'. For
    brevity, our API elides two additional parameters used to store channels at
    runtime, which are reinserted with the \code{\#[session]} macro.

    \begin{lstlisting}[language=Rust, label={lst:API}, caption={\rumpsteak API for the kernel.}, float=t]
#[derive(Role)]@\label{line:RoleStart}@
#[message(Label)]
struct K(
  #[route(S)] Channel,
  #[route(T)] Channel
);@\label{line:RoleEnd}@

#[derive(Message)]
enum Label {@\label{line:Message}@
  Ready(Ready),
  Value(Value),
}

struct Ready;@\label{line:ReadyLabel}@
struct Value(i32);@\label{line:CopyLabel}@

#[session]
type Kernel = Send<S, Ready, KernelLoop>;

#[session]
struct KernelLoop(Send<S, Ready, Receive<S, Value,@\label{line:SendPrimitive}@
    Receive<T, Ready, Send<T, Value, Self>>>>);
\end{lstlisting}

  \item[Labels.] Internally, \rumpsteak sends a \code{Label} enum
    (\cref{line:Message}) over reusable channels to communicate with other
    participants. Each label is represented as a type
    (\crefrange{line:ReadyLabel}{line:CopyLabel}) and our
    \code{\#[derive(Message)]} macro generates methods for converting to and
    from the \code{Label} enum.

  \item[Choice.] The kernel's API does not contain choice but we show an example
    of this below. Choice is represented as an enum where each branch contains
    the label sent/received and a continuation. This design allows processes to
    pattern match on external choices to determine which label was received.
    Methods allowing the enum to be used with \code{Branch} (an \emph{external}
    choice, \cref{line:BranchPrimitive}) or \code{Select} (an \emph{internal}
    choice) are also generated with the \code{\#[session]} macro. Even so, the
    enum is still necessary since Rust's lack of variadic generics means choice
    cannot be easily implemented as its own primitive.
    \begin{lstlisting}[language=Rust]
#[session]
enum Choice {
  Continue(Continue, Branch<A, Self>),@\label{line:BranchPrimitive}@
  Stop(Stop, End),
}
    \end{lstlisting}
\end{description}

\paragraph{Process implementation} Using the API $\api_\ppt{k}$, the developer
writes an implementation for the process $\proc_\ppt{k}$; we show an example in
\cref{lst:Process}. We discuss how our API uses Rust's type system to check that
$\proc_\ppt{k}$ conforms to the protocol.
\begin{lstlisting}[language=Rust, caption={Process implementation for the kernel.}, label={lst:Process}, float=t]
async fn kernel(role: &mut K) -> Result<()> {
  try_session(role, |s: Kernel<'_, _>| async {@\label{line:TrySession}@
    let mut t = s.send(Ready).await?;
    loop {@\label{line:LoopBottom}@
      let s = t.into_session().send(Ready).await?;
      let (Value(value), s) = s.receive().await?;
      let (Ready, s) = s.receive().await?;
      t = s.send(Value(value)).await?;
    }
  }).await
}
\end{lstlisting}
\begin{description}
  \item[Linearity.] The linear usage of channels is checked by Rust's
    \emph{affine type system}, preventing channels from being used multiple
    times. When a primitive is executed, it consumes itself, preventing reuse,
    and returns its continuation.

    Developers are prevented from constructing primitives directly using
    visibility modifiers. They must instead use \code{try\_session}
    (\cref{line:TrySession}). This function accepts a reference to the role
    being implemented and a closure (or function pointer) to its process
    implementation. The closure takes the input session type as an argument and
    returns the terminal type \code{End}. If a session is discarded, thereby
    breaking linearity, then the developer will have no \code{End} to return and
    Rust's compiler will complain that the closure does not satisfy its type.

  \item[Infinite recursion.] At first glance, it seems impossible to implement
    processes with infinitely recursive types using this approach. Nevertheless,
    for types without an \code{End} primitive, such as \code{Kernel}, we can use
    an infinite loop (\cref{line:LoopBottom}) to get around this problem.
    Conveniently, infinite loops are assigned \code{!}, Rust's never (or bottom)
    type, which can be implicitly cast to any other type. This allows the type
    of the loop to be coerced to \code{End}, enabling the closure to pass the
    type checker as before.

  \item[Channel reuse.] We allow roles to be reused across sessions since the
    channels they contain are usually expensive to create. Crucially, to prevent
    communication mismatches between different sessions, we must ensure that the
    same role is not used in multiple sessions at once. Therefore,
    \code{try\_session} takes an \emph{exclusive} reference to the role, causing
    Rust's borrow checker to enforce this requirement.
\end{description}

\subsection{Bottom-Up Approach}
\label{subsec:bottomup}
In the top-down approach, we generate a Rust API $\api_i$ from an optimised FSM
$\oefsm_i$. In the bottom-up approach (\cref{fig:bottomup}), we do the reverse:
we \textbf{\emph{serialise}} each API $\api_i$ to obtain an FSM $\oefsm_i$.
Next, we use \kmc on the set of FSMs $\oefsm_{1 \ldots n}$. If they are indeed
compatible, then the processes $\proc_{1 \ldots n}$, which implement their
respective APIs, are free from deadlocks.

\kmc takes the optimised FSMs of all participants and verifies deadlock freedom.
In contrast, asynchronous subtyping checks the optimisation of a single
participant's FSM in isolation. Therefore, \kmc can be seen as a global analysis
of the protocol and asynchronous subtyping as a local analysis of a single
participant.

To perform the serialisation of an API to an FSM we provide a Rust function
\code{serialize<S>() -> Fsm} (this is a simplified version). It takes a session
type API (such as \code{Kernel} from \cref{sec:TopDown}) as a generic type
parameter \code{S} and returns its corresponding FSM. This FSM can be printed in
a variety of formats and passed into the \kmc tool for verification.

\subsection{Hybrid Approach}
\label{subsec:hybrid}
Developers may naturally prefer the bottom-up approach since code generation, as
used in the top-down approach, can be quite opaque and difficult to understand
or debug. Nevertheless, the top-down approach has the advantage of using
asynchronous subtyping rather than \kmc---analysing the local types for all
participants in the protocol at once is challenging to do scalably (see
\cref{sec:Evaluation}).

Moreover, when a \kmc analysis fails, it can be difficult to determine how a
developer should update a complex protocol to make it free from deadlocks.
Safety by construction, as used in the top-down approach, is easier to work with
since verification is done locally on each participant.

Therefore, we also propose a third, hybrid approach (\cref{fig:hybridapproach}).
In this workflow, a global type $\gt$ is provided by the developer and projected
to obtain the FSMs $\efsm_{1 \ldots n}$ as before. Rather than the developer
proposing the optimised FSMs $\oefsm_{1 \ldots n}$ directly, they instead simply
write the APIs $\api_{1 \ldots n}$ (as in the bottom-up approach). These are
serialised to $\oefsm_{1 \ldots n}$ which can (as in the top-down approach) be
checked for safety against $\efsm_{1 \ldots n}$ using asynchronous subtyping. In
essence, the hybrid approach uses the same theory as the top-down approach, but
presents a more developer-friendly interface that uses serialisation rather than
code generation.

\section{A Sound Asynchronous Multiparty Session Subtyping Algorithm}
\label{sec:Theory}

This section proposes an algorithm for asynchronous multiparty subtyping
(\cref{sec:SoundAlgorithm}), shows its soundness (\cref{thm:Soundness}),
termination (\cref{thm:Termination}) and complexity (\cref{thm:Complexity}), and
sketches its implementation.

We begin by formally defining global and local session types. In the remainder
of this section, we will omit sorts $S$ from the syntax to simplify the
presentation.

\begin{definition}[Global and local types]
  \label{def:syntax}
  \[
    \begin{array}{rclr}
      S   & \bnfas & \code{i32} \bnfalt \code{u32} \bnfalt \code{i64} \bnfalt \code{u64} \bnfalt \ldots \\
      \gt & \bnfas & \gtend  \bnfalt  \gtmsg{p}{q}{\ell_i(S_i).\gt_i}_{i
        \in I}
      \bnfalt \gtrec{t}\gt
      \bnfalt  \gtvar{t}                                                                                \\
      \lt & \bnfas & \tend
      \bnfalt
      \tsel{i \in I} \tout{p}\ell_i(S_i).\lt_i
      \bnfalt  \tbra{i \in I} \tin{p}\ell_i(S_i).\lt_i
      \bnfalt  \trec{t}{\lt}
      \bnfalt  \tvar{t}
    \end{array}
  \]
  $\tsel{i \in I} \tout{p}\ell_i(S_i).\lt_i$ and $\tbra{i \in I}
    \tin{p}\ell_i(S_i).\lt_i$ represent \emph{internal} and \emph{external}
  choices respectively where $!$ and $?$ denote send and receive respectively,
  and all $\ell_i$ are pairwise distinct.
\end{definition}

Defining sound message reordering is non-trivial since it may introduce
deadlocks, as shown in the example below.

\begin{example}[Correct/incorrect \amr]
  \label{ex:amr}
  Consider the following local types
  \begin{equation*}
    \lt_Q = \tin{p}\ell_1.\tout{p}\ell_2.\tend \quad
    \lt_P = \tout{q}\ell_1.\tin{q}\ell_2.\tend
  \end{equation*}
  which are given by projecting a global type
  \begin{equation*}
    \gtmsg{p}{q}{\ell_1.\gtmsg{q}{p}{\ell_2.\tend}}
  \end{equation*}

  Reordering the actions of \(\ppt{q}\) to become \(\lt_Q^\prime =
  \tout{p}\ell_2.\tin{p}\ell_1.\tend\), first sending before receiving, retains
  deadlock-freedom as messages are stored in queues and their reception can be
  delayed. However, if instead we reorder \(\ppt{p}\)'s interactions to become
  \(\lt_P^\prime = \tin{q}\ell_2.\tout{q}\ell_1.\tend\), we arrive at a deadlock
  since both processes are simultaneously expecting to receive a message that
  has not yet been sent.
\end{example}

\subsection{Precise Asynchronous Multiparty Session Subtyping}
\label{sec:PreciseSubtyping}

\emph{Multiparty session subtyping} formulates \emph{refinement} for
communication protocols with more than two communicating processes. A process
implementing a session type $\lt$ can be safely used whenever a process
implementing one of its supertypes $\lt'$ is expected. This applies in any
context, ensuring that no deadlocks or other communication errors will be
introduced. Replacing the implementation of $\lt'$ with that of the subtype
$\lt$ may allow for more \emph{optimised} communication patterns as we have
seen. \citet{Ghilezan2021} present the \emph{precise} asynchronous subtyping
relation $\leq$ for multiparty session processes. Precision is characterised by
both \emph{soundness} (safe process replacement is guaranteed) and
\emph{completeness} (any extension of the relation is unsound).

Crucially, asynchronous subtyping supports the optimisation of communications.
Under certain conditions, the subtype can anticipate some input/output actions
occurring in the supertype, performing them \emph{earlier} than prescribed to
achieve the most flexible and precise subtyping. Such reorderings can take two
forms:
\begin{enumerate}[label=\textbf{R\arabic*.},ref=\textbf{R\arabic*}]
  \item\label{itm:ReorderA} anticipating an input from participant $\ppt{p}$
  before a finite number of inputs that are not from $\ppt{p}$; or
  \item\label{itm:ReorderB} anticipating an output to participant $\ppt{p}$
  before a finite number of inputs (from any participant), and also before other
  outputs that are not to $\ppt{p}$.
\end{enumerate}
To denote such reorderings, two kinds of finite sequences of inputs/outputs are
defined by \cite{Ghilezan2021}, where $\ppt{p} \neq \ppt{q}$. $\refa{p}$ is a
sequence containing receives from participants apart from $\ppt{p}$; $\refb{p}$
is a sequence containing receives from \emph{any} participant (\(\ppt{r} =
\ppt{p}\) is allowed) and sends to participants apart from $\ppt{p}$.
\begingroup
\small
\begin{gather*}
  \refa{p}\! ::= \!\tin{q}\ell \bnfalt \tin{q}\ell.\refa{p} \quad
  \refb{p} \! ::= \! \enskip \tin{r}\ell \bnfalt \tout{q}\ell \bnfalt \tin{r}\ell.\refb{p} \bnfalt \tout{q}\ell.\refb{p}
\end{gather*}
\endgroup
The \emph{tree refinement relation} $\lesssim$ is defined coinductively on
infinite session type trees that contain only single-inputs (SI) and
single-outputs (SO). We leave the formal definition of $\lesssim$ in
\cite{Ghilezan2021} and shall explain its essence with our sound
algorithm. Based on $\lesssim$, the subtyping relation $\leq$ for all types
(including internal and external choice) is given as
\vspace{-.2cm}
\begin{prooftree}
  \small
  \def\defaultHypSeparation{\hskip 2.5pt}
  \AxiomC{$\forall U \in \so{T}$}
  \AxiomC{$\forall V' \in \si{T'}$}
  \AxiomC{$\exists W \in \si{U}$}
  \AxiomC{$\exists W' \in \so{V'}$}
  \AxiomC{$W \lesssim W'$}
  \QuinaryInfC{$\lt \leq \lt'$}
\end{prooftree}
where $\so{T}$ (resp. $\si{T}$) is the minimal set of trees containing only
single \emph{outputs} (resp. \emph{inputs}) of the session type tree $T$. Using
existential quantifiers for $\si{U}$ and $\so{V'}$ allows external choices to be
added and internal choices to be removed (see \inApp{subsec:subtypeex} for
examples of $\leq$).

\subsection{Our Algorithm}
\label{sec:SoundAlgorithm}

The precise asynchronous subtyping in \cite{Ghilezan2021} is
\emph{undecidable}---even when limited to two participants, as proven
in~\cite{Lange2017}. This subsection introduces our \emph{practical} algorithm
which is \emph{sound} and \emph{terminates}, through the use of a bound on
recursion.

\paragraph{Prefixes} We first define reduction rules for finite single-input and
single-output (SISO) session type \emph{prefixes}.

\begin{definition}[Prefix reduction]
  \label{def:prefix}
  Let us define the syntax of \emph{prefixes} as
  \(
  \begin{array}{lclr}
    \pi, \rho  \bnfas   \epsilon
    \bnfalt  \tout{p}\ell(S)
    \bnfalt  \tin{p}\ell(S)
    \bnfalt  \pi_1.\pi_2
  \end{array}
  \)
  where ``$.$'' denotes the concatenation operator. We define the reduction
  $\pi \reduce \pi'$ as the smallest relation that ensures
  \begingroup
  \small
  \begin{equation*}
    \begin{array}{rrcl}
      \ruleRed{i}
                                                          & \pair{\tin{p}\ell.\pi}{\tin{p}\ell.\pi'}
                                                          & \reduce                                  & \pair{\pi}{\pi'}          \\
      \ruleRed{o}                                         &
      \pair{\tout{p}\ell.\pi}{\tout{p}\ell.\pi'}          & \reduce                                  &
      \pair{\pi}{\pi'}                                                                                                           \\
      \ruleRedA                                           &
      \pair{\tin{p}\ell.\pi}{\refa{p}.\tin{p}\ell.\pi'}   & \reduce                                  & \pair{\pi}{\refa{p}.\pi'} \\
      \ruleRedB                                           &
      \pair{\tout{p}\ell.\pi}{\refb{p}.\tout{p}\ell.\pi'} & \reduce                                  & \pair{\pi}{\refb{p}.\pi'}
    \end{array}
  \end{equation*}
  \endgroup
\end{definition}
\noindent\ruleRed{i} and \ruleRed{o} erase the top input and output prefixes
respectively; \ruleRedA{} formalises \ref{itm:ReorderA} from
\cref{sec:PreciseSubtyping} (permuting the input $\tin{p}\ell$) while
\ruleRedB{} represents \ref{itm:ReorderB} (permuting the output $\tout{p}\ell$).

\begin{example}[Prefix reduction]
  \label{ex:prefix}
  In \cref{ex:amr}, We can consider both \(\lt_Q\) and \(\lt_P\) as prefixes since
  they already contain no choice. The (safe) reordering \(\lt_Q'\) can be achieved
  using \ruleRedB, where $\refb{p} = \tin{p}\ell_1$:
  \begin{equation*}
    \pair{\tout{p}\ell_2.\tin{p}\ell_1.\tend}{\tin{p}\ell_1.\tout{p}\ell_2.\tend} \reduce \pair{\tin{p}\ell_1.\tend}{\tin{p}\ell_1.\tend}
  \end{equation*}
  \noindent On the other hand, the (unsafe)
  reordering \(\lt_P'\) cannot be achieved with \ruleRedA, since $\refa{q} =
    \tout{q}\ell_1$ violates the definition of $\refa{q}$.
\end{example}

\begin{figure}[t]
  {\small
    \begin{prooftree}
      \RightLabel{\ruleAlg{init}}
      \AxiomC{$\epsilon; \nil \vdash_k \triple{\epsilon}{\lt}{n} \leq \triple{\epsilon}{\lt'}{n}$}
      \AxiomC{$k, n \in \mathbb{N}$}
      \BinaryInfC{$\lt \leq \lt'$}
    \end{prooftree}
    \begin{prooftree}
      \AxiomC{}
      \RightLabel{\ruleAlg{end}}
      \UnaryInfC{$\rho; \Sigma \vdash_k \triple{\epsilon}{\tend}{n} \leq \triple{\epsilon}{\tend}{n'}$}
    \end{prooftree}
    \begin{prooftree}
      \AxiomC{$\Sigma\map{\pair{\pi}{\lt} \leq \pair{\pi'}{\lt'}} = \rho\quad \act{\rho'} \supseteq \act{\pi'}$}
      \RightLabel{\ruleAlg{asm}}
      \UnaryInfC{$\rho.\rho'; \Sigma \vdash_k \triple{\pi}{\lt}{n} \leq \triple{\pi'}{\lt'}{n'}$}
    \end{prooftree}
    \begin{prooftree}
      \AxiomC{$\pair{\pi_1}{\pi'_1} \reduce \pair{\pi_2}{\pi'_2}$}
      \AxiomC{$\rho; \Sigma \vdash_k \triple{\pi_2}{\lt}{n} \leq \triple{\pi'_2}{\lt'}{n'}$}
      \RightLabel{\ruleAlg{sub}}
      \BinaryInfC{$\rho; \Sigma \vdash_k \triple{\pi_1}{\lt}{n} \leq \triple{\pi'_1}{\lt'}{n'}$}
    \end{prooftree}
    \begin{prooftree}
      \AxiomC{$\forall i \in I.\
          \forall j \in J.\
          \rho . \tout{p} \ell_i; \Sigma \vdash_k \triple{\pi . \tout{p} \ell_i}{\lt_i}{n} \leq \triple{\pi' . \tin{q} \ell_j}{\lt'_j}{n'}$}
      \RightLabel{\ruleAlg{oi}}
      \UnaryInfC{$\rho; \Sigma \vdash_k \triple{\pi}{\tsel{i \in I} \tout{p} \ell_i . \lt_i}{n} \leq \triple{\pi'}{\tbra{j \in J} \tin{q} \ell_j. \lt'_j}{n'}$}
    \end{prooftree}
    \begin{prooftree}
      \AxiomC{$\forall i \in I.\
          \exists j \in J.\
          \rho . \tout{p} \ell_i; \Sigma \vdash_k \triple{\pi . \tout{p} \ell_i}{\lt_i}{n} \leq \triple{\pi' . \tout{q} \ell_j}{\lt'_j}{n'}$}
      \RightLabel{\ruleAlg{oo}}
      \UnaryInfC{$\rho; \Sigma \vdash_k \triple{\pi}{\tsel{i \in I} \tout{p} \ell_i. \lt_i}{n} \leq \triple{\pi'}{\tsel{j \in J} \tout{q} \ell_j. \lt'_j}{n'}$}
    \end{prooftree}
    \begin{prooftree}
      \AxiomC{$\forall j \in J.\
          \exists i \in I.\
          \rho . \tin{p} \ell_i; \Sigma \vdash_k \triple{\pi . \tin{p} \ell_i}{\lt_i}{n} \leq \triple{\pi' . \tin{q} \ell_j}{\lt'_j}{n'}$}
      \RightLabel{\ruleAlg{ii}}
      \UnaryInfC{$\rho; \Sigma \vdash_k \triple{\pi}{\tbra{i \in I} \tin{p} \ell_i. \lt_i}{n} \leq \triple{\pi'}{\tbra{j \in J} \tin{q} \ell_j. \lt'_j}{n'}$}
    \end{prooftree}
    \begin{prooftree}
      \AxiomC{$\exists i \in I.\
          \exists j \in J.\
          \rho . \tin{p} \ell_i; \Sigma \vdash_k \triple{\pi . \tin{p} \ell_i}{\lt_i}{n} \leq \triple{\pi' . \tout{q} \ell_j}{\lt'_j}{n'}$}
      \RightLabel{\ruleAlg{io}}
      \UnaryInfC{$\rho; \Sigma \vdash_k \triple{\pi}{\tbra{i \in I} \tin{p} \ell_i. \lt_i}{n} \leq \triple{\pi'}{\tsel{j \in J} \tout{q} \ell_j. \lt'_j}{n'}$}
    \end{prooftree}
    \begin{prooftree}
      \AxiomC{$
          \begin{array}{c}
            \Sigma'=\Sigma\map{\pair{\pi}{\trec{t}{\lt}} \leq \pair{\pi'}{\lt'} \mapsto \rho} \\
            \rho; \Sigma' \vdash_k \triple{\pi}{\lt[\trec{t}{\lt} / \tvar{t}]}{n -
              1} \leq \triple{\pi'}{\lt'}{n'}
            \quad\quad n > 0
          \end{array}
        $}
      \RightLabel{\ruleAlg{\(\mu\)l}}
      \UnaryInfC{$\rho; \Sigma \vdash_k \triple{\pi}{\trec{t}{\lt}}{n} \leq \triple{\pi'}{\lt'}{n'}$}
    \end{prooftree}
    \begin{prooftree}
      \AxiomC{$
          \begin{array}{c}
            \Sigma'=\Sigma\map{\pair{\pi}{\lt} \leq \pair{\pi'}{\trec{t}{\lt'}} \mapsto \rho} \\
            \rho; \Sigma' \vdash_k \triple{\pi}{\lt}{n} \leq
            \triple{\pi'}{\lt'[\trec{t}{\lt'} / \tvar{t}]}{n' - 1}\quad\quad n' > 0
          \end{array}$}
      \RightLabel{\ruleAlg{\(\mu\)r}}
      \UnaryInfC{$\rho; \Sigma \vdash_k \triple{\pi}{\lt}{n} \leq \triple{\pi'}{\trec{t}{\lt'}}{n'}$}
    \end{prooftree}
    \begin{prooftree}
      \AxiomC{$
          \begin{array}{c}
            \rho; \Sigma \vdash_{k - 1} \triple{\pi}{\lt}{n} \leq \triple{\pi'}{\lt'}{n'} \\
            \rho; \Sigma \vdash_{k - 1} \triple{\pi'}{\lt'}{n'} \leq
            \triple{\pi''}{\lt''}{n''}
            \quad k > 0
          \end{array}
        $}
      \RightLabel{\ruleAlg{tra}}
      \UnaryInfC{$\rho; \Sigma \vdash_k \triple{\pi}{\lt}{n} \leq \triple{\pi''}{\lt''}{n''}$}
    \end{prooftree}
  }
  \caption{Asynchronous subtyping algorithm rules.}
  \label{fig:SubtypingAlgorithm}
\end{figure}

\paragraph{Algorithm} We define the rules for our asynchronous subtyping
algorithm in \cref{fig:SubtypingAlgorithm}, following the style of
\cite{Gay2005}. Our rules use the function $\act{W}$, the set of input and
output actions of $\pi$ such that $\act{\epsilon} = \nil$,
$\act{\tin{p}\ell.\pi} = \{\tin{p}\} \cup \act{\pi}$ and $\act{\tout{p}\ell.\pi}
  = \{\tout{p}\} \cup \act{\pi}$.

Our algorithm operates on triples of $\triple{\pi}{\lt}{n}$, where $\pi$ is a
session prefix and $n$ is a bound on the number of recursions to unroll. We keep
track of $\rho$, a prefix containing all actions in the subtype seen so far, and
$\Sigma$, a set of subtyping assumptions. Each assumption in $\Sigma$ is
associated with the value of $\rho$ as it was at the time of the assumption. An
additional bound $k$ is included to limit applications of the \ruleAlg{tra}
rule. We use our algorithm to check whether $\lt$ is a subtype of $\lt'$ by
beginning with the \ruleAlg{init} rule. If a proof derivation can be found then
we conclude that $\lt \leq \lt'$. If not, then either $\lt \not\leq \lt'$ or
$\lt \leq \lt'$ but this cannot be shown by our algorithm since $\leq$ is
undecidable (hence our algorithm cannot be complete).

Our algorithm works as follows:
\begin{enumerate*}[label=\textbf{(\arabic*)}]
  \item\label{itm:AlgorithmStart} If both prefixes $\pi$ and $\pi'$ are empty
  and $\lt = \lt' = \tend$, then we have nothing left to check and we terminate
  with success (\ruleAlg{end}).
  \item If $\pair{\pi}{\lt} \leq \pair{\pi'}{\lt'}$ is already in our set of
  assumptions, then we perform a check on $\pi'$ to ensure that no actions have
  been forgotten by the subtype (see \inApp{sec:Proofs} for more detail). If this
  check passes, then we terminate with success (\ruleAlg{asm}).
  \item We attempt to reduce the pair of prefixes $\pair{\pi}{\pi'}$. If we can,
  then the algorithm repeats from \ref{itm:AlgorithmStart} (\ruleAlg{sub}).
  \item If not, we try to pop one action from the start of both $\lt$ and $\lt'$
  and push them to the end of $\pi$ and $\pi'$ respectively. If this is
  possible, we repeat from \ref{itm:AlgorithmStart} (\ruleAlg{oi,oo,ii,io}).
  Note that the quantifiers permit subtyping for internal and external choices.
  For example, \ruleAlg{oo} says $\lt$ is a subtype of $\lt'$ if it has a subset
  of $\lt'$'s internal choices (defined with $\forall i\in I . \exists j\in J$).
  \item Otherwise, we attempt to unroll recursion in $\lt$ (resp.~$\lt'$),
  decrement the bound $n$ (resp.~$n'$) by one and repeat from
  \ref{itm:AlgorithmStart}. If the bound is already zero, we instead terminate
  (\ruleAlg{\(\mu\)l,\(\mu\)r}).
\end{enumerate*}

\paragraph{Double buffering example} We show the execution of our algorithm to
check the optimised type from \cref{sec:Overview} for the kernel in the double
buffering protocol (i.e.\ $\lt \leq \lt'$).
\begin{center}
  \small
  \begin{equation*}
    \lt = \tout{s}\mathit{ready}.\lt' \quad
    \lt' = \trec{x}{\tout{s}\mathit{ready}.\tin{s}\mathit{copy}.\tin{t}\mathit{ready}.\tout{t}\mathit{copy}.\tvar{x}}
  \end{equation*}
  \begin{prooftree}
    \AxiomC{\((\ast)\)}
    \AxiomC{$(\dagger)$}
    \AxiomC{$(\ddagger)$}
    \AxiomC{(\(\star\))}
    \AxiomC{(\(\div\))}
    \AxiomC{$\act{\pi_1.\pi_2.\pi_3.\pi_4} \supseteq \act{\epsilon}$}
    \RightLabel{\ruleAlg{asm}}
    \UnaryInfC{$\rho_5; \Sigma_3 \vdash \triple{\epsilon}{\lt^\prime}{1} \leq \triple{\epsilon}{\lt^{\prime\prime}}{0} $}
    \RightLabel{\ruleAlg{sub}}
    \BinaryInfC{$\rho_5; \Sigma_3 \vdash \triple{\pi_4}{\lt^\prime}{1} \leq \triple{\pi_4}{\lt^{\prime\prime}}{0} $}
    \RightLabel{\ruleAlg{sub}}
    \BinaryInfC{$\rho_5; \Sigma_3 \vdash \triple{\pi_3.\pi_4}{\lt^\prime}{1} \leq \triple{\pi_3.\pi_4}{\lt^{\prime\prime}}{0} $}
    \RightLabel{\ruleAlg{sub}}
    \BinaryInfC{$\rho_5; \Sigma_3 \vdash \triple{\pi_2.\pi_3.\pi_4}{\lt^\prime}{1} \leq \triple{\pi_2.\pi_3.\pi_4}{\lt^{\prime\prime}}{0} $}
    \RightLabel{\ruleAlg{sub}}
    \BinaryInfC{$\rho_5; \Sigma_3 \vdash \triple{\pi_1.\pi_2.\pi_3.\pi_4}{\lt^\prime}{1} \leq \triple{\pi_2.\pi_3.\pi_4.\pi_1}{\lt^{\prime\prime}}{0} $}
    \RightLabel{\ruleAlg{oo}}
    \UnaryInfC{$\rho_4; \Sigma_3 \vdash \triple{\pi_1.\pi_2.\pi_3}{\tout{t}\mathit{copy}.\lt^\prime}{1} \leq \triple{\pi_2.\pi_3.\pi_4}{\tout{s}\mathit{ready}.\lt^{\prime\prime}}{0} $}
    \RightLabel{\ruleAlg{\(\mu\)r}}
    \UnaryInfC{$\rho_4; \Sigma_2 \vdash \triple{\pi_1.\pi_2.\pi_3}{\tout{t}\mathit{copy}.\lt^\prime}{1} \leq \triple{\pi_2.\pi_3.\pi_4}{\lt^\prime}{1} $}
    \RightLabel{\ruleAlg{io}}
    \UnaryInfC{$\rho_3; \Sigma_2 \vdash \triple{\pi_1.\pi_2}{\tin{t}\mathit{ready}.\tout{t}\mathit{copy}.\lt^\prime}{1} \leq \triple{\pi_2.\pi_3}{\tout{t}\mathit{copy}.\lt^\prime}{1} $}
    \RightLabel{\ruleAlg{ii}}
    \UnaryInfC{$\rho_2; \Sigma_2 \vdash \triple{\pi_1}{\lt^{\prime\prime}}{1} \leq \triple{\pi_2}{\tin{t}\mathit{ready}.\tout{t}\mathit{copy}.\lt^\prime}{1} $}
    \RightLabel{\ruleAlg{oi}}
    \UnaryInfC{$\rho_1; \Sigma_2 \vdash \triple{\epsilon}{\tout{s}\mathit{ready}.\lt^{\prime\prime}}{1} \leq \triple{\epsilon}{\lt^{\prime\prime}}{1} $}
    \RightLabel{\ruleAlg{\(\mu\)l}}
    \UnaryInfC{$\rho_1; \Sigma_1 \vdash \triple{\epsilon}{\lt^\prime}{2} \leq \triple{\epsilon}{\lt^{\prime\prime}}{1} $}
    \RightLabel{\ruleAlg{sub}}
    \def\defaultHypSeparation{\hskip -35pt}
    \BinaryInfC{$\rho_1; \Sigma_1 \vdash \triple{\pi_1}{\lt^\prime}{2} \leq \triple{\pi_1}{\lt^{\prime\prime}}{1} $}
    \RightLabel{\ruleAlg{oo}}
    \UnaryInfC{$\epsilon; \Sigma_1 \vdash \triple{\epsilon}{\lt}{2} \leq \triple{\epsilon}{\tout{s}\mathit{ready}.\lt^{\prime\prime}}{1} $}
    \RightLabel{\ruleAlg{\(\mu\)r}}
    \UnaryInfC{$\epsilon; \nil \vdash \triple{\epsilon}{\lt}{2} \leq \triple{\epsilon}{\lt^\prime}{2}$}
    \RightLabel{\ruleAlg{init}}
    \UnaryInfC{$\lt \leq \lt'$}
  \end{prooftree}
  \begin{gather*}
    \lt^{\prime\prime} = \tin{s}\mathit{copy}.\tin{t}\mathit{ready}.\tout{t}\mathit{copy}.\lt^\prime \\
    \pi_1 = \tout{s}\mathit{ready} \quad
    \pi_2 = \tin{s}\mathit{copy} \quad
    \pi_3 = \tin{t}\mathit{ready} \quad
    \pi_4 = \tout{t}\mathit{copy} \\
    \rho_1 = \pi_1 \quad
    \rho_2 = \rho_1.\pi_1 \quad
    \rho_3 = \rho_2.\pi_2 \quad
    \rho_4 = \rho_3.\pi_3 \quad
    \rho_5 = \rho_4.\pi_4
  \end{gather*}
  \begin{gather*}
    \Sigma_1 = \map{\pair{\epsilon}{\lt} \leq \pair{\epsilon}{\lt^\prime}\mapsto \epsilon} \\
    \Sigma_2 = \Sigma_1\map{\pair{\epsilon}{\lt^\prime} \leq \pair{\epsilon}{\lt^{\prime\prime}}\mapsto \tout{s}\mathit{ready}} \\
    \Sigma_3 = \Sigma_2\map{\pair{\pi_1.\pi_2.\pi_3}{\tout{t}\mathit{copy}.\lt^\prime}\leq \pair{\pi_2.\pi_3.\pi_4}{\lt^\prime}\mapsto\rho_4}
  \end{gather*}
  \begin{gather*}
    (\dagger) = \pair{\pi_1.\pi_2.\pi_3.\pi_4}{\pi_2.\pi_3.\pi_4.\pi_1} \reduce \pair{\pi_2.\pi_3.\pi_4}{\pi_2.\pi_3.\pi_4} \; \ruleRedA \\
    (\ddagger) = \pair{\pi_2.\pi_3.\pi_4}{\pi_2.\pi_3.\pi_4} \reduce \pair{\pi_3.\pi_4}{\pi_3.\pi_4} \; \ruleRed{i} \\
    (\star) = \pair{\pi_3.\pi_4}{\pi_3.\pi_4} \reduce \pair{\pi_4}{\pi_4} \; \ruleRed{i} \\
    (\ast) = \pair{\pi_1}{\pi_1} \reduce \pair{\epsilon}{\epsilon} \; \ruleRed{o} \quad
    (\div) = \pair{\pi_4}{\pi_4} \reduce \pair{\epsilon}{\epsilon} \; \ruleRed{o}
  \end{gather*}
\end{center}
See also \inApp{subsec:algorthmex} for the ring and alternating bit protocols,
which include internal and external choices.

\paragraph{Properties} We prove the correctness and complexity of our algorithm.
See \inApp{sec:AppendixSubtyping} for the proofs. We define the size of prefixes
as $\tsize{\epsilon} = 0$ and $\tsize{\tin{p}\ell.\pi} =
  \tsize{\tout{p}\ell.\pi} = 1 + \tsize{\pi}$.

\begin{lemma}
  \label{thm:PrefixesTermination}
  Given finite prefixes $\pi$ and $\pi'$, $\pair{\pi}{\pi'}$ can be reduced only
  a finite number of times.
\end{lemma}

\begin{theorem}[Termination]
  \label{thm:Termination}
  Our subtyping algorithm always eventually terminates.
\end{theorem}

\begin{theorem}[Soundness]
  \label{thm:Soundness}
  Our subtyping algorithm is sound.
\end{theorem}

\begin{lemma}
  \label{thm:PrefixesComplexity}
  Given finite prefixes $\pi$ and $\pi'$, the time complexity of reducing
  $\pair{\pi}{\pi'}$ is $\oh(\min(\tsize{\pi}, \tsize{\pi'}))$.
\end{lemma}

\begin{theorem}[Complexity]
  \label{thm:Complexity}
  Consider $\lt$ and $\lt'$ as (possibly infinite) trees $\tr{\lt}$ and
  $\tr{\lt'}$ with \emph{asymptotic branching factors} $b$ and $b'$
  respectively~\cite{Edelkamp1998,Korf1985}. Our algorithm has time complexity
  $\oh(n\min(b, b')^n)$ and space complexity $\oh(n\min(b, b'))$ in the worst
  case to determine if $\lt \leq \lt'$ with bound $n$.
\end{theorem}

\paragraph{Algorithm implementation}
\label{sec:AlgorithmImplementation}

In practice, we make some minor alterations to the algorithm when implementing
it in \rumpsteak. As outlined in \cref{sec:Overview}, \rumpsteak acts on FSMs
rather than local types so we modify our bounds-checking accordingly. We also
represent prefixes as lazily-removable lists for greater memory
efficiency---this additionally allows a slightly simplified termination
condition in the case of \ruleAlg{asm}. Finally, we provide some opportunities
for the algorithm to ``short circuit'' in cases where we can tell early that a
subtype is not valid. See \inApp{subsec:impalg} for more details.

\begin{figure*}
  \centering
  \sffamily\footnotesize
  \pgfplotslegendfromname{pgfplots:Runtime}
  \begin{tikzpicture}[every axis plot/.append style={thick}]
    \tikzset{mark options={mark size=1pt}}
    \begin{groupplot}[
        group style={group size=3 by 1},
        height=4cm,
        width=6.5cm
      ]
      \nextgroupplot[
        title=\textit{Streaming},
        xlabel=Number of values transferred (\(n\)),
        ylabel=Throughput (\(n\)/\(\mu\)s),
        xtick distance=10,
      ]
      \addplot[solid,mark=square*,color=primary,every mark/.append style={solid, fill=primary!30}] coordinates {
          (10,0.019389)
          (20,0.028142)
          (30,0.034193)
          (40,0.036566)
          (50,0.040315)
        };
      \addplot[dashed,mark=square*,color=secondary,every mark/.append style={solid, fill=secondary!30}] coordinates {
          (10,0.011678)
          (20,0.014325)
          (30,0.015160)
          (40,0.016072)
          (50,0.016577)
        };
      \addplot[densely dotted,mark=square*,color=tertiary,every mark/.append style={solid, fill=tertiary!30}] coordinates {
          (10,0.011386)
          (20,0.012994)
          (30,0.013463)
          (40,0.013671)
          (50,0.014126)
        };
      \addplot[dashed,mark=*,color=quinary,every mark/.append style={solid, fill=quinary!30}] coordinates {
          (10,0.202587)
          (20,0.336988)
          (30,0.427489)
          (40,0.488886)
          (50,0.545378)
        };
      \addplot[densely dotted,mark=*,color=senary,every mark/.append style={solid, fill=senary!30}] coordinates {
          (10,0.215583)
          (20,0.356978)
          (30,0.437795)
          (40,0.517468)
          (50,0.583366)
        };
      \nextgroupplot[
        title={\textit{Double Buffering}~\cite{CastroPerez2020,Huang2002}},
        xlabel=Number of values in each buffer (\(n\)),
        xtick distance=2500,
        ytick distance=20,
        scaled x ticks=false
      ]
      \addplot[solid,mark=square*,color=primary,every mark/.append style={solid, fill=primary!30}] coordinates {
          (2500,6.929567)
          (5000,13.138401)
          (7500,18.739983)
          (10000,24.103215)
          (12500,28.609966)
        };
      \addplot[dashed,mark=square*,color=secondary,every mark/.append style={solid, fill=secondary!30}] coordinates {
          (2500,5.675414)
          (5000,11.254181)
          (7500,16.187341)
          (10000,20.481378)
          (12500,25.050058)
        };
      \addplot[densely dotted,mark=square*,color=tertiary,every mark/.append style={solid, fill=tertiary!30}] coordinates {
          (2500,7.617643)
          (5000,14.649028)
          (7500,20.429845)
          (10000,25.506427)
          (12500,29.629025)
        };
      \addplot[dashed,mark=*,color=quinary,every mark/.append style={solid, fill=quinary!30}] coordinates {
          (2500,27.704354)
          (5000,44.154722)
          (7500,56.813002)
          (10000,67.595301)
          (12500,75.848611)
        };
      \addplot[densely dotted,mark=*,color=senary,every mark/.append style={solid, fill=senary!30}] coordinates {
          (2500,32.340989)
          (5000,50.126532)
          (7500,67.884430)
          (10000,82.039366)
          (12500,96.010424)
        };
      \nextgroupplot[
        title={\textit{FFT}~\cite{CastroPerez2020}},
        xlabel=Number of columns (\(n\)),
        xtick distance=1000,
        ytick distance=2,
        legend columns=-1,
        legend entries={\sesh,\multicrusty,\ferrite,\rustfft,\rumpsteak,\rumpsteak (optimised)},
        legend style={draw=none,column sep=5pt},
        legend to name=pgfplots:Runtime
      ]
      \addplot[solid,mark=square*,color=primary,every mark/.append style={solid, fill=primary!30}] coordinates {
          (1000,0.551154)
          (2000,1.050958)
          (3000,1.510567)
          (4000,1.935263)
          (5000,2.303627)
        };
      \addplot[dashed,mark=square*,color=secondary,every mark/.append style={solid, fill=secondary!30}] coordinates {
          (1000,0.810134)
          (2000,1.515538)
          (3000,2.163629)
          (4000,2.783617)
          (5000,3.261020)
        };
      \addplot[densely dotted,mark=square*,color=tertiary,every mark/.append style={solid, fill=tertiary!30}] coordinates {
          (1000,1.458279)
          (2000,2.513855)
          (3000,3.496405)
          (4000,4.198723)
          (5000,4.811375)
        };
      \addplot[solid,mark=*,color=quaternary,every mark/.append style={solid, fill=quaternary!30}] coordinates {
          (1000,9.320778)
          (2000,9.313359)
          (3000,9.333569)
          (4000,9.336939)
          (5000,9.323199)
        };
      \addplot[dashed,mark=*,color=quinary,every mark/.append style={solid, fill=quinary!30}] coordinates {
          (1000,5.038554)
          (2000,7.206404)
          (3000,8.421026)
          (4000,9.262763)
          (5000,9.316716)
        };
      \addlegendimage{densely dotted,mark=*,color=senary,every mark/.append style={solid, fill=senary!30}};
    \end{groupplot}
    \path (current bounding box.north)--++(0,5pt);
  \end{tikzpicture}
  \caption{Benchmarking \rumpsteak's runtime performance against previous work in Rust (raw data in \inApp{sec:RuntimeData}).}
  \label{fig:RuntimeBenchmarks}
\end{figure*}

\section{Evaluation}
\label{sec:Evaluation}

In this section, we evaluate how \rumpsteak performs with respect to existing
tools. First, in \cref{sec:RuntimeEvaluation}, we evaluate the runtime
performance of programs written in \rumpsteak versus the same benchmarks
implemented using other Rust session type tools. In the spirit of Rust's
emphasis on efficiency, this runtime performance is of particular significance
for developers. Secondly, in \cref{sec:eval_async_verif}, we evaluate how
\rumpsteak's verification of message reordering (our subtyping algorithm) scales
compared to existing verification tools. Although not a runtime cost, subtyping
is known to be a computationally challenging problem that often scales poorly.

\subsection{Session-Based Rust Implementations}
\label{sec:RuntimeEvaluation}

We compare \rumpsteak's runtime against three other session type implementations
in Rust:
\begin{enumerate*}[label=\textbf{(\arabic*)}]
  \item \sesh~\cite{Kokke2019}, a synchronous implementation of binary session
  types;
  \item \ferrite~\cite{Chen2021}, an implementation of shared binary session
  types~\cite{Balzer2017} supporting asynchronous execution; and
  \item \multicrusty~\cite{Lagaillardie2020}, a synchronous \mpst implementation
  based on \sesh.
\end{enumerate*}

We perform a series of benchmarks shown in \cref{fig:RuntimeBenchmarks}. We
execute these using a 16-core AMD Opteron\textsuperscript{TM} 6200 Series CPU @
2.6GHz with hyperthreading, 128GB of RAM, Ubuntu 18.04.5 LTS and Rust Nightly
2021-07-06. We use version 0.3.5 of the Criterion.rs library~\cite{Criterion} to
perform microbenchmarking and a \emph{multi-threaded} asynchronous runtime from
version 1.11.0 of the Tokio library~\cite{Tokio}.

\begin{description}
  \item[Streaming.] This protocol has two participants, a source \ppt{s} and a sink
    \ppt{t}, with a combination of recursion and choice. As stated previously, its global type
    $\gt_\textsf{ST}$ is given as
    \begin{equation*}
      \gt_\textsf{ST} = \gtrec{x}\gtmsg{t}{s}{\mathit{ready}.\gtmsg{s}{t}{
          \mathit{value}.\gtvar{x}, \quad
          \mathit{stop}.\gtend
        }}
    \end{equation*}
    We benchmark the protocol by varying the number of \(\mathit{value}\)s that
    are sent before the exchange \(\mathit{stop}\)s. We use \amr to create an
    optimised version for \rumpsteak: if the source knows it will send at least
    \(n\) \(\mathit{value}\)s before \(\mathit{stop}\)ping, then these messages
    can be \emph{unrolled} and sent all at once---only receiving
    \(\mathit{ready}\) from the sink after sending all \(n\)
    \(\mathit{value}\)s. For benchmarking, we unroll the first 5
    \(\mathit{value}\)s in the optimised version.

    Our results show \rumpsteak reaches around 14.5x the throughput of other
    implementations. At first, it is limited by channel creation overheads, the
    cost of which becomes less significant as more messages are sent. The
    throughput levels off as message passing overheads become the bottleneck.
    The optimised version eases this overhead, as the source is less frequently
    blocked on receiving \(\mathit{ready}\) from the sink---this effect could be
    increased by unrolling more messages.

    \sesh and \multicrusty have much lower throughputs since they use more
    expensive synchronous communication. These implementations also create a new
    channel for each interaction. This causes their throughput to stay constant
    as there are fewer one-time costs to spread than in \rumpsteak.
    Interestingly, \ferrite performs similarly to \sesh and \multicrusty,
    despite employing asynchronous execution like \rumpsteak. This is because
    its design requires that the source and sink are implemented using
    \emph{recursion} rather than \emph{iteration}. \ferrite also requires
    particularly strong safety requirements at compile time. Therefore, the
    sink's output buffer must be guarded with a mutex instead of being accessed
    directly.

  \item[Double Buffering.] We benchmark our running example by performing only
    two iterations. This allows for both of the kernel's buffers to be filled
    and, importantly, for the protocol to terminate. \sesh and \ferrite do not
    support \mpst so their implementations use binary session types between
    pairs of participants. This approach does \emph{not} provide the same safety
    as \mpst and so we cannot verify deadlock-freedom as we can in \multicrusty
    and \rumpsteak (see \cref{tab:Expressiveness}).

    We parameterise the size of the buffers and measure the throughput.
    Performance is similar to that of the stream benchmark; \rumpsteak's
    throughput reaches around 3.2x that of the others. Moreover, we confirm the
    intuition described in \cref{sec:Introduction}. Increasing the size of the
    buffers does generally improve throughput. However, it remains limited when
    only a single buffer can be used at once. \cref{fig:RuntimeBenchmarks} shows
    that optimising the kernel as described in \cref{sec:Overview} increases the
    throughput since the source and sink operate on different buffers at once.

    \begin{table*}
      \centering
      \caption{Expressiveness of \rumpsteak compared to previous work.}
      \label{tab:Expressiveness}
      \sffamily\footnotesize
      \setlength{\tabcolsep}{6pt}
      \begin{tabular}{lc|cccc|ccc|c|cc}
        \toprule
        Protocol                                                    & \(n\) & \feat{C} & \feat{R} & \feat{IR} & \feat{AMR} & \sesh  & \ferrite & \multicrusty & \rumpsteak & \kmc   & \concur \\
        \midrule
        Two Adder~\cite{NuScr}                                      & 2     & \tick    & \tick    &           &            & \ftick & \ftick   & \ftick       & \ftick     & \ftick & \ftick  \\
        Three Adder                                                 & 3     &          &          &           &            & \htick & \htick   & \ftick       & \ftick     & \ftick & \cross  \\
        Streaming                                                   & 2     & \tick    & \tick    &           &            & \ftick & \ftick   & \ftick       & \ftick     & \ftick & \ftick  \\
        Optimised Streaming                                         & 2     & \tick    & \tick    &           & \tick      & \cross & \cross   & \cross       & \ftick     & \ftick & \ftick  \\
        Ring~\cite{CastroPerez2020}                                 & 3     &          & \tick    & \tick     &            & \htick & \htick   & \ftick       & \ftick     & \ftick & \cross  \\
        Optimised Ring~\cite{CastroPerez2020}                       & 3     &          & \tick    & \tick     & \tick      & \htick & \htick   & \htick       & \ftick     & \ftick & \cross  \\
        Ring With Choice~\cite{CastroPerez2020}                     & 3     & \tick    & \tick    & \tick     &            & \htick & \htick   & \ftick       & \ftick     & \ftick & \cross  \\
        Optimised Ring With Choice~\cite{CastroPerez2020}           & 3     & \tick    & \tick    & \tick     & \tick      & \htick & \htick   & \htick       & \ftick     & \ftick & \cross  \\
        Double Buffering~\cite{CastroPerez2020}                     & 3     &          & \tick    & \tick     &            & \htick & \htick   & \ftick       & \ftick     & \ftick & \cross  \\
        Optimised Double Buffering~\cite{CastroPerez2020,Huang2002} & 3     &          & \tick    & \tick     & \tick      & \cross & \cross   & \cross       & \ftick     & \ftick & \cross  \\
        Alternating Bit~\cite{Lange2019,AlternatingBit}             & 2     & \tick    & \tick    & \tick     &            & \htick & \htick   & \htick       & \ftick     & \ftick & \ftick  \\
        Elevator~\cite{Lange2019,Bouajjani2018}                     & 3     & \tick    & \tick    & \tick     & \tick      & \cross & \cross   & \cross       & \ftick     & \ftick & \cross  \\
        FFT~\cite{CastroPerez2020}                                  & 8     &          &          &           &            & \htick & \htick   & \ftick       & \ftick     & \ftick & \cross  \\
        Optimised FFT~\cite{CastroPerez2020}                        & 8     &          &          &           & \tick      & \cross & \cross   & \cross       & \ftick     & \ftick & \cross  \\
        Authentication~\cite{Neykova2013}                           & 3     & \tick    &          &           &            & \htick & \htick   & \ftick       & \ftick     & \ftick & \cross  \\
        Client-Server Log~\cite{Lagaillardie2020}                   & 3     & \tick    & \tick    & \tick     &            & \htick & \htick   & \ftick       & \ftick     & \ftick & \cross  \\
        Hospital~\cite{Bravetti2019}                                & 2     & \tick    & \tick    & \tick     & \tick      & \cross & \cross   & \cross       & \htick     & \cross & \ftick  \\
        \bottomrule
      \end{tabular} \\[4pt]
      \footnotesize
      \(n\)\; Number of participants \quad \feat{C}\; Choice \quad \feat{R}\; Recursion \quad \feat{IR}\; Infinite recursion \quad \feat{AMR}\; Asynchronous message reordering \\[2pt]
      \ftick\; Expressible \quad \htick\; Expressible using endpoint types
      (but without deadlock-freedom guarantee) \quad \cross\; Not expressible
    \end{table*}

  \item[FFT.] The fast Fourier transform (FFT) is an algorithm for computing
    the discrete Fourier transform of a vector. We take \(n \times 8\)
    \emph{matrices}, where the FFT is computed on each column to produce a new,
    transformed matrix, and implement the Cooley-Tukey FFT algorithm which
    divides the problem into---in our case---eight. Each of these problems can
    be solved independently by different processes, which communicate with one
    another using message passing. The exchanges between participants are
    described in \cite[Fig. 7]{CastroPerez2020}.

    We have chosen to use FFT as a benchmark as it is a standard problem with
    numerous implementations, and we can therefore compare implementations
    based on MPST with actual high-optimised implementations.

    We implement a concurrent version of the algorithm in \sesh, \ferrite,
    \multicrusty and \rumpsteak, which uses eight processes. Each process works
    on an array of \(n\) inputs, representing one column. Arithmetic operations
    are performed in a pairwise fashion on two of these arrays.

    We use the same approach as in the double buffering protocol to write the
    binary implementations for \sesh and \ferrite. Interestingly, to represent
    this as a sequence of binary sessions, we require additional synchronisation
    of all participants at each stage of the protocol. While we can perform \amr
    to \rumpsteak's version, in practice we found that this does not have much
    effect since the protocol is already heavily synchronised so cannot progress
    any faster.

    We benchmark the protocol by varying the number of columns in the matrix. We
    also compare these implementations with the most downloaded open-source Rust
    FFT implementation, \rustfft~\cite{RustFFT}. \rustfft does not use
    concurrency and computes the FFT of a matrix by iteratively performing the
    Cooley-Tukey algorithm on each column.

    \ferrite performs better than before since its FFT implementation does not
    suffer from the limitations explained previously. Like \rumpsteak, it
    benefits from the use of asynchronous execution, although the additional
    synchronisation from representing the problem as a set of binary sessions
    causes \rumpsteak's throughput to remain around 1.9x greater.

    Most excitingly, \rumpsteak achieves \rustfft's throughput for large
    matrices. \rustfft's implementation is highly tuned for low-level efficiency
    and does not incur any overheads associated with message passing. Moreover,
    asynchronous executors are generally designed for higher-level tasks such as
    networking; MPI-based communication~\cite{MPI} would more likely be used to
    parallelise this problem in practice. Therefore, we think it impressive that
    even a basic parallel implementation using \rumpsteak can reach the same
    performance as a state-of-the-art sequential implementation.

  \item[Expressiveness.] \cref{tab:Expressiveness} discusses the expressiveness
    of \rumpsteak compared with \sesh, \ferrite and \multicrusty. Since \sesh
    and \ferrite support only binary session types, they are unable to guarantee
    deadlock-freedom in protocols with more than two participants. \multicrusty
    has greater expressiveness than \sesh and \ferrite since it implements \mpst
    but, unlike \rumpsteak, still cannot ensure deadlock-freedom for protocols
    optimised using \amr. In addition, many optimisations, like the ones we have
    benchmarked, break duality between pairs of participants so are not
    expressible at all by \sesh, \ferrite and \multicrusty. Meanwhile, our
    powerful API and new subtyping algorithm allow \rumpsteak to express many
    examples using \amr.
\end{description}

\subsection{Verifying Asynchronous Message Reordering}
\label{sec:eval_async_verif}

We perform a second set of benchmarks, shown in \cref{fig:SubtypingBenchmarks},
to evaluate \rumpsteak's asynchronous subtyping algorithm. We compare it against
\begin{enumerate*}[label=\textbf{(\arabic*)}]
  \item \concur~\cite{Bravetti2019}, a sound subtyping algorithm defined for
  \emph{binary} session types only; and
  \item \kmc~\cite{Lange2019}, an algorithm for directly checking the
  compatibility of a set of FSMs without the need for a global type.
  Compatibility is checked up to a bound \(k\) on the size of each process'
  asynchronous queue.
\end{enumerate*}
We note that neither \concur nor \kmc provide a runtime framework like
\rumpsteak does, they are used only to verify the \amr.

\begin{figure*}
  \centering
  \sffamily\footnotesize
  \pgfplotslegendfromname{pgfplots:Subtyping}
  \begin{tikzpicture}[every axis plot/.append style={thick}]
    \tikzset{mark options={mark size=1pt}}
    \begin{groupplot}[
        group style={group size=4 by 1},
        height=4cm,
        width=4.8cm,
        ymode=log,
        ytick distance=10
      ]
      \nextgroupplot[
        title=\textit{Streaming},
        xlabel=Number of unrolls ($n$),
        ylabel=Running time (\(\log \text{s}\)),
        legend columns=-1,
        legend entries={\concur,\kmc,\rumpsteak},
        legend style={draw=none,column sep=5pt},
        legend to name=pgfplots:Subtyping
      ]
      \addplot[solid,mark=square*,color=primary,every mark/.append style={solid, fill=primary!30}] coordinates {
          (0,0.003476)
          (10,0.008556)
          (20,0.020673)
          (30,0.041673)
          (40,0.076425)
          (50,0.127865)
          (60,0.198541)
          (70,0.292471)
          (80,0.422571)
          (90,0.583863)
          (100,0.767426)
        };
      \addplot[dashed,mark=square*,color=secondary,every mark/.append style={solid, fill=secondary!30}] coordinates {
          (0,0.005504)
          (10,0.019316)
          (20,0.057417)
          (30,0.142145)
          (40,0.276446)
          (50,0.496929)
          (60,0.805577)
          (70,1.233327)
          (80,1.780778)
          (90,2.475443)
          (100,3.349204)
        };
      \addplot[densely dotted,mark=*,color=quinary,every mark/.append style={solid, fill=quinary!30}] coordinates {
          (0,0.001872)
          (10,0.001899)
          (20,0.001848)
          (30,0.001906)
          (40,0.001874)
          (50,0.002080)
          (60,0.002083)
          (70,0.002064)
          (80,0.002178)
          (90,0.002190)
          (100,0.002249)
        };
      \nextgroupplot[
        title={\textit{Nested Choice}~\cite[Fig. 3]{Chen2017}},
        xlabel=Number of levels ($n$),
        xtick distance=1
      ]
      \addplot[solid,mark=square*,color=primary,every mark/.append style={solid, fill=primary!30}] coordinates {
          (1,0.002295)
          (2,0.004504)
          (3,0.016347)
          (4,0.224858)
          (5,4.692525)
        };
      \addplot[dashed,mark=square*,color=secondary,every mark/.append style={solid, fill=secondary!30}] coordinates {
          (1,0.006554)
          (2,0.014901)
          (3,0.072423)
          (4,1.515528)
          (5,41.688068)
        };
      \addplot[densely dotted,mark=*,color=quinary,every mark/.append style={solid, fill=quinary!30}] coordinates {
          (1,0.000702)
          (2,0.000755)
          (3,0.001745)
          (4,0.007656)
          (5,0.157548)
        };
      \nextgroupplot[
        title={\textit{Ring}~\cite{CastroPerez2020}},
        xlabel=Number of participants ($n$)
      ]
      \addplot[dashed,mark=square*,color=secondary,every mark/.append style={solid, fill=secondary!30}] coordinates {
          (2,0.004007)
          (4,0.007239)
          (6,0.011806)
          (8,0.018822)
          (10,0.024842)
          (12,0.049232)
          (14,0.102257)
          (16,0.191078)
          (18,0.340262)
          (20,0.570656)
          (22,0.913412)
          (24,1.391075)
          (26,2.042452)
          (28,2.918943)
          (30,4.099072)
        };
      \addplot[densely dotted,mark=*,color=quinary,every mark/.append style={solid, fill=quinary!30}] coordinates {
          (2,0.000675)
          (4,0.000731)
          (6,0.000701)
          (8,0.000835)
          (10,0.000757)
          (12,0.000777)
          (14,0.000744)
          (16,0.000813)
          (18,0.000817)
          (20,0.000766)
          (22,0.000911)
          (24,0.000737)
          (26,0.000752)
          (28,0.000732)
          (30,0.000769)
        };
      \nextgroupplot[
        title={\textit{\(k\)-Buffering}~\cite{CastroPerez2020,Huang2002}},
        xlabel=Number of unrolls ($n$)
      ]
      \addplot[dashed,mark=square*,color=secondary,every mark/.append style={solid, fill=secondary!30}] coordinates {
          (0,0.004825)
          (5,0.007668)
          (10,0.013613)
          (15,0.018770)
          (20,0.031376)
          (25,0.054910)
          (30,0.080879)
          (35,0.122315)
          (40,0.170533)
          (45,0.236354)
          (50,0.305749)
          (55,0.406071)
          (60,0.506069)
          (65,0.639521)
          (70,0.773931)
          (75,0.954399)
          (80,1.127240)
          (85,1.359600)
          (90,1.571745)
          (95,1.869339)
          (100,2.111687)
        };
      \addplot[densely dotted,mark=*,color=quinary,every mark/.append style={solid, fill=quinary!30}] coordinates {
          (0,0.000630)
          (5,0.000747)
          (10,0.000705)
          (15,0.000667)
          (20,0.000825)
          (25,0.000718)
          (30,0.000760)
          (35,0.000853)
          (40,0.000802)
          (45,0.000792)
          (50,0.000916)
          (55,0.000882)
          (60,0.000959)
          (65,0.001028)
          (70,0.001057)
          (75,0.001045)
          (80,0.001125)
          (85,0.001120)
          (90,0.001164)
          (95,0.001156)
          (100,0.001234)
        };
    \end{groupplot}
    \path (current bounding box.north)--++(0,5pt);
  \end{tikzpicture}
  \caption{Benchmarking \rumpsteak's subtyping performance against previous work  (raw data in \inApp{sec:VerificationData}).}
  \label{fig:SubtypingBenchmarks}
\end{figure*}

We benchmark with the same machine as before. \concur and \kmc are written in
Haskell so we must run each tool's binary rather than simply timing Rust
functions. We, therefore, provide a command-line for \rumpsteak's subtyping
algorithm so it is comparable with these other tools. \rumpsteak's binary is
compiled with Rust 1.54.0 and we use Hyperfine~\cite{Hyperfine} to compare the
execution time for each tool.

\begin{description}
  \item[Streaming (from \cref{sec:RuntimeEvaluation}).] We vary \(n\), the
    number of \(\mathit{value}\)s we unroll. Using \concur and \rumpsteak, we
    check that the optimised source is a subtype of its projected version and,
    using \kmc, we check that the optimised source is compatible with the sink.
    Our results show that \rumpsteak scales significantly better than both
    \concur and \kmc. While all three implementations can verify up to 100
    unrolls in under a second, the execution time taken by \concur and \kmc
    increases exponentially while \rumpsteak's remains mostly flat. This is
    consistent with our complexity analysis in \cref{thm:Complexity}, since
    unrolling more \(\mathit{value}\)s does not increase the branching factor of
    the subtype.

  \item[Nested choice.] We next consider a protocol from \citet[Fig.
      3]{Chen2017} containing nested choice. We perform this nesting up to a
    parameterised number of levels $n$ to increase the complexity.
    Specifically, we check that $\lt_n \leq \lt'_n$ where
    \begin{center}
      $\lt_0 = \lt'_0 = \tend$ \\
      $\lt_{n + 1}  = \tout{}\mathit{m}.(\tin{}\mathit{r}.\lt_n\ \&\ \tin{}\mathit{s}.\lt_n\ \&\ \tin{}\mathit{u}.\lt_n) \oplus \tout{}\mathit{p}.(\tin{}\mathit{r}.\lt_n\ \&\ \tin{}\mathit{s}.\lt_n)$        \\
      $\lt'_{n + 1} = \tin{}\mathit{r}.(\tout{}\mathit{m}.\lt'_n \oplus \tout{}\mathit{p}.\lt'_n \oplus \tout{}\mathit{q}.\lt'_n)\ \&\ \tin{}\mathit{s}.(\tout{}\mathit{m}.\lt'_n \oplus \tout{}\mathit{p}.\lt'_n)$
    \end{center}

    We find that \rumpsteak again performs more scalably than \concur and \kmc.
    Here, \rumpsteak's improved efficiency is significant---for five levels,
    \kmc takes around 40s while \rumpsteak requires only a fraction of a second.
    Nonetheless, we observe that \rumpsteak also exhibits exponential
    performance, in contrast to our stream benchmark. This is because our
    algorithm bounds recursion but not choice, causing an explosion in the paths
    to visit as the number of nested choices increases. This is also explained
    by our complexity analysis in \cref{thm:Complexity} since a greater number
    of nested choices increases the depth to which we must explore.

  \item[Ring~\cite{CastroPerez2020}.] We consider a ring protocol with a parameterised number of
    participants. Each participant (aside from the first, which initiates the
    protocol) receives a value from its preceding neighbour then sends a value
    to its succeeding neighbour. Providing that this receive is not dependent on
    the send, we can use \amr to send before receiving.

    We benchmark verifying this optimisation using both \kmc and \rumpsteak. In
    this case, we cannot use \concur since the protocol is multiparty. While
    \rumpsteak can verify each participant's subtype individually, \kmc must
    consider the entire protocol at once. Unsurprisingly, \kmc is significantly
    less scalable as its running time grows exponentially, whereas \rumpsteak's
    performance remains mostly constant.

  \item[\(k\)-buffering.] We extend the double buffering protocol to a
    parameterised number of buffers. We benchmark verifying the optimisation to
    the kernel discussed previously. Crucially, with a greater number of
    buffers, we can unroll a greater number of \(\text{ready}\) messages. We
    again compare only \kmc and \rumpsteak since the protocol is multiparty.
    Similarly to other benchmarks, \rumpsteak scales more efficiently than \kmc.

  \item[Expressiveness.] \cref{tab:Expressiveness} discusses the expressiveness
    of \rumpsteak's algorithm against \concur and \kmc (note again that \concur
    and \kmc only verify \amr and do not provide language implementations like
    \rumpsteak). As we discussed, \concur supports only two parties so cannot
    express many of the case studies. However, it can express some unbounded
    binary protocols that \rumpsteak and \kmc cannot, such as the Hospital
    example~\cite[\textbf{\S~1}]{Bravetti2019}\footnote{Notice, for the case of
      \rumpsteak, that we can manually write the endpoints, and the framework
      checks the conformance to the protocol. However, \rumpsteak cannot verify
      the \emph{deadlock-freedom} property of the protocol, hence the \emph{amber
        cross}.}. Unsurprisingly, \kmc's expressiveness in the table coincides with
    \rumpsteak's as \kmc with $k{=}\infty$ is equivalent to the liveness
    property induced by asynchronous multiparty subtyping~\cite[Theorem
      7.3]{Ghilezan2021}. On the other hand, \kmc can verify a wider syntax of
    local FSMs than those corresponding to \cref{def:syntax}.
\end{description}

\section{Conclusion and Related Work}
\label{sec:Related}

There are a vast number of studies on session types, some of which are
implemented in programming languages~\cite{Ancona2016} and tools~\cite{Gay2017}.
The top-down approach using Scribble~\cite{ScribbleWebsite,Yoshida2013,NuScr} we
presented in this paper has been implemented for a number of other programming
languages such as
Java~\cite{Hu2016,kouzapas_typechecking_2018,hu_explicit_2017},
Go~\cite{Castro2019}, TypeScript~\cite{Miu2021}, Scala~\cite{Scalas2017},
MPI-C~\cite{ng_protocols_2015},
Erlang~\cite{NY2017}, Python~\cite{DHHNY2015},
F\#~\cite{Neykova2018}, F$\star$~\cite{Zhou2020}
and Actor DSL~\cite{harvey_multiparty_2021}.
Several implementations use an EFSM-based approach to generate
APIs from Scribble for target programming languages such
as~\cite{Hu2016,hu_explicit_2017,Neykova2018,Castro2019,Zhou2020,Miu2021,Scalas2017}
to ensure correctness by construction,
but none are integrated with \amr like \rumpsteak.

\rumpsteak provides three approaches for gaining \emph{efficiency} and ensuring
\emph{deadlock-freedom} in message-passing Rust applications: (1) the
\emph{top-down} approach for ensuring \emph{correctness/safety by construction}
and maximising \emph{asynchrony} with local analysis; (2) the \emph{bottom-up}
approach using global analysis on a set of FSMs; and (3) the \emph{hybrid}
approach, which combines inference with local analysis. All three approaches are
backed by \rumpsteak's API (described in \cref{sec:Overview}), which uses Rust's
affine type system to ensure protocol compliance.

Note that \rumpsteak is the first framework (for any programming language) to
enable the bottom-up (by integrating with \kmc~\cite{Lange2019}) and the hybrid
approaches.

\begin{description}
  \item[Rust session type implementations.]
    We closely compared the performance and expressiveness of \rumpsteak with
    three existing works in \cref{sec:Evaluation}:
    \begin{enumerate*}[label=\textbf{(\arabic*)}]
      \item \sesh~\cite{Kokke2019}, a synchronous implementation of binary
      session types;
      \item \ferrite~\cite{Chen2021}, an implementation of shared binary session
      types~\cite{Balzer2017} supporting asynchronous execution; and
      \item \multicrusty~\cite{Lagaillardie2020}, a synchronous \mpst
      implementation based on \sesh.
    \end{enumerate*}

    Of these previous implementations, only \multicrusty can also support \mpst
    with deadlock-freedom. However, to represent a multiparty session,
    \multicrusty requires defining a tuple of binary sessions for each role as
    well as their order of use. This leads to a more complex and less intuitive
    API than \rumpsteak's. Our API requires fewer definitions and is also both
    more expressive and performant (see \cref{sec:Evaluation}).

    \rumpsteak and \ferrite support asynchronous execution, which is more
    efficient than the synchronous and blocking communication used by \sesh and
    \multicrusty. However, only \rumpsteak supports Rust's idiomatic
    \async/\await syntax. \ferrite instead requires nesting sequential
    communication, which is more verbose and less efficient. Iteration cannot be
    easily expressed (see~\cref{sec:Evaluation}) and it requires stricter
    compile-time concurrency guarantees than \rumpsteak. We do not compare
    directly against~\cite{Jespersen2015} (another synchronous implementation of
    binary session types) as this uses an older edition of Rust and has several
    limitations already discussed in~\cite{Lagaillardie2020,Kokke2019}.

    \rumpsteak can perform \amr, which is not possible in any of these existing
    implementations. Combined with its straightforward design and use of
    asynchronous execution, we found that it is therefore much more efficient
    than all previous work. It has around 14.5x, 3.2x and 1.9x greater
    throughput than the next fastest implementation in the stream, double
    buffering and FFT protocols respectively. \rumpsteak can also compete with a
    popular Rust implementation in the FFT benchmark, achieving the same
    throughput for large input sizes.

    Recent work~\cite{Duarte2021} builds a DSL to offer protocol conformance for
    Rust, supported by typestates. They do not yet explore combining the DSL
    with communication channels like \rumpsteak and
    \cite{Jespersen2015,Lagaillardie2020,Kokke2019,Balzer2017}. If this is
    achieved, it would be interesting to integrate session type tooling with
    their DSL and compare this with \rumpsteak.

  \item[Verification of asynchronous subtyping.] Interest in \amr has grown
    recently, both in theory and practice. An extension of the Iris framework,
    Actris~\cite{Hinrichsen_sessions_2021}, formalises a variant of binary
    asynchronous subtyping, which is in turn implemented in the Coq proof
    assistant. The asynchronous session subtyping defined
    in~\cite{Chen2017,Ghilezan2021} is \emph{precise} but was shown to be
    undecidable, even for binary sessions \cite{Lange2017,Bravetti2018}. Hence,
    in general, checking $\efsm_i' \leq \efsm_i$ is undecidable. Various limited
    classes of session types for which $\efsm_i' \leq \efsm_i$ is decidable
    \cite{Bravetti2019,Bravetti2017,Lange2017,CastroPerez2020} are proposed but
    they are not applicable to our use cases since \textbf{(1)} the relations in
    \cite{Bravetti2019,Bravetti2018,Lange2017} are \emph{binary} and the same
    limitations do not work for multiparty sessions; and \textbf{(2)} the
    relation in \cite[Def.~6.1]{CastroPerez2020} does not handle subtyping
    across unrolling recursions, e.g.\ the relation is inapplicable to the
    double buffering algorithm~\cite{Huang2002} (see \cite[Remark
      8.1]{CastroPerez2020}).

    Our new algorithm for asynchronous optimisation is \emph{terminating} (see
    \cref{thm:Termination}), \emph{sound} (see \cref{thm:Soundness}) and capable
    of verifying optimisations in a range of classical examples (see
    \cref{sec:Evaluation}). It is also more performant than a global analysis
    based on $k$-MC, and an existing two-party sound
    algorithm~\cite{Bravetti2019} (see \cref{sec:Evaluation}). This is because
    our algorithm is implemented efficiently (see
    \cref{sec:AlgorithmImplementation}) and can often execute in linear rather
    than exponential time (see \cref{thm:Complexity}).

  \item[Deadlock detection.] There is also recent
    progress~\cite{Qin2020Detector,Qin2020} on static deadlock detection in Rust
    by tracking locks and unlocks in Rust's intermediate representation
    (MIR)~\cite{RustMIR}. This targets shared memory, rather than
    message-passing applications, and does not attempt safety by construction.
    In future work, we could also analyse the MIR directly in \rumpsteak,
    replacing the use of our API. However, it is difficult to identify
    concurrency primitives in the MIR (this would be made harder with
    \async/\await); \cite{Qin2020Detector} currently supports only three
    concurrent data structures from Rust's ecosystem.
\end{description}

\section*{Acknowledgments}
We thank (in alphabetical order) Julia Gabet, Paul Kelly, Roland Kuhn, Nicolas
Lagaillardie, Julien Lange, and Fangyi Zhou, for their valuable proofreading,
feedback and advice on early versions of this work. We also thank the PPoPP
reviewers for the helpful comments in their reviews.

The work is supported by EPSRC EP/T006544/1, EP/K011715/1, EP/K034413/1,
EP/L00058X/1, EP/N027833/1, EP/N028201/1, EP/T014709/1, EP/V000462/1, and
NCSS/EPSRC VeTSS.

%\input{sections/case_study}

%\begin{acks}
%  We thank Nicolas Lagaillardie and Fangyi Zhou for their helpful comments and
%  suggestions. The work is supported by EPSRC, grants EP/T006544/1,
%  EP/K011715/1, EP/K034413/1, EP/L00058X/1, EP/N027833/1, EP/N028201/1,
%  EP/T014709/1, and EP/V000462/1 and by NCSS/EPSRC VeTSS.
%\end{acks}

\newpage

\appendix

\section{Artifact}

\subsection{Content of the artifact\label{content-of-the-artifact}}

The artifact~\cite{Artifact} contains the following files:

\begin{verbatim}
.
|-- Artifact.md
|-- Artifact.pdf
|-- Dockerfile
|-- gen_fig_6.sh
|-- gen_fig_7.sh
`-- getting_started.sh
\end{verbatim}

Note that an internet connection is required to use the artifact.

\subsection{Getting started guide (Docker artifact)\label{getting-started-guide-docker-artifact}}

In this subsection, we will run the benchmarks used to produce \cref{fig:RuntimeBenchmarks,fig:SubtypingBenchmarks}  of the paper in a Docker container. At the end of this section,
you should be have all the plots of those figures.

\begin{quote}
	Note that the benchmark results may not match those in
	\cref{fig:RuntimeBenchmarks,fig:SubtypingBenchmarks} when run in a
	Docker container or a machine with different specification than used in
	the paper. See the subsection on \emph{Claims supported or not by the
	artifact} for more discussion of this.
\end{quote}

To run the getting started guide, we assume you are running a Linux
machine with Docker and Gnuplot installed, as well as standard Unix
tools (\texttt{tail}, \texttt{awk}, \texttt{cut}, etc.).

This \emph{Getting started} subsection is fully automated: extract the
archive and run the \texttt{getting\_started.sh} script. \emph{The
script takes a long time to run all the benchmarks. On the author's
laptop, it took approximately 2 hours and 15 minutes.} When the script
finishes, you should have a few \texttt{*.png} plots which correspond to
the plots shown in \cref{fig:RuntimeBenchmarks,fig:SubtypingBenchmarks}  of the paper.

	\begin{lstlisting}[language=sh]
$ cd /path/to/extracted/artifact/
$ ./getting_started.sh
	\end{lstlisting}

Once run, to clean-up your system, in addition to removing the archive
folder, you should remove the docker image (the following command
assumes you don't have other docker images/containers on your system):

	\begin{lstlisting}[language=sh]
$ docker rmi rumpsteak_tool
$ docker system prune
	\end{lstlisting}

\subsubsection{Output}

The \texttt{getting\_started.sh} generates figures similar to the one used in
the paper\footnote{They are not exactly the same in the sense that they are not
produced by the same tool (Tikz vs. Gnuplot), but the data they use is produced
similarly.}.

We show in~\cref{fig:ex_fig} the figures \texttt{fft.png} and \texttt{ring.png}
that are the equivalent of \cref{fig:RuntimeBenchmarks} (subfigure~FFT) and
\cref{fig:SubtypingBenchmarks} (subfigure~Ring) in the paper.

\begin{figure}
	\begin{subfigure}[b]{.48\linewidth}
		\includegraphics[width=\textwidth]{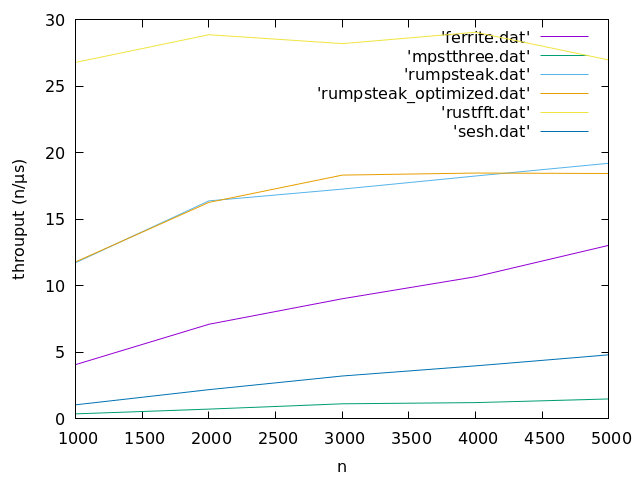}
		\caption{Figure~\texttt{fig\_6\_fft.png} generated by the artifact,
		which corresponds to Example~FFT in
		\cref{fig:RuntimeBenchmarks} in the paper.}
	\end{subfigure}
	\hfill
	\begin{subfigure}[b]{.48\linewidth}
		\includegraphics[width=\textwidth]{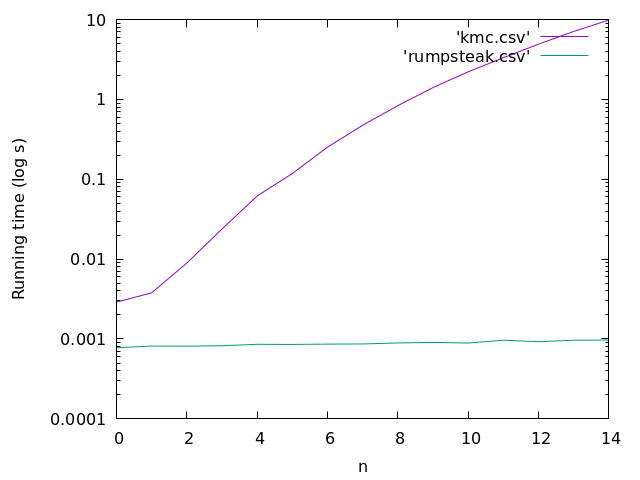}
		\caption{Figure~\texttt{fig\_7\_ring.png} generated by the artifact,
		which corresponds to Example~Ring in
		\cref{fig:SubtypingBenchmarks} in the paper.}
	\end{subfigure}
	\caption{Two examples of the figures generated. Note that those figure
	are obtain from a different run of the artifact, in a lower-end
	computer, which explains the differences with figures presented in the
	paper.}
	\label{fig:ex_fig}
\end{figure}

\subsubsection{Analysing the raw data\label{analysing-the-raw-data}}

\paragraph{Runtime benchmarks\label{a.-runtime-benchmarks}}

The results used to generate \cref{fig:RuntimeBenchmarks} are generated with
\href{https://crates.io/crates/criterion}{Criterion}. Criterion produces
the tree of files shown in \cref{fig:criterion}. The report can be viewed in a
web brower.

\begin{figure}
	\begin{center}
\begin{verbatim}
.
|-- double_buffering
|   |-- ...
|-- fft
|   |-- ...
|-- report
|   `-- index.html
`-- stream
    |-- ...
\end{verbatim}
\end{center}
	\caption{The files created by the Criterion benchmark (in
	\texttt{results/criterion/}). A detailed report is available by viewing
	\texttt{index.html} in a browser.}
	\label{fig:criterion}
\end{figure}

\paragraph{Subtyping benchmarks\label{b.-subtyping-benchmarks}}

Hyperfine generates a CSV file (located at \texttt{results/data/*/*.csv}) for the results
of each benchmark. For each row, this file contains

\begin{enumerate}
\def\labelenumi{\arabic{enumi}.}
\item
  \texttt{command}: the actual command that was run;
\item
  \texttt{mean}: the mean time to execute the command;
\item
  \texttt{stddev}: the standard deviation of executing the command;
\item
  \texttt{median}: the median time to execute the command;
\item
  \texttt{user}: the mean \emph{user} time to run the command;
\item
  \texttt{system}: the mean \emph{system} time to run the command;
\item
  \texttt{min}: the minimum time to run the command;
\item
  \texttt{max}: the maximum time to run the command; and
\item
  \texttt{parameter\_n}: the value of the parameter used in the command
\end{enumerate}

where all times are given in seconds. Since we do not care about
execution time spent in the kernel (all relevant computation is done in
user space) we take only the timings in the \texttt{user} column. We
provide a script to plot the mean user execution time against the
parameter value for each tool.

\subsection{Claims supported or not by the artifact\label{claims-supported-or-not-by-the-artifact}}

\subsubsection{Claims supported by the artifact\label{claims-supported-by-the-artifact}}

\begin{itemize}
\item
	Results presented in \cref{fig:SubtypingBenchmarks} are fairly stable across different
  machines, even when run in Docker. One should expect to obtain similar
  plots on their machine.
\end{itemize}

\subsubsection{Claims partially supported by the artifact\label{claims-partially-supported-by-the-artifact}}

\begin{itemize}
\item
	Results presented in \cref{fig:RuntimeBenchmarks} are quite dependent on the number of cores
  provided by the machine and on other programs being run on the machine
  simultaneously. The results presented in the paper are run on a
  16-core AMD Opteron 6200 Series at 2.6GHz with hyperthreading and
  128GB of RAM.

  Attempts run on lower-end processors or from within a Docker container
  may not provide similar results (in particular, the \texttt{rustfft}
  implementation is highly optimised for sequential execution and
  outperforms approaches based on message passing on processors with
  fewer cores). Nonetheless, the performance of Rumpsteak vs.~other MPST
  implementations would remain mostly comparable.
\end{itemize}

\subsubsection{Claims not supported by the artifact\label{claims-not-supported-by-the-artifact}}

\begin{itemize}
\item
  NuScr is an external tool not contributed by the authors, and
  therefore not part of the claims of the paper.
\end{itemize}

\iftoggle{full}{
	\renewcommand{\thefigure}{A.\arabic{figure}}

\section{Multiparty Asynchronous Subtyping}
\label{sec:AppendixSubtyping}

In the main paper, we mentioned a few definitions from~\cite{Ghilezan2021} that
we omitted due to space constraints. We explain these in the first section of
this appendix. In the later sections, we provide examples of the rules shown in
the paper, as well as proofs of the theorems stated. Finally, we provide more
details on its implementation.

\subsection{Synchronous session subtyping}

We first give the rules for \emph{synchronous session subtyping} given by
\citet{Chen2017}
in \cref{fig:SynchronousSubtyping}. The relation $\leq:$ on sorts
is defined as the least reflexive binary relation such that $\tnat \leq: \tint$
\cite{Ghilezan2021}.

\begin{figure}[h]
  \begin{prooftree}
    \alwaysDoubleLine
    \AxiomC{}
    \RightLabel{\ruleSub{end}}
    \UnaryInfC{$\tend \leq \tend$}
  \end{prooftree}
  \begin{prooftree}
    \alwaysDoubleLine
    \AxiomC{$\forall i \in I$}
    \AxiomC{$S_i \leq: S_i'$}
    \AxiomC{$\lt_i \leq \lt_i'$}
    \RightLabel{\ruleSub{bra}}
    \TrinaryInfC{$\tbra{i \in I \cup J} \tin{p} \ell_i(S_i) . \lt_i \leq \tbra{i \in I} \tin{p} \ell_i(S_i') . \lt_i'$}
  \end{prooftree}
  \begin{prooftree}
    \alwaysDoubleLine
    \AxiomC{$\forall i \in I$}
    \AxiomC{$S_i' \leq: S_i$}
    \AxiomC{$\lt_i \leq \lt_i'$}
    \RightLabel{\ruleSub{sel}}
    \TrinaryInfC{$\tsel{i \in I} \tout{p} \ell_i(S_i) . \lt_i \leq \tsel{i \in I \cup J} \tout{p} \ell_i(S_i') . \lt_i'$}
  \end{prooftree}
  \caption{Subtyping rules for synchronous session types.}
  \label{fig:SynchronousSubtyping}
\end{figure}

\subsection{Precise Asynchronous Multiparty Subtyping}

\begin{figure*}
  \begin{prooftree}
    \alwaysDoubleLine
    \AxiomC{}
    \RightLabel{\ruleRef{end}}
    \UnaryInfC{$\tend \lesssim \tend$}
    \DisplayProof
    \hspace{.2cm}
    \alwaysDoubleLine
    \AxiomC{$S' \leq: S$}
    \AxiomC{$W \lesssim W'$}
    \RightLabel{\ruleRef{in}}
    \BinaryInfC{$\tin{p}\ell(S).W \lesssim \tin{p}\ell(S').W'$}
    \DisplayProof
    \hspace{.2cm}
    \alwaysDoubleLine
    \AxiomC{$S \leq: S'$}
    \AxiomC{$W \lesssim W'$}
    \RightLabel{\ruleRef{out}}
    \BinaryInfC{$\tout{p}\ell(S).W \lesssim \tout{p}\ell(S').W'$}
  \end{prooftree}
  \begin{prooftree}
    \alwaysDoubleLine
    \AxiomC{$S' \leq: S$}
    \AxiomC{$W \lesssim \refa{p}.W'$}
    \AxiomC{$\act{W} = \act{\refa{p}.W'}$}
    \RightLabel{\ruleRefA}
    \TrinaryInfC{$\tin{p}\ell(S).W \lesssim \refa{p}.\tin{p}\ell(S').W'$}
  \end{prooftree}
  \begin{prooftree}
    \alwaysDoubleLine
    \AxiomC{$S \leq: S'$}
    \AxiomC{$W \lesssim \refb{p}.W'$}
    \AxiomC{$\act{W} = \act{\refb{p}.W'}$}
    \RightLabel{\ruleRefB}
    \TrinaryInfC{$\tout{p}\ell(S).W \lesssim \refb{p}.\tout{p}\ell(S').W'$}
  \end{prooftree}
  \caption{Tree refinement relation rules for asynchronous session type trees.}
  \label{fig:TreeRefinement}
\end{figure*}
Asynchronous subtyping is more complex as it allows the order of operations to
be swapped for efficiency.
\begin{figure}[h]
  \centering
  \begin{gather*}
    \act{\tend} = \nil \qquad \act{\tin{p}\ell(S).W} = \{\tin{p}\} \cup \act{W} \\
    \act{\tout{p}\ell(S).W} = \{\tout{p}\} \cup \act{W}
  \end{gather*}
  \caption{Definition of the function \act{W} on a tree $W$.}
  \label{fig:ActionsFunction}
\end{figure}
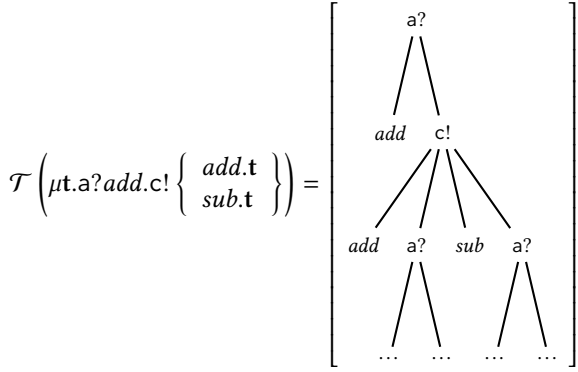
\begin{figure}[h]
  \centering
  \begin{equation*}
    \tr{\trec{t}{\tin{a}\mathit{add}.\tout{c}\left\{\begin{array}{l}
          \mathit{add}.\tvar{t} \\
          \mathit{sub}.\tvar{t}
        \end{array}\right\}}} = \left[
      \begin{tikzpicture}[thick,baseline={([yshift=-.5ex]current bounding box.center)},sibling distance=2em,style={font=\footnotesize}]
        \node {$\tin{a}$}
        child { node {$\mathit{add}$} }
        child { node {$\tout{c}$}
            child { node {$\mathit{add}$} }
            child { node {$\tin{a}$}
                child { node {\ldots} }
                child { node {\ldots} } }
            child { node {$\mathit{sub}$} }
            child { node {$\tin{a}$}
                child { node {\ldots} }
                child { node {\ldots} } } };
      \end{tikzpicture}
      \right]
  \end{equation*}
  \caption{An example of a session type and its corresponding type tree.}
  \label{fig:TypeTree}
\end{figure}
The \emph{tree refinement relation} $\lesssim$ is defined coinductively on
session types that have only single-inputs (SI) and single-outputs (SO). It is
specified for type trees, which are possibly infinite trees representing a
session type. An example of a type tree is given in \cref{fig:TypeTree} and the
tree refinement relation by \citet{Ghilezan2021} is given in
\cref{fig:TreeRefinement}. The function $\act{W}$, the set of input and output
actions in a tree $W$, is defined in \cref{fig:ActionsFunction}.

Single-input and single-output types are session types which do \emph{not}
include branching, i.e.\ a type generated from the grammar
(\citet{Ghilezan2021}) \(T ::= \tend\ |\ \tin{p}\ell.T\ |\ \tout{p}\ell.T\).

\begin{remark}
  As mentioned in \cite{Ghilezan2021}, checking the set of actions within
  \ruleRefA{} and \ruleRefB{} is important. If this were not included, then
  unsound recursive subtypes that ``forget'' some interactions would be allowed.
  \citet{Ghilezan2021} give the following example of a potential subtype that
  forgets to input an initial $\ell'$ message. If the \ruleRefA{} rule were
  allowed to be used then $T$ would incorrectly be a subtype of $T'$.
  \begin{samepage}
    \begin{gather*}
      T = \tr{\trec{t}{\tin{p}\ell.\tvar{t}}} \qquad T' = \tin{q}\ell'.T = \tin{q}\ell'.\tr{\trec{t}{\tin{p}\ell.\tvar{t}}}
    \end{gather*}
    \begin{prooftree}
      \alwaysDoubleLine
      \AxiomC{$T \leq \tin{q}\ell'.T'$}
      \RightLabel{\ruleRefA}
      \UnaryInfC{$T = \tin{p}\ell.T \leq \tin{q}\ell'.\tin{p}\ell.T = T'$}
    \end{prooftree}
    \smallskip
  \end{samepage}
\end{remark}
%Since $\lesssim$ is only defined for single-input (SI) and single-output (SO)
%type trees, a relation for all session type trees is given by
%\citet{Ghilezan2021} using the rule in \cref{fig:AsynchronousSubtyping}.
%$\so{T}$ is the minimal set of trees containing only single outputs which span
%the tree $T$. $\si{T'}$ is similarly defined as the minimal set of trees
%containing only single inputs which span the tree $T'$. Using existential
%quantifiers for $\si{U}$ and $\so{V'}$ allows external choices to be added and
%internal choices to be removed, as we observed in synchronous subtyping.

\subsubsection{Examples of asynchronous subtyping}
\label{subsec:subtypeex}
\paragraph{Ring protocol.} We show an example of the subtyping rules for a ring
protocol with choice. The projected and optimised local types are given by
$\lt'$ and $\lt$ respectively.
\begin{gather*}
  \lt' = \trec{t}{\tin{a}\mathit{add}.\tout{c}\left\{\begin{array}{l}
      \mathit{add}.\tvar{t} \\
      \mathit{sub}.\tvar{t}
    \end{array}\right\}} \qquad
  \lt = \trec{t}{\tout{c}\left\{\begin{array}{l}
      \mathit{add}.\tin{a}\mathit{add}.\tvar{t} \\
      \mathit{sub}.\tin{a}\mathit{add}.\tvar{t}
    \end{array}\right\}}
\end{gather*}
Considering the tree types $T = \tr{\lt}$ and $T' = \tr{\lt'}$, we must show
that $T \leq T'$ using the tree refinement definition from the main paper in
order to prove that $\lt \leq \lt'$.
\begin{gather*}
		\forall U \in \so{T} :
		\forall V' \in \si{T'} :
		\exists W \in \si{U} :
		\exists W' \in \so{V'} :\\
		W \lesssim W'
\end{gather*}
However, in our case, since $T$ and $T'$ are already SI trees, we can express
the definition more simply using only SO tree transformations.
\begin{equation*}
  \forall W \in \so{T} :
  \exists W' \in \so{T'} :
  W \lesssim W'
\end{equation*}
We  define the SO sets for each tree coinductively and use coinduction to show
that $W \lesssim W'$ in all cases.
\begin{gather*}
  \pi_1 = \tout{c}\mathit{add}.\tin{a}\mathit{add} \qquad
  \pi_2 = \tout{c}\mathit{sub}.\tin{a}\mathit{add} \qquad
  \pi_3 = \tin{a}\mathit{add} \\
  \forall \pi_1.W \in \so{T} : W \in \so{T} \qquad \forall \pi_2.W \in \so{T} : W \in \so{T} \\
  \forall \pi_3.\pi_1.W' \in \so{T'} : \pi_3.W' \in \so{T'} \\ \forall \pi_3.\pi_2.W' \in \so{T'} : \pi_3.W' \in \so{T'}
\end{gather*}
\begin{enumerate}
  \item Using the coinductive hypothesis $\pi_1.W \lesssim \pi_3.\pi_1.W'$, we
        show that $W \lesssim \pi_3.W'$.
        \begin{prooftree}
          \alwaysDoubleLine
          \AxiomC{$W \lesssim \tin{a}\mathit{add}.W' = \pi_3.W'$}
          \RightLabel{\ruleRef{in}}
          \UnaryInfC{$\tin{a}\mathit{add}.W \lesssim \tin{a}\mathit{add}.\tin{a}\mathit{add}.W'$}
          \RightLabel{\ruleRefB}
          \UnaryInfC{$\tout{c}\mathit{add}.\tin{a}\mathit{add}.W \lesssim \tin{a}\mathit{add}.\tout{c}\mathit{add}.\tin{a}\mathit{add}.W'$}
        \end{prooftree}

  \item Using the coinductive hypothesis $\pi_2.W \lesssim \pi_3.\pi_2.W'$, we
        show that $W \lesssim \pi_3.W'$.
        \begin{prooftree}
          \alwaysDoubleLine
          \AxiomC{$W \lesssim \tin{a}\mathit{add}.W' = \pi_3.W'$}
          \RightLabel{\ruleRef{in}}
          \UnaryInfC{$\tin{a}\mathit{add}.W \lesssim \tin{a}\mathit{add}.\tin{a}\mathit{add}.W'$}
          \RightLabel{\ruleRefB}
          \UnaryInfC{$\tout{c}\mathit{sub}.\tin{a}\mathit{add}.W \lesssim \tin{a}\mathit{add}.\tout{c}\mathit{sub}.\tin{a}\mathit{add}.W'$}
        \end{prooftree}
\end{enumerate}
\paragraph{Double buffering protocol.} We also show the optimisation to the
double buffering protocol discussed in \cref{sec:Overview}. The kernel sends
two $\mathit{ready}$ messages at once, allowing the source to fill up both
buffers sooner.
\begin{gather*}
  \lt = \tout{s}\mathit{ready}.\lt' = \tout{s}\mathit{ready}.\trec{x}{\tout{s}\mathit{ready}.\tin{s}\mathit{copy}.\tin{t}\mathit{ready}.\tout{t}\mathit{copy}.\tvar{x}} \\
  \lt' = \trec{x}{\tout{s}\mathit{ready}.\tin{s}\mathit{copy}.\tin{t}\mathit{ready}.\tout{t}\mathit{copy}.\tvar{x}}
\end{gather*}
As before, we consider the tree types $T = \tr{\lt}$ and $T' = \tr{\lt'}$ which
are already both SI and SO trees. Therefore, to prove that $\lt \leq \lt'$ we
need only to show that $T \lesssim T'$.
\begin{prooftree}
  \alwaysDoubleLine
  \AxiomC{$T' \leq W$}
  \RightLabel{\ruleRef{in}}
  \UnaryInfC{$\tout{t}\mathit{copy}.T' \leq \tout{t}\mathit{copy}.W$}
  \RightLabel{\ruleRef{in}}
  \UnaryInfC{$\tin{t}\mathit{ready}.\tout{t}\mathit{copy}.T' \leq \tin{t}\mathit{ready}.\tout{t}\mathit{copy}.W$}
  \RightLabel{\ruleRef{in}}
  \UnaryInfC{$\tin{s}\mathit{copy}.\tin{t}\mathit{ready}.\tout{t}\mathit{copy}.T' \leq \tin{s}\mathit{copy}.\tin{t}\mathit{ready}.\tout{t}\mathit{copy}.W$}
  \RightLabel{\ruleRefB}
  \UnaryInfC{$\tout{s}\mathit{ready}.W \leq \tin{s}\mathit{copy}.\tin{t}\mathit{ready}.\tout{t}\mathit{copy}.\tout{s}\mathit{ready}.W$}
  \RightLabel{\ruleRef{out}}
  \UnaryInfC{$T = \tout{s}\mathit{ready}.T' \leq \tout{s}\mathit{ready}.W = T'$}
\end{prooftree}
\begin{equation*}
  W = \tin{s}\mathit{copy}.\tin{t}\mathit{ready}.\tout{t}\mathit{copy}.T'
\end{equation*}

\subsection{Proofs for the Subtyping Algorithm}
\label{sec:Proofs}

\setcounter{theorem}{2}

\begin{lemma}
  Given finite prefixes $\pi$ and $\pi'$, $\pair{\pi}{\pi'}$ can be reduced only
  a finite number of times.
\end{lemma}

\begin{proof}
  We prove this in a similar style to \cite[Lemma~10]{Gay2005}, using a
  well-founded relation. We first consider the relation $R =
    \{(\pair{\pi_1}{\pi_1'}, \pair{\pi_2}{\pi_2'}) \mid \pair{\pi_2}{\pi_2'}
    \reduce \pair{\pi_1}{\pi_1'} \}$. We define the function $\fterms{\pi}$,
  which returns the number of terms in $\pi$, such that
  \begin{gather*}
    \fterms{\epsilon}        = 0                             \\
    \fterms{\tout{p}\ell(S)} = 1                             \qquad
    \fterms{\tin{p}\ell(S)}  = 1                             \\
    \fterms{\pi_1.\pi_2}     = \fterms{\pi_1} + \fterms{\pi_2}
  \end{gather*}
  and we define $\fterms{\pair{\pi}{\pi'}}$ such that
  $\fterms{\pair{\pi}{\pi'}} = \pair{\fterms{\pi}}{\fterms{\pi'}}$. We
  also define the lexicographical ordering
  \begin{align*}
     & \pair{m}{n} < \pair{m'}{n'}                                    \\
     & \quad \iff m < m' \textit{ or } (m = m' \textit{ and } n < n')
  \end{align*}
  We next show that reducing a pair of prefixes decrements the number of terms
  in the pair by induction over our reduction rules.
  \begin{equation}
    \label{eqn:PrefixReduction}
    \begin{aligned}
       & (\pair{\pi_1}{\pi_1'}, \pair{\pi_2}{\pi_2'}) \in R                           \\
       & \quad \implies \fterms{\pair{\pi_1}{\pi_1'}} < \fterms{\pair{\pi_2}{\pi_2'}}
    \end{aligned}
  \end{equation}
  \begin{itemize}
    \item Case \ruleRed{i}. We have that $\pi_2 = \tin{p}\ell(S).\pi_1$ and
          $\pi'_2 = \tin{p}\ell(S').\pi'_1$.

          Since $\pair{\pi_2}{\pi_2'} \reduce \pair{\pi_1}{\pi_1'}$,
	  we have that, by definition of \(R\), $(\pair{\pi_1}{\pi_1'},
	  \pair{\pi_2}{\pi_2'}) \in R$.

          As we know that $\fterms{\pi_2} = \fterms{\pi_1} + 1$ and
          $\fterms{\pi'_2} = \fterms{\pi'_1} + 1$, we also have that
          \[\fterms{\pair{\pi_1}{\pi_1'}} < \fterms{\pair{\pi_2}{\pi_2'}}\] as
          required.

    \item Case \ruleRed{o}. We have that $\pi_2 = \tout{p}\ell(S).\pi_1$
          and $\pi'_2 = \tout{p}\ell(S').\pi'_1$.

          Since $\pair{\pi_2}{\pi_2'} \reduce \pair{\pi_1}{\pi_1'}$, we have
	  that, by definition of $R$, $(\pair{\pi_1}{\pi_1'},
	  \pair{\pi_2}{\pi_2'}) \in R$.

          As we know that $\fterms{\pi_2} = \fterms{\pi_1} + 1$ and
          $\fterms{\pi'_2} = \fterms{\pi'_1} + 1$, we also have that
          \[\fterms{\pair{\pi_1}{\pi_1'}} < \fterms{\pair{\pi_2}{\pi_2'}}\] as
          required.

    \item Case \ruleRedA. We have that $\pi_2 = \tin{p}\ell(S).\pi_1$ and
          $\pi'_2 = \refa{p}.\tin{p}\ell(S').\pi'_1$.

          Since $\pair{\pi_2}{\pi_2'} \reduce \pair{\pi_1}{\refa{p}.\pi_1'}$,
	  by definition of $R$,
          we have that $(\pair{\pi_1}{\refa{p}.\pi_1'}, \pair{\pi_2}{\pi_2'}) \in R$
          .

          As we know $\fterms{\pi_2} = \fterms{\pi_1} + 1$ and $\fterms{\pi'_2}
            = \fterms{\refa{p}.\pi'_1} + 1$, we have that
          $\fterms{\pair{\pi_1}{\refa{p}.\pi_1'}} <
            \fterms{\pair{\pi_2}{\pi_2'}}$ as required.

    \item Case \ruleRedB. We have that $\pi_2 = \tout{p}\ell(S).\pi_1$ and
          $\pi'_2 = \refb{p}.\tout{p}\ell(S').\pi'_1$.

          Since $\pair{\pi_2}{\pi_2'} \reduce \pair{\pi_1}{\refb{p}.\pi_1'}$,
	  by definition of $R$,
          we have that $(\pair{\pi_1}{\refb{p}.\pi_1'}, \pair{\pi_2}{\pi_2'})\in R$
          .

          As we know $\fterms{\pi_2} = \fterms{\pi_1} + 1$ and $\fterms{\pi'_2}
            = \fterms{\refb{p}.\pi'_1} + 1$, we have that
          $\fterms{\pair{\pi_1}{\refb{p}.\pi_1'}} <
            \fterms{\pair{\pi_2}{\pi_2'}}$ as required.
  \end{itemize}
  The ordering of terms of pairs is well-founded since both components are
  bounded from below by 0. Therefore, from \cref{eqn:PrefixReduction}, $R$ is
  also well-founded and so pairs of prefixes cannot be reduced ad infinitum.
\end{proof}
\begin{figure*}
  \begin{minipage}{\textwidth}
    \centering
    \begin{prooftree}
      \AxiomC{$\Sigma_1\map{\pair{\tin{p}\ell}{\lt} \leq \pair{\tin{q}\ell'}{\lt}} = \tin{p}\ell$}
      \AxiomC{$\act{\tin{p}\ell.\tend} \not\supseteq \act{\tin{q}\ell'.\tend}$}
      \RightLabel{\ruleAlg{asm}}
      \BinaryInfC{$\tin{p}\ell.\tin{p}\ell; \Sigma_3 \vdash \triple{\tin{p}\ell}{\lt}{0} \leq \triple{\tin{q}\ell'}{\lt}{1}$}
      \RightLabel{\ruleAlg{sub}}
      \UnaryInfC{$\tin{p}\ell.\tin{p}\ell; \Sigma_3 \vdash \triple{\tin{p}\ell.\tin{p}\ell}{\lt}{0} \leq \triple{\tin{q}\ell'.\tin{p}\ell}{\lt}{1}$}
      \RightLabel{\ruleAlg{ii}}
      \UnaryInfC{$\tin{p}\ell; \Sigma_3 \vdash \triple{\tin{p}\ell}{\tin{p}\ell.\lt}{0} \leq \triple{\tin{q}\ell'}{\tin{p}\ell.\lt}{1}$}
      \RightLabel{\ruleAlg{\(\mu\)r}}
      \UnaryInfC{$\tin{p}\ell; \Sigma_2 \vdash \triple{\tin{p}\ell}{\tin{p}\ell.\lt}{1} \leq \triple{\tin{q}\ell'}{\lt}{2}$}
      \RightLabel{\ruleAlg{\(\mu\)l}}
      \UnaryInfC{$\tin{p}\ell; \Sigma_1 \vdash \triple{\tin{p}\ell}{\lt}{1} \leq \triple{\tin{q}\ell'}{\lt}{2}$}
      \RightLabel{\ruleAlg{ii}}
      \UnaryInfC{$\epsilon; \Sigma_1 \vdash \triple{\epsilon}{\tin{p}\ell.\lt}{1} \leq \triple{\epsilon}{\tin{q}\ell'.\lt}{2}$}
      \RightLabel{\ruleAlg{\(\mu\)l}}
      \UnaryInfC{$\epsilon; \nil \vdash \triple{\epsilon}{\lt}{2} \leq \triple{\epsilon}{\tin{q}\ell'.\lt}{2}$}
    \end{prooftree}
    \begin{gather*}
      \Sigma_1 = \map{\pair{\epsilon}{\lt} \leq \pair{\epsilon}{\lt'} \mapsto \epsilon}                              \quad
      \Sigma_2 = \Sigma_1\map{\pair{\tin{p}\ell}{\lt} \leq \pair{\tin{q}\ell'}{\lt} \mapsto \tin{p}\ell}             \\
      \Sigma_3 = \Sigma_2\map{\pair{\tin{p}\ell}{\tin{p}\ell.\lt} \leq \pair{\tin{q}\ell'}{\lt} \mapsto \tin{p}\ell}
    \end{gather*}
  \end{minipage}
  \caption{Demonstration of how our algorithm correctly prevents actions from being forgotten.}
  \label{fig:AlgorithmForgottenActions}
\end{figure*}

\begin{theorem}[Termination]
  Our subtyping algorithm always eventually terminates.
\end{theorem}

\begin{proof}
  We prove termination by arguing that each of our subtyping algorithm rules
  can be run only a finite number of times.
  \begin{itemize}
    \item \ruleAlg{end} and \ruleAlg{asm} can each be run only once since
          they are terminating rules.

    \item \ruleAlg{sub} can be run only a finite number of times since a
          pair of prefixes can be reduced only a finite number of times, as
          proven in Lemma~3.

    \item \ruleAlg{oi}, \ruleAlg{oo}, \ruleAlg{ii} and
          \ruleAlg{io} can be run only a finite number of times since
          \begin{enumerate*}[label=\textbf{(\arabic*)}]
            \item the number of terms in $\lt$ and $\lt'$ is finite and
            \item recursion must be explicitly unrolled with
            \ruleAlg{\(\mu\)l} or \ruleAlg{\(\mu\)r}, which themselves can only
            be run a finite number of times.
          \end{enumerate*}

    \item \ruleAlg{\(\mu\)l} and \ruleAlg{\(\mu\)r} can be run only a finite
          number of times since
          \begin{enumerate*}[label=\textbf{(\arabic*)}]
            \item the bounds $n$ and $n'$ are finite;
            \item each execution of the rule decrements $n$ or $n'$
            respectively; and
            \item no rule allows $n$ or $n'$ to be incremented.
          \end{enumerate*}

    \item \ruleAlg{tra} can be run only a finite number of times since
          \begin{enumerate*}[label=\textbf{(\arabic*)}]
            \item the bound $k$ is finite;
            \item each execution of the rule decrements $k$; and
            \item no rule allows $k$ to be incremented.
          \end{enumerate*}
  \end{itemize}
\end{proof}

\begin{theorem}[Soundness]
  Our subtyping algorithm is sound.
\end{theorem}

\begin{proof}
  To prove that our algorithm is sound, we must show that each rule in the
  precise subtyping by \cite{Ghilezan2021} is matched by a rule in our
  algorithm.
  \begin{itemize}
    \item The subtyping relation rule is
          matched by \ruleAlg{oi}, \ruleAlg{oo}, \ruleAlg{ii} and
          \ruleAlg{io}.

    \item \ruleRef{end} is matched by \ruleAlg{end}. Since both prefixes are
          empty and $\lt = \lt' = \tend$, we trivially have that $\pi.\lt
            \leq \pi'.\lt'$.

    \item The coinductive behaviour of the refinement relation rules is
          matched by \ruleAlg{asm}. From our map of assumptions $\Sigma$, we
          have that $\pi.\lt \leq \pi'.\lt'$.

    \item \ruleRef{in} and \ruleRef{out} are matched by \ruleRed{i} and
          \ruleRed{o} respectively, which can be applied with
          \ruleAlg{sub}.

    \item \ruleRefA{} and \ruleRefB{} are matched by \ruleRedA{} and
          \ruleRedB{} respectively, which can be applied with \ruleAlg{sub}.
          To ensure that actions are not forgotten when using coinduction,
          the refinement relation rules additionally require that the
          actions of the resulting type trees are equal.
          Our reduction rules emit
          this check since they deal only with finite sequences. We instead
          prevent actions from being forgotten in our \ruleAlg{asm} rule by
          ensuring that each reordered action in the supertype (which is
          contained in $\pi'$) is encountered in the recursive part of the
          subtype ($\rho'$) by using the subset relation.

          \citet{Ghilezan2021} use the example of checking $\lt =
            \trec{t}{\tin{p}\ell.\tvar{t}} \leq \tin{q}\ell'.\lt = \lt'$
          to demonstrate why comparing the actions is necessary---we
          show how our algorithm also correctly rejects this subtype in
          \cref{fig:AlgorithmForgottenActions} (with a bound of 2 for
          brevity).

          We cannot apply \ruleAlg{asm} as the final rule since
          \[\act{\tin{p}\ell.\tend} = \{\tin{p}\} \not\supseteq \{\tin{q}\} =
            \act{\tin{q}\ell'.\tend}.\] This subset check spots that the
          $\tin{q}$ action was not present in the recursive part of the
          supposed subtype. Otherwise, our algorithm would incorrectly
          conclude that $\lt \leq \lt'$.
  \end{itemize}
  It is straightforward to also argue that our algorithm preserves
  reflexivity. We have that $\triple{\epsilon}{\lt}{n} \leq
    \triple{\epsilon}{\lt}{n}$, providing $n$ is sufficiently large to ensure
  that each recursion can be visited at least once. The prefixes of both sides
  will always be identical so will reduce to $\pair{\epsilon}{\epsilon}$ using
  \ruleRed{i} or \ruleRed{o}. These reductions will be applied using
  \ruleAlg{sub} until either $\tend$ is encountered and \ruleAlg{end} can be
  applied or the algorithm loops, which allows the application of
  \ruleAlg{asm}.
\end{proof}

\begin{lemma}
  Given finite prefixes $\pi$ and $\pi'$, the time complexity of reducing
  $\pair{\pi}{\pi'}$ is $\oh(\min(\tsize{\pi}, \tsize{\pi'}))$.
\end{lemma}

\begin{proof}
  We prove this quite simply by induction over our reduction rules.
  \begin{itemize}
    \item Inductive case $\pair{\tin{p}\ell.\pi}{\tin{p}\ell.\pi'}$. Then, we
          can perform a \ruleRed{i} reduction to get $\pair{\pi}{\pi'}$. By the
          inductive hypothesis, the complexity of reducing $\pair{\pi}{\pi'}$ is
          $\oh(\min(\tsize{\pi}, \tsize{\pi'}))$. Therefore, since we require
          one additional reduction step, the complexity of reducing
          $\pair{\tin{p}\ell.\pi}{\tin{p}\ell.\pi'}$ is $\oh(\min(\tsize{\pi},
              \tsize{\pi'}) + 1) = \oh(\min(\tsize{\pi} + 1, \tsize{\pi'} + 1)) =
            \oh(\min(\tsize{\tin{p}\ell.\pi}, \tsize{\tin{p}\ell.\pi'}))$ as
          required.

    \item Inductive case $\pair{\tout{p}\ell.\pi}{\tout{p}\ell.\pi'}$. The proof
          is the same as for the $\pair{\tin{p}\ell.\pi}{\tin{p}\ell.\pi'}$
          case, except a \ruleRed{o} reduction is applied.

    \item Inductive case $\pair{\tin{p}\ell.\pi}{\refa{p}.\tin{p}\ell.\pi'}$.
          Then, we can perform a \ruleRedA{} reduction to get
          $\pair{\pi}{\refa{p}.\pi'}$. By the inductive hypothesis, the
          complexity of reducing $\pair{\pi}{\refa{p}.\pi'}$ is
          $\oh(\min(\tsize{\pi}, \tsize{\refa{p}.\pi'}))$. Therefore, since we
          require one additional reduction step, the complexity of reducing
          $\pair{\tin{p}\ell.\pi}{\refa{p}.\tin{p}\ell.\pi'}$ is
          $\oh(\min(\tsize{\pi}, \tsize{\refa{p}.\pi'}) + 1) =
            \oh(\min(\tsize{\pi} + 1, \tsize{\refa{p}.\pi'} + 1)) =
            \oh(\min(\tsize{\tin{p}\ell.\pi}, \tsize{\refa{p}.\tin{p}\ell.\pi'}))$
          as required.

    \item Inductive case $\pair{\tout{p}\ell.\pi}{\refb{p}.\tout{p}\ell.\pi'}$.
          The proof is the same as for the
          $\pair{\tin{p}\ell.\pi}{\refa{p}.\tin{p}\ell.\pi'}$ case, except a
          \ruleRedB{} reduction is applied.

    \item Base case $\pair{\pi}{\pi'}$ where $\pair{\pi}{\pi'}$ cannot be
          reduced. Since $\forall \pi . \tsize{\pi} > 0$, the complexity of
          reducing $\pair{\pi}{\pi'}$ is $\oh(0) = \oh(\min(\tsize{\pi},
              \tsize{\pi'}))$ as required.
  \end{itemize}
\end{proof}

\begin{theorem}[Complexity]
  Consider $\lt$ and $\lt'$ as (possibly infinite) trees $\tr{\lt}$ and
  $\tr{\lt'}$ with \emph{asymptotic branching factors} $b$ and $b'$
  respectively~\cite{Edelkamp1998,Korf1985}. Our algorithm has time complexity
  $\oh(n\min(b, b')^n)$ and space complexity $\oh(n\min(b, b'))$ in the worst
  case to determine if $\lt \leq \lt'$ with bound $n$.
\end{theorem}

\begin{proof}
  Let us consider for now only the left tree $\tr{\lt}$ which has asymptotic
  branching factor $b$. In the worst case, the number of nodes we have to
  explore in the tree is
  \begin{equation*}
    1 + b + b^2 + \ldots + b^{n - 1} = \sum_{i = 0}^{n - 1} b^i = \frac{b^n - 1}{b - 1}
  \end{equation*}
  Therefore, exploring the left tree up to a depth of $n$ has time complexity
  $\oh(b^n)$. Similarly, exploring the right tree $\tr{\lt'}$ to depth $n$ has
  complexity $\oh(b'^n)$. Note that in the worst case we cannot reduce any
  prefixes as we go along and must therefore do the entire reduction at the end
  of each exploration path.

  Suppose at the end of some exploration path we have the pair of prefixes
  $\pair{\pi}{\pi'}$. Since $\tsize{\pi} = \tsize{\pi'}$ (our algorithm does not
  add an uneven number of terms to either side of the prefix pair), the time
  complexity of reducing this pair is $\oh(\min(\tsize{\pi}, \tsize{\pi'})) =
    \oh(\tsize{\pi})$ from \cref{thm:PrefixesComplexity}.

  At the end of our exploration, we will have $b^{n - 1}$ prefixes, each with
  size $n$. Therefore, the complexity of reducing all the pairs of prefixes is
  $\oh(nb^{n - 1})$ so the total complexity for the exploration and reduction is
  $\oh(b^n + nb^{n - 1}) = \oh(nb^n)$.

  Our algorithm stops when the bound $n$ is reached in
  \emph{either} of the two trees so the overall time complexity of exploration
  is $\oh(\min(nb^n, nb'^n)) = \oh(n\min(b, b')^n)$.

  Considering space complexity, it is clear that the greatest amount of memory
  will be required at the end of an exploration path (when the prefixes are
  greatest in length). For the left tree, at this point, we will need to store
  \begin{enumerate*}[label=\textbf{(\arabic*)}]
    \item a prefix of length $n$ (since we are considering the worst case); and
    \item the other $b$ siblings to visit at $n - 1$ levels.
  \end{enumerate*}

  Therefore, the total space complexity for the left tree is $\oh(n + b(n - 1))
    = \oh(nb)$ and the space complexity for exploring both trees is
  $\oh(\min(nb, nb')) = \oh(n\min(b, b'))$.
\end{proof}

\subsection{Algorithm Examples}
\label{subsec:algorthmex}

\paragraph{Ring protocol.} We again use the ring protocol with choice and show
that our algorithm can successfully check the optimisation to \ppt{b}. The
derivation trees are shown in \cref{fig:app:derivation_ring}.
\begingroup
\small
\begin{gather*}
  \lt = \trec{t}{\tout{c}\left\{\begin{array}{l}
      \mathit{add}.\tin{a}\mathit{add}.\tvar{t} \\
      \mathit{sub}.\tin{a}\mathit{add}.\tvar{t}
    \end{array}\right\}} \qquad
  \lt' = \trec{t}{\tin{a}\mathit{add}.\tout{c}\left\{\begin{array}{l}
      \mathit{add}.\tvar{t} \\
      \mathit{sub}.\tvar{t}
    \end{array}\right\}}
\end{gather*}
\endgroup
\begin{figure*}
\begin{prooftree}
  \AxiomC{$(\star)$}
  \AxiomC{$\rho_3; \Sigma_2 \vdash \triple{\tout{c} \mathit{add} . \tin{a} \mathit{add}}{\lt}{0} \leq \triple{\tin{a} \mathit{add} . \tout{c}  \mathit{sub}}{\lt'}{0}$}
  \RightLabel{\ruleAlg{io}}
  \BinaryInfC{$\rho_1; \Sigma_2 \vdash \triple{\tout{c} \mathit{add}}{\tin{a} \mathit{add} . \lt}{0} \leq \triple{\tin{a} \mathit{add}}{\tout{c}\left\{\begin{array}{l}
          \mathit{add}.\lt' \\
          \mathit{sub}.\lt'
        \end{array}\right\}}{0}$}
  \AxiomC{$(\dagger)$}
  \AxiomC{$\rho_4; \Sigma_2 \vdash \triple{\tout{c} \mathit{sub} . \tin{a} \mathit{add}}{\lt}{0} \leq \triple{\tin{a} \mathit{add} . \tout{c}  \mathit{add}}{\lt'}{0}$}
  \RightLabel{\ruleAlg{io}}
  \BinaryInfC{$\rho_2; \Sigma_2 \vdash \triple{\tout{c} \mathit{sub}}{\tin{a} \mathit{add} . \lt}{0} \leq \triple{\tin{a} \mathit{add}}{\tout{c}\left\{\begin{array}{l}
          \mathit{add}.\lt' \\
          \mathit{sub}.\lt'
        \end{array}\right\}}{0}$}
  \RightLabel{\ruleAlg{oi}}
  \BinaryInfC{$\epsilon; \Sigma_2 \vdash \triple{\epsilon}{\tout{c}\left\{\begin{array}{l}
          \mathit{add}.\tin{a}\mathit{add}.\lt \\
          \mathit{sub}.\tin{a}\mathit{add}.\lt
        \end{array}\right\}}{0} \leq \triple{\epsilon}{\tin{a}\mathit{add}.\tout{c}\left\{\begin{array}{l}
          \mathit{add}.\lt' \\
          \mathit{sub}.\lt'
        \end{array}\right\}}{0}$}
  \RightLabel{\ruleAlg{\(\mu\)r}}
  \UnaryInfC{$\epsilon; \Sigma_1 \vdash \triple{\epsilon}{\tout{c}\left\{\begin{array}{l}
          \mathit{add}.\tin{a}\mathit{add}.\lt \\
          \mathit{sub}.\tin{a}\mathit{add}.\lt
        \end{array}\right\}}{0} \leq \triple{\epsilon}{\lt'}{1}$}
  \RightLabel{\ruleAlg{\(\mu\)l}}
  \UnaryInfC{$\epsilon; \nil \vdash \triple{\epsilon}{\lt}{1} \leq \triple{\epsilon}{\lt'}{1}$}
\end{prooftree}
\begin{equation*}
  (\star)\ =\
  \AxiomC{}
  \RightLabel{\ruleRedB}
  \UnaryInfC{$\pair{\tout{c} \mathit{add} . \tin{a} \mathit{add}}{\tin{a} \mathit{add} . \tout{c}  \mathit{add}} \reduce \pair{\tin{a} \mathit{add}}{\tin{a} \mathit{add}}$}
  \AxiomC{}
  \RightLabel{\ruleRed{i}}
  \UnaryInfC{$\pair{\tin{a} \mathit{add}}{\tin{a} \mathit{add}} \reduce \pair{\epsilon}{\epsilon}$}
  \AxiomC{$\act{\rho_3.\tend} \supseteq \act{\tend}$}
  \RightLabel{\ruleAlg{asm}}
  \UnaryInfC{$\rho_3; \Sigma_2 \vdash \triple{\epsilon}{\lt}{0} \leq \triple{\epsilon}{\lt'}{0}$}
  \RightLabel{\ruleAlg{sub}}
  \BinaryInfC{$\rho_3; \Sigma_2 \vdash \triple{\tin{a} \mathit{add}}{\lt}{0} \leq \triple{\tin{a} \mathit{add}}{\lt'}{0}$}
  \RightLabel{\ruleAlg{sub}}
  \BinaryInfC{$\rho_3; \Sigma_2 \vdash \triple{\tout{c} \mathit{add} . \tin{a} \mathit{add}}{\lt}{0} \leq \triple{\tin{a} \mathit{add} . \tout{c}  \mathit{add}}{\lt'}{0}$}
  \DisplayProof
\end{equation*}
\begin{equation*}
  (\dagger)\ =\
  \AxiomC{}
  \RightLabel{\ruleRedB}
  \UnaryInfC{$\pair{\tout{c} \mathit{sub} . \tin{a} \mathit{add}}{\tin{a} \mathit{add} . \tout{c}  \mathit{sub}} \reduce \pair{\tin{a} \mathit{add}}{\tin{a} \mathit{add}}$}
  \AxiomC{}
  \RightLabel{\ruleRed{i}}
  \UnaryInfC{$\pair{\tin{a} \mathit{add}}{\tin{a} \mathit{add}} \reduce \pair{\epsilon}{\epsilon}$}
  \AxiomC{$\act{\rho_4.\tend} \supseteq \act{\tend}$}
  \RightLabel{\ruleAlg{asm}}
  \UnaryInfC{$\rho_4; \Sigma_2 \vdash \triple{\epsilon}{\lt}{0} \leq \triple{\epsilon}{\lt'}{0}$}
  \RightLabel{\ruleAlg{sub}}
  \BinaryInfC{$\rho_4; \Sigma_2 \vdash \triple{\tin{a} \mathit{add}}{\lt}{0} \leq \triple{\tin{a} \mathit{add}}{\lt'}{0}$}
  \RightLabel{\ruleAlg{sub}}
  \BinaryInfC{$\rho_4; \Sigma_2 \vdash \triple{\tout{c} \mathit{sub} . \tin{a} \mathit{add}}{\lt}{0} \leq \triple{\tin{a} \mathit{add} . \tout{c}  \mathit{sub}}{\lt'}{0}$}
  \DisplayProof
\end{equation*}
\begin{gather*}
  \rho_1 = \tout{c} \mathit{add} \qquad
  \rho_2 = \tout{c} \mathit{sub} \qquad
  \rho_3 = \rho_1 . \tin{a} \mathit{add} \qquad
  \rho_4 = \rho_2 . \tin{a} \mathit{add} \\
  \Sigma_1 = \map{\pair{\epsilon}{\lt} \leq \pair{\epsilon}{\lt'} \mapsto \epsilon} \qquad
  \Sigma_2 = \Sigma_1\map{\pair{\epsilon}{\tout{c}\left\{\begin{array}{l}
        \mathit{add}.\tin{a}\mathit{add}.\lt \\
        \mathit{sub}.\tin{a}\mathit{add}.\lt
      \end{array}\right\}} \leq \pair{\epsilon}{\lt'} \mapsto \epsilon}
\end{gather*}
\caption{Derivation trees to verify the subtyping of the Ring protocol}
\label{fig:app:derivation_ring}
\end{figure*}

\paragraph{Alternating bit protocol.} We consider the alternating bit
protocol \cite{AlternatingBit}. We
construct a global type $\gt$ for the protocol such that when projected onto
the receiver, its local type matches the protocol specification.

\begingroup
\small
\begin{equation*}
  \gt = \gtrec{t}\gtmsg{s}{r}{\mathit{d0} . \gtmsg{r}{s}{\begin{array}{l}
        \mathit{a0} . \gtrec{u}\gtmsg{s}{r}{\mathit{d1} . \gtmsg{r}{s}{\begin{array}{l}
              \mathit{a0} . \gtvar{u} \\
              \mathit{a1} . \gtvar{t}
            \end{array}}} \\
        \mathit{a1} . \gtvar{t}
      \end{array}}}
\end{equation*}
\begin{equation*}
  \gtproj{\gt}{s} = \trec{t}{\tout{r} \mathit{d0} . \tin{r} \left\{\begin{array}{l}
      \mathit{a0} . \trec{x}{\tout{r} \mathit{d1} . \tin{r} \left\{\begin{array}{l}
          \mathit{a0} . \tvar{x} \\
          \mathit{a1} . \tvar{t}
        \end{array}\right\}} \\
      \mathit{a1} . \tvar{t}
    \end{array}\right\}} \qquad
\end{equation*}
\begin{equation*}
  \gtproj{\gt}{r} = \trec{t}{\tin{s} \mathit{d0} . \tout{s} \left\{\begin{array}{l}
      \mathit{a0} . \trec{x}{\tin{s} \mathit{d1} . \tout{s} \left\{\begin{array}{l}
          \mathit{a0} . \tvar{x} \\
          \mathit{a1} . \tvar{t}
        \end{array}\right\}} \\
      \mathit{a1} . \tvar{t}
    \end{array}\right\}}
\end{equation*}
\begin{gather*}
  \lt = \trec{t}{\tin{s} \left\{\begin{array}{l}
      \mathit{d0} . \tout{s} \mathit{a0} . \tvar{t} \\
      \mathit{d1} . \tout{s} \mathit{a1} . \tvar{t}
    \end{array}\right\}} \qquad
  \lt' = \gtproj{\gt}{r} \\
\end{gather*}
\endgroup
We then use our subtyping algorithm to confirm that the type given by the
protocol specification for the receiver \cite{AlternatingBit} is a subtype
of its projected version. In this derivation, we omit some exploration paths
for brevity. The derivation tree is in~\cref{fig:app:derivation_AB}.

\begin{figure*}
\begingroup
\small
\begin{prooftree}
  \AxiomC{$\act{\tin{s} \mathit{d0} . \tout{s} \mathit{a0} . \tend} \supseteq \act{\tend}$}
  \RightLabel{\ruleAlg{asm}}
  \UnaryInfC{$\tin{s} \mathit{d0} . \tout{s} \mathit{a0}; \Sigma_4 \vdash \triple{\epsilon}{\lt}{0} \leq \triple{\epsilon}{\lt'}{0}$}
  \RightLabel{\ruleAlg{sub}}
  \UnaryInfC{$\tin{s} \mathit{d0} . \tout{s} \mathit{a0}; \Sigma_4 \vdash \triple{\tout{s} \mathit{a1}}{\lt}{0} \leq \triple{\tout{s} \mathit{a1}}{\lt'}{0}$}
  \AxiomC{\ldots}
  \RightLabel{\ruleAlg{out-out}}
  \BinaryInfC{$\tin{s} \mathit{d0} . \tout{s} \mathit{a0}; \Sigma_4 \vdash \triple{\epsilon}{\tout{s} \mathit{a1} . \lt}{0} \leq \triple{\epsilon}{\tout{s} \left\{\begin{array}{l}
          \mathit{a0} . \lt_1 \\
          \mathit{a1} . \lt'
        \end{array}\right\}}{0}$}
  \RightLabel{\ruleAlg{sub}}
  \UnaryInfC{$\tin{s} \mathit{d0} . \tout{s} \mathit{a0}; \Sigma_4 \vdash \triple{\tin{s} \mathit{d1}}{\tout{s} \mathit{a1} . \lt}{0} \leq \triple{\tin{s} \mathit{d1}}{\tout{s} \left\{\begin{array}{l}
          \mathit{a0} . \lt_1 \\
          \mathit{a1} . \lt'
        \end{array}\right\}}{0}$}
  \AxiomC{\ldots}
  \RightLabel{\ruleAlg{in-in}}
  \BinaryInfC{$\tin{s} \mathit{d0} . \tout{s} \mathit{a0}; \Sigma_4 \vdash \triple{\epsilon}{\tin{s} \left\{\begin{array}{l}
          \mathit{d0} . \tout{s} \mathit{a0} . \lt \\
          \mathit{d1} . \tout{s} \mathit{a1} . \lt
        \end{array}\right\}}{0} \leq \triple{\epsilon}{\tin{s} \mathit{d1} . \tout{s} \left\{\begin{array}{l}
          \mathit{a0} . \lt_1 \\
          \mathit{a1} . \lt'
        \end{array}\right\}}{0}$}
  \RightLabel{\ruleAlg{\(\mu\)r}}
  \UnaryInfC{$\tin{s} \mathit{d0} . \tout{s} \mathit{a0}; \Sigma_3 \vdash \triple{\epsilon}{\tin{s} \left\{\begin{array}{l}
          \mathit{d0} . \tout{s} \mathit{a0} . \lt \\
          \mathit{d1} . \tout{s} \mathit{a1} . \lt
        \end{array}\right\}}{0} \leq \triple{\epsilon}{\lt_1}{1}$}
  \RightLabel{\ruleAlg{\(\mu\)l}}
  \UnaryInfC{$\tin{s} \mathit{d0} . \tout{s} \mathit{a0}; \Sigma_2 \vdash \triple{\epsilon}{\lt}{1} \leq \triple{\epsilon}{\lt_1}{1}$}
  \RightLabel{\ruleAlg{sub}}
  \UnaryInfC{$\tin{s} \mathit{d0} . \tout{s} \mathit{a0}; \Sigma_2 \vdash \triple{\tout{s} \mathit{a0}}{\lt}{1} \leq \triple{\tout{s} \mathit{a0}}{\lt_1}{1}$}
  \AxiomC{\ldots}
  \RightLabel{\ruleAlg{out-out}}
  \BinaryInfC{$\tin{s} \mathit{d0}; \Sigma_2\vdash \triple{\epsilon}{\tout{s} \mathit{a0} . \lt}{1} \leq \triple{\epsilon}{\tout{s} \left\{\begin{array}{l}
          \mathit{a0} . \lt_1 \\
          \mathit{a1} . \lt'
        \end{array}\right\}}{1}$}
  \RightLabel{\ruleAlg{sub}}
  \UnaryInfC{$\tin{s} \mathit{d0}; \Sigma_2 \vdash \triple{\tin{s} \mathit{d0}}{\tout{s} \mathit{a0} . \lt}{1} \leq \triple{\tin{s} \mathit{d0}}{\tout{s} \left\{\begin{array}{l}
          \mathit{a0} . \lt_1 \\
          \mathit{a1} . \lt'
        \end{array}\right\}}{1}$}
  \AxiomC{\ldots}
  \RightLabel{\ruleAlg{in-in}}
  \BinaryInfC{$\epsilon; \Sigma_2 \vdash \triple{\epsilon}{\tin{s} \left\{\begin{array}{l}
          \mathit{d0} . \tout{s} \mathit{a0} . \lt \\
          \mathit{d1} . \tout{s} \mathit{a1} . \lt
        \end{array}\right\}}{1} \leq \triple{\epsilon}{\tin{s} \mathit{d0} . \tout{s} \left\{\begin{array}{l}
          \mathit{a0} . \lt_1 \\
          \mathit{a1} . \lt'
        \end{array}\right\}}{1}$}
  \RightLabel{\ruleAlg{\(\mu\)r}}
  \UnaryInfC{$\epsilon; \Sigma_1 \vdash \triple{\epsilon}{\tin{s} \left\{\begin{array}{l}
          \mathit{d0} . \tout{s} \mathit{a0} . \lt \\
          \mathit{d1} . \tout{s} \mathit{a1} . \lt
        \end{array}\right\}}{1} \leq \triple{\epsilon}{\lt'}{2}$}
  \RightLabel{\ruleAlg{\(\mu\)l}}
  \UnaryInfC{$\epsilon; \nil \vdash \triple{\epsilon}{\lt}{2} \leq \triple{\epsilon}{\lt'}{2}$}
\end{prooftree}
\begin{gather*}
  \lt_1 = \trec{x}{\tin{s} \mathit{d1} . \tout{s} \left\{\begin{array}{l}
      \mathit{a0} . \tvar{x} \\
      \mathit{a1} . \lt'
    \end{array}\right\}} \\
  \Sigma_1 = \map{\pair{\epsilon}{\lt} \leq \pair{\epsilon}{\lt'} \mapsto \epsilon} \qquad
  \Sigma_2 = \Sigma_1\map{\pair{\epsilon}{\tin{s} \left\{\begin{array}{l}
        \mathit{d0} . \tout{s} \mathit{a0} . \lt \\
        \mathit{d1} . \tout{s} \mathit{a1} . \lt
      \end{array}\right\}} \leq \pair{\epsilon}{\lt'} \mapsto \epsilon} \\
  \Sigma_3 = \Sigma_2\map{\pair{\epsilon}{\lt} \leq \pair{\epsilon}{\lt_1} \mapsto \tin{s} \mathit{d0} . \tout{s} \mathit{a0}} \qquad
  \Sigma_4 = \Sigma_3\map{\pair{\epsilon}{\tin{s} \left\{\begin{array}{l}
        \mathit{d0} . \tout{s} \mathit{a0} . \lt \\
        \mathit{d1} . \tout{s} \mathit{a1} . \lt
      \end{array}\right\}} \leq \pair{\epsilon}{\lt_1} \mapsto \tin{s} \mathit{d0} . \tout{s} \mathit{a0}}
\end{gather*}
\endgroup
\caption{Derivation trees to verify the subtyping of the Alternating-Bit protocol}
\label{fig:app:derivation_AB}
\end{figure*}

\subsection{Implementation of the Algorithm}
\label{subsec:impalg}
In practice, we implement our asynchronous subtyping algorithm on FSMs $\efsm$
and $\oefsm$ rather than local types $\lt$ and $\lt'$. We discuss the practical
considerations behind some of our implementation decisions and explain why these
are equivalent to the theory presented in \cref{sec:PreciseSubtyping}.

\paragraph{Prefixes.} We define prefixes somewhat differently in Rust to avoid
copying memory where possible. A prefix is a struct containing three elements:
\begin{enumerate}
  \item A list of lazy-removable \code{transitions} which make up the prefix.
        A boolean for each element indicates whether the corresponding
        transition has been lazily removed. A transition is either
        $\tout{p}\ell(S)$ or $\tin{p}\ell(S)$, which is identical to a prefix
        term in the theory.

  \item A \code{start} index, which indicates that the first \code{start}
        elements in \code{transitions} should be ignored as they have been
        lazily removed.

  \item A list of indexes of elements that have been lazily \code{removed} by
        setting their boolean to \code{true}.
\end{enumerate}

\noindent\begin{minipage}{\linewidth}
  \begin{lstlisting}[language=Rust, aboveskip=\baselineskip, belowskip=\baselineskip]
struct Prefix {
  transitions: Vec<(bool, Transition)>,
  start: usize,
  removed: Vec<usize>,
}
\end{lstlisting}
\end{minipage}
Elements can be lazily removed either by incrementing \code{start} or by setting
the element's boolean to \code{true} and adding its index to \code{removed}. We
favour the first option so as to maintain the invariant
\begin{equation*}
  \code{transitions.len() > 0} \implies \code{!transitions[0].0}
\end{equation*}
where the tuple indexing syntax \code{(x, y).0} will evaluate to \code{x}, the
first element of the tuple. To ensure that this invariant holds, we must advance
\code{start} as far as possible when removing a transition at the head of the
prefix.

We also give the option of storing snapshots to previous versions of a prefix. A
snapshot stores
\begin{enumerate*}[label=\textbf{(\arabic*)}]
  \item the \code{size} of the transitions list;
  \item the value of the \code{start} field; and
  \item the size of the \code{removed} list,
\end{enumerate*}
all taken at the time of the snapshot.

\noindent\begin{minipage}{\linewidth}
  \begin{lstlisting}[language=Rust, aboveskip=\baselineskip, belowskip=\baselineskip]
struct Snapshot {
  size: usize,
  start: usize,
  removed: usize,
}
\end{lstlisting}
\end{minipage}
We can easily revert a prefix to a previous snapshot by
\begin{enumerate*}[label=\textbf{(\arabic*)}]
  \item finding the elements of \code{removed} that have been added since the
  snapshot;
  \item setting the boolean to \code{false} for each of these elements to
  restore them;
  \item truncating \code{transitions} to its previous \code{size};
  \item restoring \code{start} to its previous value; and
  \item truncating \code{removed} to its previous size.
\end{enumerate*}

\paragraph{Visitor.} We use the visitor pattern \cite{Palsberg1998} to traverse
a pair of FSMs $\efsm$ and $\oefsm$. In our visitor, we store
\begin{enumerate*}[label=\textbf{(\arabic*)}]
  \item the \code{fsms} we are traversing;
  \item a matrix of \code{history} (as we will see, this is equivalent to the
  assumptions map $\Sigma$ in the theory); and
  \item a pair of \code{prefixes}, as in the theory.
\end{enumerate*}

\noindent\begin{minipage}{\linewidth}
  \begin{lstlisting}[language=Rust, aboveskip=\baselineskip, belowskip=\baselineskip]
struct SubtypeVisitor {
  fsms: Pair<Fsm>,
  history: Matrix<Previous>,
  prefixes: Pair<Prefix>,
}
\end{lstlisting}
\end{minipage}
The \code{history} matrix stores a value for each combination of states in
$\efsm$ and $\oefsm$ (it effectively has the type $\lvert \efsm \rvert \times
  \lvert \oefsm \rvert \reduce \code{Previous}$). Each of these values stores
a \code{Previous} struct containing the number of \code{visits} this combination
of states has remaining and optionally (if it has been visited before) a pair of
\code{snapshots} taken during the last visit to this combination.

\noindent\begin{minipage}{\linewidth}
  \begin{lstlisting}[language=Rust, aboveskip=\baselineskip, belowskip=\baselineskip]
struct Previous {
  visits: usize,
  snapshots: Option<Pair<Snapshot>>,
}
\end{lstlisting}
\end{minipage}
In the theory, termination is guaranteed by allowing recursions to be unrolled
only $n$ times. Here, our `$n$' is the value of \code{visits}, which limits how
many times the same combination of states can be visited. Since $\efsm$ and
$\oefsm$ each contain a finite number of states and their cross product is also
finite, this will achieve termination just as in the theory (provided that $n$
is also finite). Otherwise, this is identical to the theory---our \code{history}
matrix corresponds to the map of assumptions $\Sigma$ and the \code{Previous}
struct represents a single mapping (we use snapshots in place of prefixes).

Each state in an \code{Fsm} is given a unique \code{StateIndex} that identifies
it. Our \code{Visitor} is executed using its recursive \code{visit} method,
which takes a mutable reference to the \code{Visitor} and a \code{StateIndex}
for each \code{Fsm}.

\noindent\begin{minipage}{\linewidth}
  \begin{lstlisting}[language=Rust, aboveskip=\baselineskip, belowskip=\baselineskip]
impl Visitor {
  fn visit(&mut self, states: Pair<StateIndex>) -> bool {
    [...]
  }
}
\end{lstlisting}
\end{minipage}
This \code{visit} method performs our asynchronous subtyping algorithm as
follows.
\begin{enumerate}
  \item We look up the current combination states in our \code{history} to
        ensure \code{visits} is positive, as in \ruleAlg{\(\mu\)l} and
        \ruleAlg{\(\mu\)r}. If it is not then our bound has been exhausted and we
        return with \code{false}.

  \item We attempt to reduce the pair of prefixes, as in \ruleAlg{sub}. This
        reduction process follows precisely the same rules as in the theory,
        lazily removing transitions where appropriate.

  \item If the current combination of states has been visited before, we
        attempt to use our assumptions map to return \code{true}, as in
        \ruleAlg{asm}. The method we use to check the actions sets, explained
        below, differs slightly from the theory.

  \item If both FSMs are in a \emph{terminal} state and the prefixes are
        empty then we return \code{true}, as in \ruleAlg{end}.

  \item If both FSMs are in a \emph{non-terminal} state then we
        \begin{itemize}
          \item take a snapshot of the current prefixes;
          \item update the \code{history} matrix for the current combination
                of states, setting \code{visits} to \code{visits - 1} and
                \code{snapshots} to the snapshots we just took;
          \item for each pair of transitions we can take from the current
                combination of states we
                \begin{itemize}
                  \item add each transition in the pair to its
                        corresponding prefix;
                  \item recurse using the \code{visit} method, setting the
                        \code{states} argument to the pair of end states
                        corresponding to our transitions; and after the
                        recursive call returns
                  \item revert the changes made to the prefixes by using
                        the snapshot we took previously;
                \end{itemize}
          \item restore the current \code{history} matrix entry to its
                original value; and
          \item return a value depending on the results of the recursive
                calls and whether the current combination of states performs
                input or output actions, as described by the quantifiers in
                \ruleAlg{\{in,out\}-\{in,out\}}.
        \end{itemize}

  \item Otherwise, one of the FSMs has reached a terminal state but the other
        has not. In this case, there is no way to progress and we return
        \code{false}.
\end{enumerate}
By performing a depth-first search we can make changes to the \code{history} and
\code{prefixes} fields of our visitor and revert them later, using a snapshot
for each prefix. This method improves the efficiency of our algorithm
by avoiding copying memory. If we instead used a
breadth-first search, for instance, we would need to store a separate visitor
for each frontier of our search. This would require an expensive copy of the
\code{history} and \code{prefixes}.

\paragraph{Checking actions.} In the theory, our \ruleAlg{asm} rule compares two
sets of actions to ensure that it is safe to apply an assumption. Specifically,
it checks that the actions of the supposed supertype's prefix ($\pi'$) are a
subset of the actions performed by the subtype since the assumption was made
($\rho'$). In our algorithm, we can actually perform a far cheaper but
equivalent check thanks to our use of lazy removal. We need only to confirm that
\begin{equation}
  \label{eqn:ActionCheck}
	\begin{aligned}
		&\code{transitions[start..] == }\\
		&\code{transitions[..snapshot.size][snapshot.start..]}
	\end{aligned}
\end{equation}
for each prefix/snapshot combination. The syntax \code{x[i..]} evaluates to
\code{x} with the first \code{i} elements removed and \code{y[..j]} evaluates to
the first \code{j} elements of \code{y}. Surprisingly, this check is identical
in effect to the one performed in the theory due to two observations.
\begin{enumerate}
  \item Comparing the full list of transitions (which include labels and
        sorts) rather than only their actions is sound since the reduction
        rules do not allow sends or receives to or from the same participant
        to be reordered.

        We can easily prove this by contradiction. Suppose $\tin{p}\ell(S) \in
          \pi'$ and $\tin{p}\ell'(S') \in \rho'$ and we can apply
        \ruleAlg{asm}. Clearly, $\tin{p}\ell(S)$ has not been reduced by
        \ruleRef{in}, otherwise, it would not still be in $\pi'$.
        Therefore, at some point since the assumption was added to
        $\Sigma$, \ruleRefA{} must have been used to move $\tin{p}\ell(S)$
        before $\tin{p}\ell'(S')$. This is a contradiction because
        $\refa{p}$ cannot contain $\tin{p}\ell'(S')$ by definition so
        \ruleRefA{} cannot have been applied. A similar argument can be
        made for the output case.

  \item The version in the theory is intuitively checking whether there is an
        action that `hangs on' to the far left of $\pi'$ for multiple
        iterations of a recursive type without ever being reduced. If this is
        the case, then the action will not be matched by any of the actions in
        $\rho'$ (otherwise it would have been reduced) so $\rho' \not\supseteq
          \pi'$.

        In our implementation, if an action hangs on to the supertype's prefix
        then it will never be lazily removed. This means that the size of the
        prefix will grow on each iteration of the FSM since \code{start} is
        never advanced. Since
        \begin{equation*}
		\begin{aligned}
			&\code{transitions[start..].len() !=}\\
			&\code{transitions[..snapshot.size][snapshot.start..].len()}
		\end{aligned}
        \end{equation*}
        \cref{eqn:ActionCheck} is trivially false. Note that the full check in
        \cref{eqn:ActionCheck} must be performed, rather than only comparing
        the lengths, to ensure that the prefixes do actually match those of
        the assumption, as in \ruleAlg{asm}.
\end{enumerate}

\paragraph{Fail-early reductions.} Our practical implementation performs the
same reduction rules on prefixes as described in the theory. However, we add a
practical optimisation to, in some cases, determine that a particular path
cannot succeed before even reaching the bound.

For example, consider the pair
$\pair{\tin{p}\ell(S).\pi}{\tout{q}\ell'(S').\tin{p}\ell(S).\pi'}$. Regardless
of what $\pi$ and $\pi'$ are set to, this pair cannot be reduced as it will
require using the \ruleRefA{} but $\tout{q}\ell'(S')$ cannot be contained in
$\refa{p}$. Therefore, if at some point we reach a pair of prefixes which looks
like $\pair{\tin{p}\ell(S).\tin{p}\ell(S)}{\tout{q}\ell'(S').\tin{p}\ell(S)}$,
we can immediately return \code{false} as there is no way that it can ever be
reduced by adding more terms.

\section{Benchmarking results}

\subsection{Session-Based Rust Implementations}
\label{sec:RuntimeData}
\paragraph{Results for the stream benchmark.}
\begin{center}
	\sffamily\footnotesize
	\begin{tabular}{cccccc}
		\toprule
		& \multicolumn{5}{c}{Throughput ($n$/$\mu$s)}                                                            \\ \cmidrule(l){2-6}
		$n$ & \sesh                                       & \multicrusty & \ferrite & \rumpsteak & \rumpsteak (opt.) \\
		\midrule
		10  & 0.019389                                    & 0.011678     & 0.011386 & 0.202587   & 0.215583          \\
		20  & 0.028142                                    & 0.014325     & 0.012994 & 0.336988   & 0.356978          \\
		30  & 0.034193                                    & 0.015160     & 0.013463 & 0.427489   & 0.437795          \\
		40  & 0.036566                                    & 0.016072     & 0.013671 & 0.488886   & 0.517468          \\
		50  & 0.040315                                    & 0.016577     & 0.014126 & 0.545378   & 0.583366          \\
		\bottomrule
	\end{tabular}
\end{center}

\paragraph{Results for the double buffering benchmark.}
\begin{center}
	\sffamily\footnotesize
	\begin{tabular}{cccccc}
		\toprule
		& \multicolumn{5}{c}{Throughput ($n$/$\mu$s)}                                                             \\ \cmidrule(l){2-6}
		$n$   & \sesh                                       & \multicrusty & \ferrite  & \rumpsteak & \rumpsteak (opt.) \\
		\midrule
		5000  & 6.929567                                    & 5.675414     & 7.617643  & 27.704354  & 32.340989         \\
		10000 & 13.138401                                   & 11.254181    & 14.649028 & 44.154722  & 50.126532         \\
		15000 & 18.739983                                   & 16.187341    & 20.429845 & 56.813002  & 67.884430         \\
		20000 & 24.103215                                   & 20.481378    & 25.506427 & 67.595301  & 82.039366         \\
		25000 & 28.609966                                   & 25.050058    & 29.629025 & 75.848611  & 96.010424         \\
		\bottomrule
	\end{tabular}
\end{center}

\paragraph{Results for the FFT benchmark.}
\begin{center}
	\sffamily\footnotesize
	\begin{tabular}{cccccc}
		\toprule
		& \multicolumn{5}{c}{Throughput ($n$/$\mu$s)}                                                   \\ \cmidrule(l){2-6}
		$n$  & \sesh                                       & \multicrusty & \ferrite & \rustfft & \rumpsteak \\
		\midrule
		1000 & 0.551154                                    & 0.810134     & 1.458279 & 9.320778 & 5.038554   \\
		2000 & 1.050958                                    & 1.515538     & 2.513855 & 9.313359 & 7.206404   \\
		3000 & 1.510567                                    & 2.163629     & 3.496405 & 9.333569 & 8.421026   \\
		4000 & 1.935263                                    & 2.783617     & 4.198723 & 9.336939 & 9.262763   \\
		5000 & 2.303627                                    & 3.261020     & 4.811375 & 9.323199 & 9.316716   \\
		\bottomrule
	\end{tabular}
\end{center}

\pagebreak
\subsection{Verifying Asynchronous Message Reordering}
\label{sec:VerificationData}
\paragraph{Results for the stream benchmark.}
\begin{center}
	\sffamily\footnotesize
	\begin{tabular}{cccc}
		\toprule
		& \multicolumn{3}{c}{Running time (s)}                         \\ \cmidrule(l){2-4}
		$n$ & \concur                              & \kmc     & \rumpsteak \\
		\midrule
		0   & 0.003476                             & 0.005504 & 0.001872   \\
		10  & 0.008556                             & 0.019316 & 0.001899   \\
		20  & 0.020673                             & 0.057417 & 0.001848   \\
		30  & 0.041673                             & 0.142145 & 0.001906   \\
		40  & 0.076425                             & 0.276446 & 0.001874   \\
		50  & 0.127865                             & 0.496929 & 0.002080   \\
		60  & 0.198541                             & 0.805577 & 0.002083   \\
		70  & 0.292471                             & 1.233327 & 0.002064   \\
		80  & 0.422571                             & 1.780778 & 0.002178   \\
		90  & 0.583863                             & 2.475443 & 0.002190   \\
		100 & 0.767426                             & 3.349204 & 0.002249   \\
		\bottomrule
	\end{tabular}
\end{center}

\paragraph{Results for the nested choice benchmark.}

\begin{center}
	\sffamily\footnotesize
	\begin{tabular}{cccc}
		\toprule
		& \multicolumn{3}{c}{Running time (s)}                          \\ \cmidrule(l){2-4}
		$n$ & \concur                              & \kmc      & \rumpsteak \\
		\midrule
		1   & 0.002295                             & 0.006554  & 0.000702   \\
		2   & 0.004504                             & 0.014901  & 0.000755   \\
		3   & 0.016347                             & 0.072423  & 0.001745   \\
		4   & 0.224858                             & 1.515528  & 0.007656   \\
		5   & 4.692525                             & 41.688068 & 0.157548   \\
		\bottomrule
	\end{tabular}
\end{center}

\paragraph{Results for the ring benchmark.}

\begin{center}
	\sffamily\footnotesize
	\begin{tabular}{ccc}
		\toprule
		& \multicolumn{2}{c}{Running time (s)}              \\ \cmidrule(l){2-3}
		$n$ & \kmc                                 & \rumpsteak \\
		\midrule
		2   & 0.004007                             & 0.000675   \\
		4   & 0.007239                             & 0.000731   \\
		6   & 0.011806                             & 0.000701   \\
		8   & 0.018822                             & 0.000835   \\
		10  & 0.024842                             & 0.000757   \\
		12  & 0.049232                             & 0.000777   \\
		14  & 0.102257                             & 0.000744   \\
		16  & 0.191078                             & 0.000813   \\
		18  & 0.340262                             & 0.000817   \\
		20  & 0.570656                             & 0.000766   \\
		22  & 0.913412                             & 0.000911   \\
		24  & 1.391075                             & 0.000737   \\
		26  & 2.042452                             & 0.000752   \\
		28  & 2.918943                             & 0.000732   \\
		30  & 4.099072                             & 0.000769   \\
		\bottomrule
	\end{tabular}
\end{center}
\pagebreak

\paragraph{Results for $k$-buffering benchmark.}
\begin{center}
	\sffamily\footnotesize
	\begin{tabular}{ccc}
		\toprule
		& \multicolumn{2}{c}{Running time (s)}              \\ \cmidrule(l){2-3}
		$n$ & \kmc                                 & \rumpsteak \\
		\midrule
		0   & 0.004825                             & 0.000630   \\
		5   & 0.007668                             & 0.000747   \\
		10  & 0.013613                             & 0.000705   \\
		15  & 0.018770                             & 0.000667   \\
		20  & 0.031376                             & 0.000825   \\
		25  & 0.054910                             & 0.000718   \\
		30  & 0.080879                             & 0.000760   \\
		35  & 0.122315                             & 0.000853   \\
		40  & 0.170533                             & 0.000802   \\
		45  & 0.236354                             & 0.000792   \\
		50  & 0.305749                             & 0.000916   \\
		55  & 0.406071                             & 0.000882   \\
		60  & 0.506069                             & 0.000959   \\
		65  & 0.639521                             & 0.001028   \\
		70  & 0.773931                             & 0.001057   \\
		75  & 0.954399                             & 0.001045   \\
		80  & 1.127240                             & 0.001125   \\
		85  & 1.359600                             & 0.001120   \\
		90  & 1.571745                             & 0.001164   \\
		95  & 1.869339                             & 0.001156   \\
		100 & 2.111687                             & 0.001234   \\
		\bottomrule
	\end{tabular}
\end{center}

}{
}

\bibliographystyle{ACM-Reference-Format}
\bibliography{main}

\end{document}